\documentclass[]{sn-jnl} 
\jyear{2021}

\theoremstyle{thmstyletwo}%

\theoremstyle{thmstylethree}%

\usepackage{thm-restate}
\usepackage{graphicx}
\usepackage{amssymb}
\usepackage{amsmath}
\usepackage[capitalise]{cleveref}
\usepackage{xspace}
\usepackage{latexsym}
\usepackage{textcomp}
\usepackage{epsfig}
\usepackage{enumitem}
\usepackage{times}
\usepackage{subcaption}
\usepackage{floatflt}
\newcommand{\indeg}{\mathrm{indeg}\xspace}
\newcommand{\outdeg}{\mathrm{outdeg}\xspace}

\newcommand{\subs}{\mathrm{Sub}\xspace}
\newcommand{\ov}{\overline}
\newcommand{\Pio}[2]{\textcolor{blue}{$\mathrm{I}_{#1}\mathrm{O}_{#2}$}}
\newcommand{\Pin}[1]{\textcolor{blue}{$\mathrm{I}_{#1}$}}

\newtheorem{theor}{Theorem}
\newtheorem{claim}{Claim}
\newtheorem{lemma}{Lemma}
\newtheorem{clm}{Claim}

\newtheorem{proposition}{Proposition}

\usepackage{multirow}
\usepackage{multicol}
\usepackage[capitalise]{cleveref}
\usepackage{color, colortbl}
\usepackage{booktabs}

\usepackage{geometry}
\definecolor{lightcyan}{rgb}{0.88,1,1}
\definecolor{antiquewhite}{rgb}{0.98, 0.92, 0.84}

\newenvironment{claimproof}{\noindent{\itshape Claim Proof.}}{~\hfill
	$\blacksquare$\smallskip}

\begin{document}

\title[Computing Bend-Minimum Orthogonal Drawings of Plane Series-Parallel Graphs in Linear Time]
{Computing Bend-Minimum Orthogonal Drawings of Plane Series-Parallel Graphs in Linear Time\footnote{A preliminary short version of this research, restricted to rectilinear planarity testing, appears in the proceedings of the International Symposium on Graph Drawing and Network Visalization, GD~2020~\cite{DBLP:conf/gd/Didimo0LO20}. The current journal version significantly extends the preliminary GD version, by providing: Complete technical details and full proofs for all lemmas/theorems, new results about the concept of spirality (\cref{se:substitution}), a new liner-time bend-minimization algorithm (\cref{se:bend-min-overv,se:budgets,se:summary}), and a revised set of open problems (\cref{se:open}). Research partially supported by MIUR Project ``AHeAD'' under PRIN 20174LF3T8.}}

\author*[1]{\fnm{Walter} \sur{Didimo}}\email{walter.didimo@unipg.it}	
\equalcont{These authors contributed equally to this work.}
\author[2]{\fnm{Michael} \sur{Kaufmann}}\email{mk@informatik.uni-tuebingen.de}
\equalcont{These authors contributed equally to this work.}
\author[1]{\fnm{Giuseppe} \sur{Liotta}}\email{giuseppe.liotta@unipg.it}
\equalcont{These authors contributed equally to this work.}
\author[1]{\fnm{Giacomo} \sur{Ortali}}\email{giacomo.ortali@studenti.unipg.it}
\equalcont{These authors contributed equally to this work.}

\affil*[1]{\orgdiv{Department of Engineering}, \orgname{University of Perugia}, \orgaddress{\city{Perugia}, \country{Italy}}}

\affil*[2]{\orgname{University of T\"ubingen}, \orgaddress{\city{T\"ubingen}, \country{Germany}}}


\abstract{
A \emph{planar orthogonal drawing} of a planar 4-graph $G$ (i.e., a planar graph with vertex-degree at most four) is a crossing-free drawing that maps each vertex of $G$ to a distinct point of the plane and each edge of $G$ to a sequence of horizontal and vertical segments between its end-points.  
A longstanding open question in Graph Drawing, dating back over 30 years, is whether there exists a linear-time algorithm to compute an orthogonal drawing of a \emph{plane} 4-graph with the minimum number of bends. The term ``plane'' indicates that the input graph comes together with a planar embedding, which must be preserved by the drawing (i.e., the drawing must have the same set of faces as the input graph). In this paper we positively answer the question above for the widely-studied class of series-parallel graphs. Our linear-time algorithm is based on a characterization of the planar series-parallel graphs that admit an orthogonal drawing without bends. This characterization is given in terms of the orthogonal spirality that each type of triconnected component of the graph can take; the orthogonal spirality of a component measures how much that component is ``rolled-up'' in an orthogonal drawing of the graph.
}


\maketitle

\section{Introduction}
%
Given a planar 4-graph $G$ (i.e., a planar graph with vertex-degree at most four), a \emph{planar orthogonal drawing} of $G$ is a crossing-free drawing that maps each vertex of $G$ to a distinct point of the plane and each edge of $G$ to a sequence of horizontal and vertical segments between its end-points~\cite{DBLP:books/ph/BattistaETT99,DBLP:reference/crc/DuncanG13,DBLP:conf/dagstuhl/1999dg,DBLP:books/ws/NishizekiR04}. A \emph{bend} is a point where a vertical and a horizontal segment of the same edge meet; see \cref{fi:graph-minbend,fi:ortho-minbend}, where bends are depicted as~`$\times$'. Computing planar orthogonal drawings of planar graphs with the minimum number bends is one of the most studied problems in Graph Drawing. Garg and Tamassia~\cite{DBLP:journals/siamcomp/GargT01} proved that this problem is NP-hard for general planar 4-graphs if the algorithm can freely choose the planar embedding; an $O(n)$-time algorithm exists for $n$-vertex planar 3-graphs~\cite{DBLP:conf/soda/DidimoLOP20} and an $O(n^3 \log^2 n)$-time algorithm exists for $n$-vertex series-parallel 4-graphs~\cite{DBLP:conf/gd/GiacomoLM19}.

If $G$ is a \emph{plane} 4-graph, i.e., it has a planar embedding that the drawing algorithm must preserve, the problem is polynomial-time solvable. A seminal paper by Tamassia~\cite{DBLP:journals/siamcomp/Tamassia87} describes an $O(n^2 \log n)$-time algorithm based on an elegant min-cost flow-network model. Cornelsen and Karrenbauer~\cite{DBLP:journals/jgaa/CornelsenK12} reduced the time complexity to $O(n^{1.5})$, through a more efficient min-cost flow-network technique. Deciding whether there exists an (optimal) $O(n)$-time algorithm is a longstanding question that is still unanswered (see, e.g.,~\cite{DBLP:conf/gd/BrandenburgEGKLM03,DBLP:books/ph/BattistaETT99,dlt-gd-17}).

\begin{figure}[tb]
	\centering
	\begin{subfigure}{.27\columnwidth}
		\centering
		\includegraphics[width=\columnwidth,page=2]{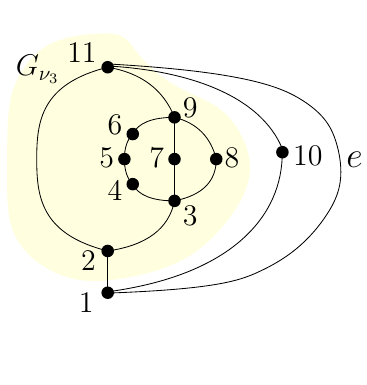}
		\subcaption{\centering $G$}
		\label{fi:graph-minbend}
	\end{subfigure}
	\hfil
	\begin{subfigure}{.27\columnwidth}
		\centering
		\includegraphics[width=\columnwidth,page=4]{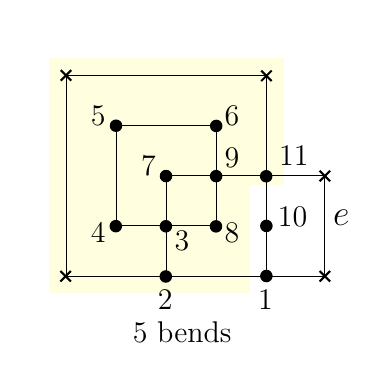}
		\subcaption{\centering $H$}
		\label{fi:ortho-minbend}
	\end{subfigure}
	\hfil
	\begin{subfigure}{.27\columnwidth}
		\centering
		\includegraphics[width=\columnwidth,page=3]{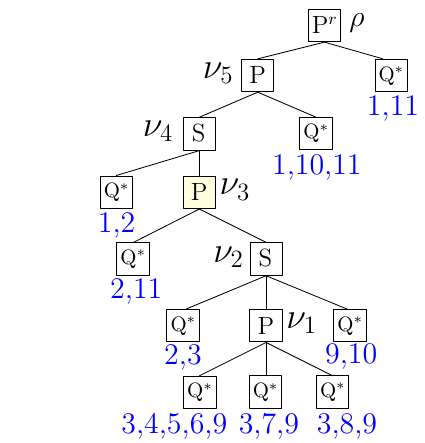}
		\subcaption{\centering $T$}
		\label{fi:spq-tree-minbend}
	\end{subfigure}
	\hfil
	\begin{subfigure}{0.93\columnwidth}
		\centering
		\includegraphics[width=0.95\columnwidth,page=5]{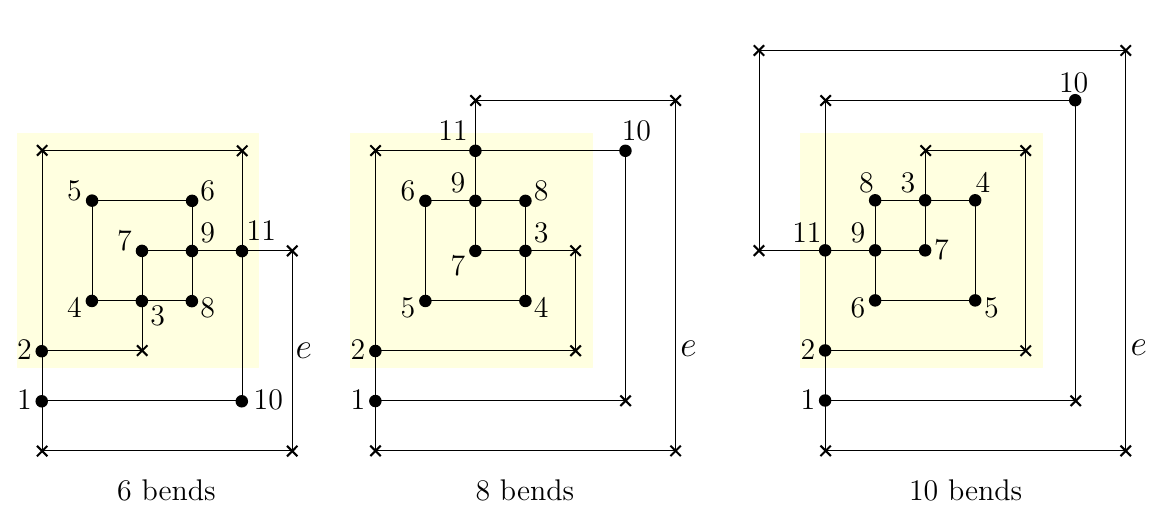}
		\subcaption{\centering}
		\label{fi:suboptimal-minbend}
	\end{subfigure}
	\caption{(a) A series-parallel graph $G$ with a given planar embedding. (b) A bend-minimum orthogonal drawing $H$ of $G$ with 5 bends, represented by cross vertices. (c) The SPQ$^*$-tree $T$ of $G$ with reference edge $e=(1,11)$; series and parallel compositions are represented by S- and P-nodes; a Q$^*$-node represents an edge or a series of edges; the root P$^r$-node is a parallel composition between $e$ and the rest of the graph. (d) Three suboptimal orthogonal drawings obtained by distributing the same number of bends of the highlighted subgraph $G_{\nu_3}$ in a different way.}\label{fi:intro}
\end{figure}

\smallskip\noindent{\bf Contribution.} In this paper we positively answer the question above for series-parallel graphs, which are a classical subject of investigation in Graph Drawing and Graph Algorithms (see, e.g., ~\cite{DBLP:books/ph/BattistaETT99,DBLP:journals/jacm/TakamizawaNS82,DBLP:journals/siamcomp/ValdesTL82,DBLP:journals/siamdm/ZhouN08}). Namely, we give an $O(n)$-time algorithm that receives as input an $n$-vertex plane series-parallel 4-graph $G$ and that computes an embedding-preserving orthogonal drawing of $G$ with the minimum number of bends. 
While $O(n)$-time algorithms that compute bend-minimum orthogonal drawings of plane graphs with maximum vertex-degree three are known (see, e.g.,~\cite{DBLP:books/ws/NishizekiR04,DBLP:journals/jgaa/RahmanNN99,DBLP:conf/wg/RahmanN02}), our result is the first linear-time algorithm for a graph family with vertices of degree four. Indeed, even for series-parallel plane graphs with degree-4 vertices, the most efficient solution known to date is the $O(n^{1.5})$-time algorithm by Cornelsen and Karrenbauer~\cite{DBLP:journals/jgaa/CornelsenK12}.

Different from the approach of Cornelsen and Karrenbauer we do not use network-flow techniques to minimize the number of bends. Instead, we rely on the observation that the bend-minimization problem for an orthogonal drawing of a plane graph $G$ is equivalent to inserting in $G$ the minimum number of subdivision vertices that make it rectilinear planar, i.e., drawable without bends. Following this idea, we first characterize those series-parallel graphs that are rectilinear planar. The characterization is expressed in terms of the ``orthogonal spirality'' for the triconnected components of $G$. 
Informally speaking, the orthogonal spirality of a component in an orthogonal drawing of $G$ measures how much the component is ``rolled-up'' (see, e.g.,~\cite{DBLP:journals/siamcomp/BattistaLV98,DBLP:conf/isaac/DidimoL98,DBLP:conf/soda/DidimoLOP20,DBLP:journals/jcss/DidimoLP19,DBLP:conf/gd/GargL99}).
We then consider the problem of efficiently adding the minimum number of subdivision vertices along the edges of those series-parallel graphs that are not rectilinear planar.

In a nutshell, our bend-minimization algorithm executes a post-order visit of a series-parallel \emph{decomposition tree} $T$ of the input graph $G$, also called \emph{SPQ$^*$-tree}. Tree $T$ represents the parallel and the series compositions that form~$G$ (see \cref{fi:graph-minbend,fi:spq-tree-minbend}, and refer to \cref{se:preli} for a formal definition of SPQ$^*$-trees). Suppose that~$\nu$ is a node of~$T$ and~$G_\nu$ its corresponding subgraph in $G$. When the algorithm visits $\nu$, it must efficiently determine whether $G_\nu$ is rectilinear planar and, if not, 
it must compute the minimum number of subdivision vertices needed to make it rectilinear planar; how to efficiently compute such a number is a first key ingredient of our approach. 

A second key ingredient is proving that, when the algorithm processes a node $\nu$ in the bottom-up visit, the addition of the minimum number of subdivision vertices that make $G_\nu$ rectilinear planar leads to the optimum in terms of total number of bends.

As a third key ingredient, the algorithm needs to concisely describe the set of rectilinear drawings of $G_\nu$ that can be obtained by distributing these subdivision vertices in all possible ways along the edges of~$G_\nu$, which gives rise to a combinatorial explosion of different possibilities. Indeed, different distributions of the same set of subdivision vertices along the edges of~$G_\nu$ can lead to orthogonal drawings of $G$ that have different number of bends. 
For example, consider the highlighted subgraph $G_{\nu_3}$ (associated with node $\nu_3$ of~$T$) in the graph of \cref{fi:graph-minbend}. Any orthogonal drawing of $G_{\nu_3}$ requires at least three bends
(subdivision vertices); placing all of them on edge $(2,11)$ yields the optimal solution of \cref{fi:ortho-minbend}, which additionally requires two bends on edge $(1,11)$. Conversely, placing the three bends on a different subset of edges of~$G_{\nu_3}$ leads to (suboptimal) solutions with more bends; see \cref{fi:suboptimal-minbend}.
To efficiently handle the combinatorially many distributions of the subdivision vertices along the edges of a subgraph $G_\nu$, we succinctly encode in $O(1)$ space the ``orthogonal shapes'' that $G_\nu$ can have in a bend-minimum planar orthogonal drawing of~$G$. This is done by looking at the set of possible orthogonal spiralities that~$G_\nu$ can take in such a drawing.

\smallskip
The remainder of the paper is organized as follows. \cref{se:preli} recalls basic definitions used throughout the paper. \cref{se:substitution} strengthen a result given in~\cite{DBLP:journals/siamcomp/BattistaLV98} about the interchangeability of orthogonal representations with the same spirality. \cref{se:rect} characterizes those plane series-parallel graphs that are rectilinear planar. \cref{se:bend-min-overv} gives an overview of our bend-minimization algorithm. \cref{se:budgets} provides details about the bottom-up visit performed by the algorithm. \cref{se:summary} summarizes our main result. 
Concluding remarks and open problems are in \cref{se:open}.
%

\section{Preliminaries}\label{se:preli}
We assume familiarity with basic concepts of graph planarity and graph drawing ~\cite{DBLP:books/ph/BattistaETT99,DBLP:conf/dagstuhl/1999dg,nc-08,DBLP:books/ws/NishizekiR04}.
We focus on \emph{orthogonal representations} rather than orthogonal drawings. An orthogonal representation~$H$ of a planar graph~$G$ describes a class of equivalent planar orthogonal drawings of $G$ in terms of planar embedding, ordered sequence of bends along the edges (i.e., sequence of left/right turns going from an end-vertex to the other) and clockwise sequence of geometric angles at each vertex, each angle formed by two (possibly coincident) consecutive edges around the vertex and expressed as a value in the set $\{90^\circ,180^\circ,270^\circ,360^\circ\}$ (angles of $360^\circ$ occur only at degree-1 vertices).
An orthogonal representation $H$ of $G$ can be described by a planar embedding of $G$ plus an \emph{angle labeling} specifying: $(i)$ For each vertex $v$ of~$G$, the geometric angles at $v$; $(ii)$ for each edge $e=(u,v)$ of~$G$, the ordered sequence of bends along $e$ as a sequence of angles in the left face (and hence in the right face) of $e$ while moving along $e$ from $u$ to $v$ (each bend determines an angle of $90^\circ$ in one of the two faces incident to $e$ and an angle of $270^\circ$ in the other face). It is well known (see, e.g.,~\cite{dett-gd-99}) that an angle labeling of $G$ describes a valid orthogonal representation if and only if the following properties hold: \textsf{(H1)} for each vertex $v$, the sum of the angles at $v$ equals $360^\circ$; \textsf{(H2)} for each face $f$, if $N_a(f)$ is the number of $a^\circ$ angles in $f$, we have $N_{90}(f)-N_{270}(f)-2N_{360}(f)=4$ (resp. $N_{90}(f)-N_{270}(f)-2N_{360(f)}=-4$) if $f$ is an internal face (resp. the external face). 

Note that, since the set of faces of a plane graph $G$ is a base for the set of simple cycles of $G$, Properties~\textsf{(H1)} and \textsf{(H2)} together are equivalent to say that for every simple cycle $C$ of $G$, the number of right turns minus the number of left turns, walking clockwise on the boundary of $C$, is equal to four. Namely, if $v$ is a vertex of $C$ we count: a right turn at $v$ if there is an angle of $90^\circ$ at $v$ inside $C$; a left turn at $v$ if the sum of the angles at $v$ inside $C$ equals $270^\circ$; two left turns at $v$ if $v$ has degree one. Also, a bend on an edge of $C$ corresponds to a right (resp. left) turn if it determines an angle of $90^\circ$ (resp., $270^\circ$) in $C$.     

Given an orthogonal representation $H$ of a plane graph $G$, a drawing of $H$ (which corresponds to an orthogonal drawing of $G$) can be computed in linear time~\cite{DBLP:journals/siamcomp/Tamassia87}. If $H$ has no bend, $H$ is a \emph{rectilinear representation}. 

\medskip
\smallskip\noindent{\bf Series-Parallel graphs and Decomposition Trees.} A \textit{two-terminal series-parallel} graph~$G$, also called \emph{series-parallel graph}, has two distinct vertices $s$ and $t$, called the \textit{source} and the \textit{sink} of $G$, respectively. A series-parallel graph can be inductively defined by, and naturally associated with, a \emph{decomposition tree} $T$: $(i)$~A single edge $(s,t)$ is a series-parallel graph with source~$s$~and~sink~$t$; in this case $T$ consists of a single Q-node, whose \emph{poles} are~$s$ and~$t$.
$(ii)$~Given $p \geq 2$ series-parallel graphs $G_1, \dots, G_p$, each $G_i$ with source $s_i$ and sink $t_i$ ($i=1, \dots, p$), a new series-parallel graph $G$ can be obtained with any of these two operations:

\smallskip\noindent {\em- Series composition}, which identifies $t_i$ with $s_{i+1}$ ($i=1,\dots,p-1$); $G$ has source $s=s_1$ and sink $t=t_p$. The composition is represented in $T$ by an S-node, with \emph{poles} $s_1$ and $t_p$, whose children are the roots of the decomposition trees $T_i$ of $G_i$ ($i=1, \dots, p$).

\smallskip\noindent {\em- Parallel composition}, which identifies all sources $s_i$ (resp. all sinks $t_i$) together~($i=1, \dots, p$); $G$ has source~$s=s_i$~and sink~$t=t_i$. The composition is represented in $T$ by a P-node, with \emph{poles} are $s$ and $t$, whose children are the roots of the decomposition trees $T_i$ of~$G_i$~($i=1, \dots, p$).

\smallskip In our algorithm we do not distinguish between Q-nodes and S-nodes whose children are all Q-nodes. We just call any of these nodes a Q$^*$-node. In other words, a Q$^*$-node represents a series of edges.
For a node $\nu$ of $T$, the \emph{pertinent graph} $G_\nu$ of $\nu$ is the subgraph of $G$ formed by all edges associated with the Q$^*$-nodes in the subtree rooted at~$\nu$. We also call $G_\nu$ a \emph{component} of $G$.

Let $G$ be a plane (two-terminal) series-parallel graph with vertex-degree at most four. Note that $G$ is either biconnected or it can be made biconnected with the addition of a single dummy edge; in this latter case we assume that the planar embedding of $G$ is such that the dummy edge can be added on the external face of $G$.
For any edge $e=(s,t)$ (possibly a dummy edge) on the external face, we can associate with $G$ a decomposition tree $T$ where the root is a P-node representing the parallel composition between $e$ and the rest of the graph. Thus, the root of $T$ is always a P-node with two children, one of which is a Q$^*$-node corresponding to $e$. It will be called the (unique) P$^r$-node of $T$, to distinguish it by the other P-nodes. Edge $e$ is the \emph{reference edge} of $T$, and $T$ is the \emph{SPQ$^*$-tree} of $G$ \emph{with respect to} $e$. Without loss of generality we assume that the external face of $G$ is to the right of $e$ while moving from $s$ to $t$. Also, it is always possible to make $T$ such that each (non-root) P-node has no P-node child and each S-node has no S-node child. Since $G$ has vertex-degree at most four, a P-node has either two or three children. Finally, we assume that the left-to-right order of the children of a P-node reflects the left-to-right order that their corresponding components have in the planar embedding of $G$. See \cref{fi:graph-minbend} and \cref{fi:spq-tree-minbend}. From now on we assume that $T$ satisfies the properties above for an $n$-vertex biconnected series-parallel graph. Observe that the number of nodes of $T$ is $O(n)$.

\smallskip\noindent{\bf Spirality of Series-Parallel Graphs.} Let $G$ be a biconnected plane series-parallel graph and let $T$ be an SPQ$^*$-tree with respect to a reference edge $e=(s,t)$. Let $H$ be an (embedding-preserving) orthogonal representation of~$G$. Also, let $\nu$ be a node of $T$ with poles $\{u,v\}$ and let $H_\nu$ be the restriction of $H$ to $G_\nu$. We also say that $H_\nu$ is a \emph{component} of $H$. For each pole $w \in \{u,v\}$, let  $\indeg_\nu(w)$ and $\outdeg_\nu(w)$ be the degree of $w$ inside and outside $H_\nu$, respectively. Define two (possibly coincident) \emph{alias vertices} of $w$, denoted by $w'$ and $w''$, as follows:
$(i)$ if $\indeg_\nu(w)=1$, then $w'=w''=w$;
$(ii)$ if $\indeg_\nu(w)=\outdeg_\nu(w)=2$, then $w'$ and $w''$ are dummy vertices, each splitting one of the two distinct edge segments incident to~$w$ outside~$H_\nu$;
$(iii)$ if $\indeg_\nu(w)>1$ and $\outdeg_\nu(w)=1$, then $w'=w''$ is a dummy vertex that splits the edge segment incident to $w$ outside $H_\nu$.

\begin{figure}[tb]
	\centering
	\includegraphics[width=1\columnwidth,page=1]{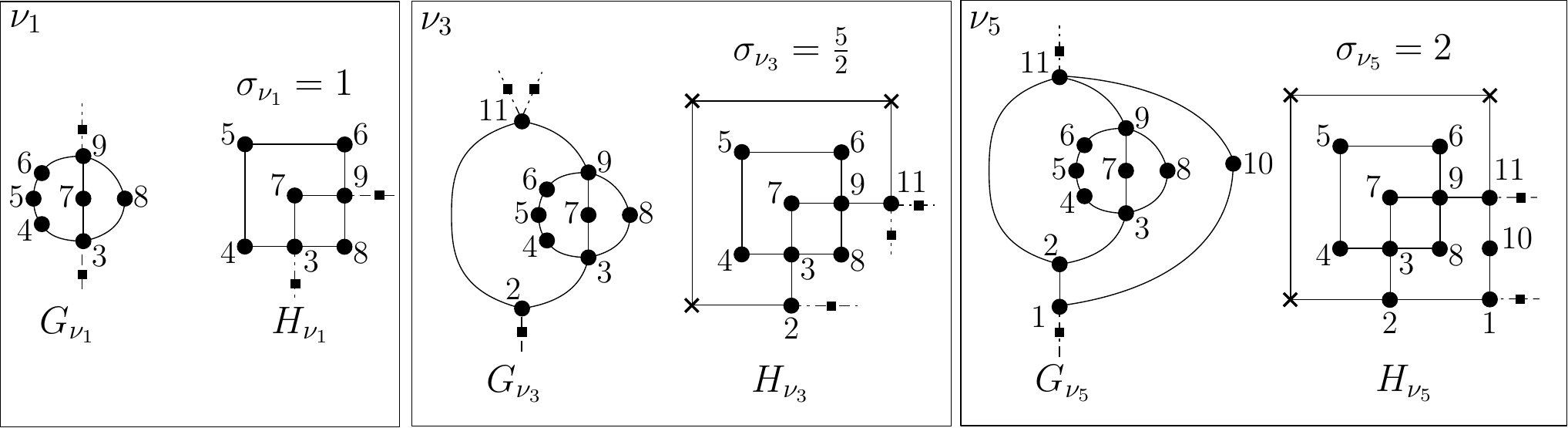}
	\caption{The components associated with the P-nodes $\nu_1$, $\nu_3$, and $\nu_5$, of the graph in \cref{fi:intro}. The alias vertices are the little squares along dashed edges. For each node $\nu_i$, $i\in\{1,3,5\}$, we report $G_{\nu_i}$, $H_{\nu_i}$, and the spirality~$\sigma_{\nu_i}$~of~$H_{\nu_i}$~in~$H$. In particular, the P-component of $\nu_3$, with poles $\{u=2, v=11\}$, has spirality $\frac{5}{2}$; the P-component of $\nu_{5}$, with poles $\{u=1, v=11\}$, has spirality  $2$.   
 }\label{fi:spiralities-minbend}
\end{figure}

Let $A^w$ be the set of distinct alias vertices of a pole $w$. Let $P^{uv}$ be any simple path from $u$ to $v$ inside $H_\nu$ and let $u'$ and $v'$ be the alias vertices of $u$ and of $v$, respectively. The path $S^{u'v'}$ obtained concatenating $(u',u)$, $P^{uv}$, and $(v,v')$ is called a \emph{spine} of $H_\nu$. Denote by $n(S^{u'v'})$ the number of right turns minus the number of left turns encountered along $S^{u'v'}$ while moving from $u'$ to $v'$.
%
The \emph{spirality} $\sigma(H_\nu)$ of $H_\nu$ is introduced by Di Battista et al.~\cite{DBLP:journals/siamcomp/BattistaLV98} and it is defined based on the following cases (see also \cref{fi:spiralities-minbend} for the spiralities of some P-components in the representation~$H$ of \cref{fi:ortho-minbend}.):
\begin{itemize}
\item $A^u=\{u'\}$ and $A^v=\{v'\}$; then $\sigma(H_\nu) = n(S^{u'v'})$.
\item $A^u=\{u'\}$ and $A^v=\{v',v''\}$; then $\sigma(H_\nu) = \frac{n(S^{u'v'}) + n(S^{u'v''})}{2}$.
\item $A^u=\{u',u''\}$ and $A^v=\{v'\}$; then $\sigma(H_\nu) = \frac{n(S^{u'v'}) + n(S^{u''v'})}{2}$.
\item $A^u=\{u',u''\}$ and $A^v=\{v',v''\}$; without loss of generality, assume that $(u,u')$ precedes $(u,u'')$ counterclockwise around $u$ and that $(v,v')$ precedes $(v,v'')$ clockwise around $v$; then $\sigma(H_\nu) = \frac{n(S^{u'v'}) + n(S^{u''v''})}{2}$.
\end{itemize}

\medskip Di Battista et al.~\cite{DBLP:journals/siamcomp/BattistaLV98} show that the spirality of~$H_\nu$ does not vary with the choice of the path~$P^{uv}$. For brevity, in the following we often denote by $\sigma_\nu$ the spirality of an orthogonal representation~$H_\nu$ of~$G_\nu$. 
%
If $\nu$ is a Q$^*$-node or a P-node with three children, $\sigma_\nu$ is always an integer. If $\nu$ is an S-node or a P-node with two children, $\sigma_\nu$ is either integer or semi-integer depending on whether the total number of alias vertices for the poles of $\nu$ is even or odd. When we say that the spirality $\sigma_\nu$ can take \emph{all} values in an interval $[a,b]$, we mean that such values are either all the integer numbers or all the semi-integer numbers in $[a,b]$, depending on the cases described above for $\nu$.

\section{Substituting Orthogonal Components with the Same Spirality}\label{se:substitution}
Let $G$ be a biconnected plane series-parallel graph and let $T$ be an SPQ$^*$-tree of $G$ with respect to a given reference edge. Di Battista et al.~\cite{DBLP:journals/siamcomp/BattistaLV98} prove that two distinct orthogonal representations of the same component $G_\nu$ that have the same spirality are ``interchangeable'', under some additional hypotheses. Roughly speaking, they prove that if a component $H_\nu$ has a certain spirality in a given orthogonal representation $H$, it can be substituted with another representation $H'_\nu$ having the same spirality, under the assumption that the angles at the poles of $\nu$ that are outside $G_\nu$ do not change. In this subsection we formalize the concept of substituting an orthogonal component with another one and give a stronger version of the result in~\cite{DBLP:journals/siamcomp/BattistaLV98}, which proves the interchangeability of two orthogonal components that have the same spirality, regardless of their angles at the poles.    

Let $H$ and $H'$ be two different orthogonal representations of $G$ with the reference edge on the external face, and let $H_\nu$ and $H'_\nu$ be the restrictions of $H$ and $H'$ to $G_\nu$, respectively. If $\sigma(H_\nu)=\sigma(H'_\nu)$, the operation of \emph{substituting} $H_\nu$ with $H'_\nu$ in $H$, denoted by $\subs(H_\nu, H'_\nu)$, defines a new plane graph $H''$ with an angle labeling, such that: $(a)$ $H''$ corresponds to a valid orthogonal representation of $G$; $(b)$  the restriction of $H''$ to $G_\nu$ coincides with $H'_\nu$; $(c)$ the restriction of $H''$ to $G \setminus G_\nu$ stays as in $H$. 

More formally, let $u$ and $v$ be the two poles of~$\nu$. The external boundary of $H_\nu$ contains a \emph{left path} $p_l$ and a \emph{right path} $p_r$, such that $p_l$ goes from $u$ to $v$ while traversing the external boundary of $H_\nu$ clockwise and $p_r$ goes from $u$ to $v$ while traversing the external boundary of $H_\nu$ counterclockwise. Denote by $f_l$ the face of~$H$ outside $H_\nu$ and incident to $p_l$, and denote by $f_r$ the face of~$H$ outside $H_\nu$ and incident to $p_r$.
Also, for each pole $w \in \{u,v\}$ of $\nu$, denote by $a_{w,l}$ (resp. $a_{w,r}$) the angle at~$w$ in face $f_l$ (resp. $f_r$) of $H$. 
Similarly, with respect to $H'_\nu$ and $H'$, define $p'_l$, $p'_r$, $f'_l$, $f'_r$, and $a'_{w,l}$, $a'_{w,r}$ for each pole $w \in \{u,v\}$.  
The operation $\subs(H_\nu, H'_\nu)$ defines $H''$ as follows (schematic illustrations are given in ~\cref{fi:substitution-1}--\ref{fi:substitution-4}):

\begin{itemize}
	\item The set of vertices and the set of edges of $H''$ are the same as in $G$.
	\item The planar embedding of $H''$ is such that: All faces of $H$ outside $H_\nu$ and distinct from $f_l$ and $f_r$, as well as all faces of $H'_\nu$, are also faces of $H''$. Also, $H''$ has two faces $f''_l$ and $f''_r$ obtained by replacing $p_l$ with $p'_l$ and $p_r$ with $p'_r$ in the boundary of $f_l$ and $f_r$, respectively.
	\item The angle labeling of $H''$ is such that: $(i)$ All the angles at the vertices and along the edges of~$G$ not belonging to $G_\nu$ are those in $H$. $(ii)$ All the angles at the vertices of $G_\nu$ distinct from $u$ and $v$ are those in $H'_\nu$. $(iii)$ All the angles along the edges of $G_\nu$ are those in $H'_\nu$. $(iv)$ For each pole $w \in \{u,v\}$ of $\nu$, the angles at $w$ that are outside $G_\nu$ and that are neither in $f''_l$ nor in $f''_r$ are those in $H$; the angles at $w$ that are inside $G_\nu$ are those in $H'_\nu$; the angle $a''_{w,l}$ at $w$ in $f''_l$ and the angle $a''_{w,r}$ at $w$ in $f''_r$ are such that $a''_{w,l}=a_{w,l}$ and $a''_{w,r}=a_{w,r}$ if $\indeg_\nu(w)=1$, while $a''_{w,l}=a'_{w,l}$ and $a''_{w,r}=a'_{w,r}$ if $\indeg_\nu(w)>1$.     	
\end{itemize}

The next theorem proves that $H''$ is a valid orthogonal representation.

\begin{theor}\label{th:substitution}
	Let $G$ be a biconnected series-parallel 4-graph, $T$ be an SPQ$^*$-tree of $G$ with respect to a reference edge $e$, and $\nu$ be a non-root node of $T$. Let $H$ and $H'$ be two different orthogonal representations of $G$ with $e$ on the external face, and let $H_\nu$ and $H'_\nu$ be the restrictions of $H$ and $H'$ to $G_\nu$, respectively. If $\sigma(H_\nu)=\sigma(H'_\nu)$ then the graph $H''$ defined by $\subs(H_\nu,H'_\nu)$ is an orthogonal representation of $G$.
\end{theor}
\begin{proof}
	We have to show that the embedded labeled graph $H''$ defined by $\subs(H_\nu,H'_\nu)$ satisfies Properties~\textsf{(H1)} and~\textsf{(H2)} of an orthogonal representation. Clearly, since $H$ and $H'$ are orthogonal representations, \textsf{(H1)} holds for all vertices of $H''$ distinct from the poles $\{u,v\}$ of $G_\nu$; indeed, each vertex distinct from $u$ and $v$ inherits all its angles either from $H$ or from~$H'$. Analogously, each face of $H''$ distinct from $f''_l$ and $f''_r$ is either a face of $H$ or a face of $H'$, thus its angles satisfy Property~\textsf{(H2)}.  
	It remains to show that \textsf{(H1)} holds for $u$ and $v$, and that \textsf{(H2)} holds for $f''_l$ and $f''_r$. To this aim, we analyze different cases based on the indegree of the two poles $\{u,v\}$ of~$G_\nu$.
	
	\smallskip\noindent{\bf Case 1: $\indeg_\nu(u)=1$ and $\indeg_\nu(v)=1$.} Refer to \cref{fi:substitution-1}.
	In this case, the alias vertices $u'$ and $v'$, associated with $u$ and $v$ respectively, coincide with the poles, i.e., $u = u'$ and $v = v'$. Let $\overline{ux}$ and $\overline{yv}$ be the two edge segments of $H_\nu$ incident to $u$ and to $v$, respectively. Analogously, let $\overline{ux'}$ and $\overline{y'v}$ be the two edge segments of $H'_\nu$ incident to $u$ and to $v$, respectively. 
	Without loss of generality, assume that $H$ and $H'$ are oriented in such a way that both $\overline{ux}$ and $\overline{ux'}$ are vertical segments and that $u$ is below $x$ in any drawing of $H$ and $u$ is below $x'$ in any drawing of $H'$. By definition, since $u = u'$ and $v = v'$, the spirality $\sigma(H_\nu)$ (resp. $\sigma(H'_\nu)$) equals the number of right turns minus the number of left turns while moving from $u$ to $v$ along any simple path of $H_\nu$ (resp. of $H'_\nu$). Hence, since by hypothesis $\sigma(H_\nu) = \sigma(H'_\nu)$, the edge segments $\overline{yv}$ and $\overline{y'v}$ are either both horizontal or both vertical, and more precisely they are incident to $v$ in $H$ and to $v$ in $H'$ from the same side (south, north, west, or east). In this case, $\subs(H_\nu,H'_\nu)$ defines $a''_{u,l}=a_{u,l}$, $a''_{u,r}=a_{u,r}$, $a''_{v,l}=a_{v,l}$, and $a''_{v,r}=a_{v,r}$, which implies that all the angles around $u$ and $v$ in $H''$ coincide with the angles around $u$ and $v$ in~$H$. Hence, Property~\textsf{(H1)} holds for both $u$ and $v$ in $H''$. Also, let $n(p_l)$ (resp. $n(p_r)$) be the number of right turns minus the number of left turns along $p_l$ (resp. $p_r$) while moving from $u$ to $v$ in $H$. Similarly, let $n(p'_l)$ (resp. $n(p'_r)$) be the number of right turns minus the number of left turns along $p'_l$ (resp. $p'_r$) while moving from $u$ to $v$ in $H'$. Since $\sigma(H_\nu) = \sigma(H'_\nu)$, and since $u=u'$ and $v=v'$, we have $n(p_l)=n(p'_l)$  and  $n(p_r)=n(p'_r)$. It follows that $N_{90}(f''_l)-N_{270}(f''_l)=N_{90}(f_l)-N_{270}(f_l)$ and $N_{90}(f''_r)-N_{270}(f''_r)=N_{90}(f_r)-N_{270}(f_r)$, which imply Property~\textsf{(H2)} for $f''_l$ and $f''_r$ (note that, since $G$ is biconnected, $N_{360}(f)=0$ for every face $f$ of $H$ and of $H'$).            
	
	\begin{figure}[tb]
		\centering
		\includegraphics[width=0.9\columnwidth,page=1]{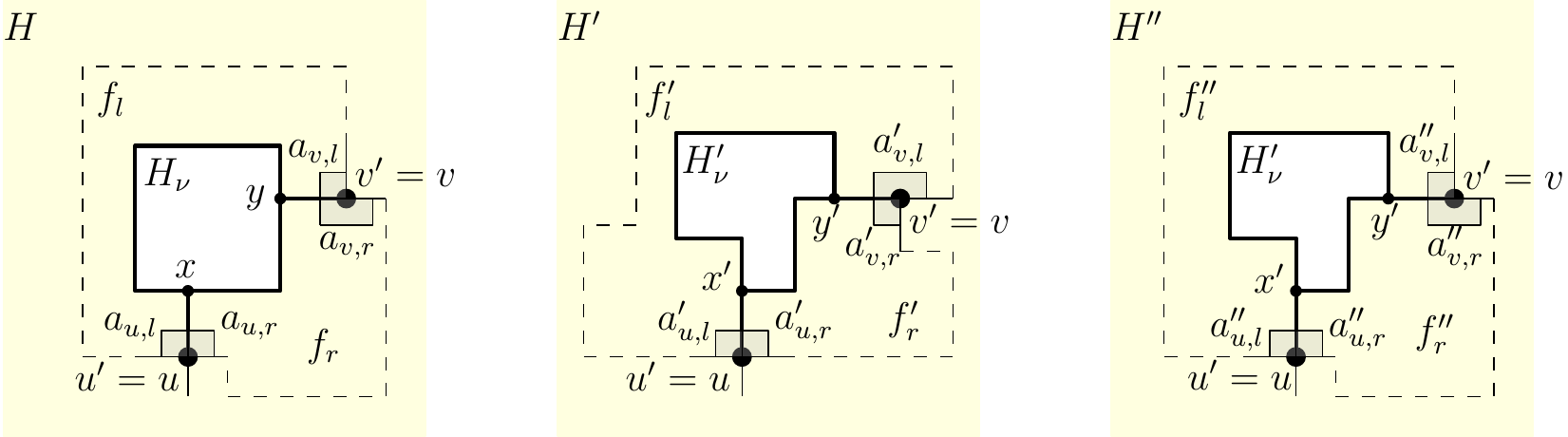}
		\caption{Case~1 of Theorem~\ref{th:substitution}: Schematic illustration of the graph $H''$ defined by $\subs(H_\nu,H'_\nu)$.}\label{fi:substitution-1}
	\end{figure}
	
	\smallskip\noindent{\bf Case 2: $\indeg_\nu(u)=1$ and $\indeg_\nu(v)>1$.} We distinguish two subcases: $\outdeg_\nu(v) = 1$ or $\outdeg_\nu(v)=2$. Assume first that $\outdeg_\nu(v) = 1$ (see \cref{fi:substitution-21}). In this case, $u$ and its alias vertex $u'$ coincide and the only alias vertex $v'$ associated with $v$ subdivides the edge segment incident to $v$ in $H$ and in $H'$. As in the analysis of Case~1, assume that $H$ and $H'$ are oriented so that each of the two edge segments $\overline{ux}$ and $\overline{ux'}$ in $H$ and in $H'$, respectively, is incident to $u$ from north. By definition, the spirality $\sigma(H_\nu)$ (resp. $\sigma(H'_\nu)$) in this case equals the number of right turns minus the number of left turns along any simple path of $H_\nu$ (resp. of $H'_\nu$) from $u$ to $v'$. Since $\sigma(H_\nu)=\sigma(H'_\nu)$, this implies that the segments $\overline{vv'}$ in $H$ and $H'$ are incident to $v$ from the same side. In this case, $\subs(H_\nu,H'_\nu)$ defines $a''_{u,l}=a_{u,l}$, $a''_{u,r}=a_{u,r}$, $a''_{v,l}=a'_{v,l}$, and $a''_{v,r}=a'_{v,r}$, which implies that all the angles around $u$ in~$H''$ coincide with the angles around $u$ and all the angles around $v$ in~$H''$ coincide with those around $v$ in~$H'$. Thus, Property~\textsf{(H1)} holds for $u$ and $v$ in $H''$. It remains to prove ~\textsf{(H2)} for $f''_l$ and $f''_r$.  Denote by $P_l$ (resp. $P_r$) the path of $H$ obtained by concatenating $p_l$ (resp. $p_r$) with the edge segment $\overline{vv'}$. Analogously, denote by $P'_l$ (resp. $P'_r$) the path of $H'$ obtained by concatenating $p'_l$ (resp. $p'_r$) with the edge segment $\overline{vv'}$. Since  $\sigma(H_\nu)=\sigma(H'_\nu)$, with the usual notation we have $n(P_l) = n(P'_l)$ and $n(P_r) = n(P'_r)$. Since, as observed above, the segments $\overline{vv'}$ in $H$ and $H'$ are incident to $v$ from the same side, and since the angles at $u$ are the same in $H''$ and in $H$, we have $N_{90}(f''_l)-N_{270}(f''_l)=N_{90}(f_l)-N_{270}(f_l)$ and $N_{90}(f''_r)-N_{270}(f''_r)=N_{90}(f_r)-N_{270}(f_r)$, which imply Property~\textsf{(H2)} for $f''_l$ and $f''_r$.     
	
	\begin{figure}[tb]
		\centering
		\begin{subfigure}{1\columnwidth}
			\centering
			\includegraphics[width=0.9\columnwidth,page=2]{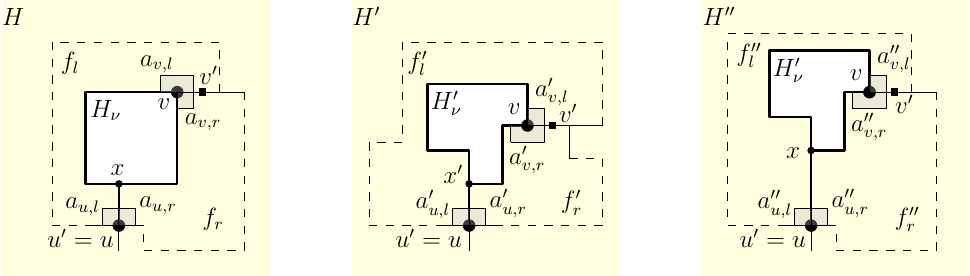}
			\subcaption{\centering}
			\label{fi:substitution-21}
		\end{subfigure}
		\hfil
		\begin{subfigure}{1\columnwidth}
			\centering
			\includegraphics[width=0.9\columnwidth,page=3]{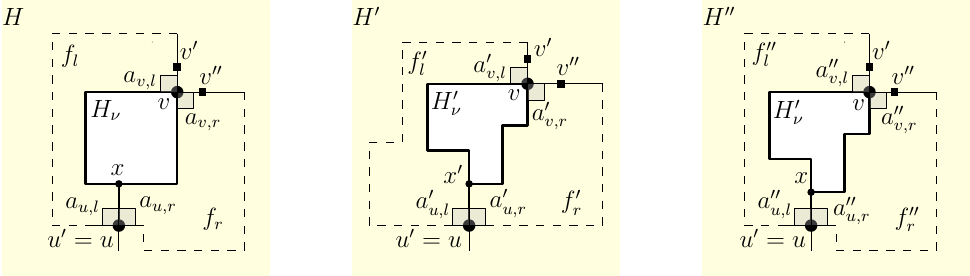}
			\subcaption{\centering}
			\label{fi:substitution-22}
		\end{subfigure}
		\caption{Case~2 of Theorem~\ref{th:substitution}: Schematic illustration of the graph $H''$ defined by $\subs(H_\nu,H'_\nu)$ when (a) $\outdeg(v)=1$ and (b) $\outdeg(v)=2$. }\label{fi:substitution-2}
	\end{figure}

	Suppose now that $\outdeg_\nu(v)=2$ (see \cref{fi:substitution-22}). In this case, $u$ and its alias vertex $u'$ coincide while $v$ has two alias vertices $v'$ and $v''$, which subdivides the two edge segments incident to $v$ in $H$ and in $H'$. Since $\deg(v)=4$, the angles at $v$ are all $90^\circ$ degree angles, both in $H$ and in $H'$. It follows that, the angles at $u$ and $v$ in $H''$ are the same as in $H$, i.e., Property~\textsf{(H1)} holds. Denote by $P_l$ (resp. $P_r$) the path of $H$ obtained by concatenating $p_l$ (resp. $p_r$) with the edge segment $\overline{vv'}$ (resp. $\overline{vv''}$). Analogously, denote by $P'_l$ (resp. $P'_r$) the path of $H'$ obtained by concatenating $p'_l$ (resp. $p'_r$) with the edge segment $\overline{vv'}$  (resp. $\overline{vv''}$). Since $\sigma(H_\nu)=\sigma(H'_\nu)$, we have $\frac{n(P_l)+n(P_r)}{2} = \frac{n(P'_l)+n(P'_r)}{2}$. On the other hand, since all angles at $v$ are right angles in $H$ and in $H'$, we have $n(P_r) = n(P_l) + 1$ and $n(P'_r) = n(P'_l) + 1$. This implies that $n(P_l) = n(P'_l)$ and $n(P_r) = n(P'_r)$, which, together with the fact that the angles at $u$ are the same in $H$ and $H''$, implies that $N_{90}(f''_l)-N_{270}(f''_l)=N_{90}(f_l)-N_{270}(f_l)$ and $N_{90}(f''_r)-N_{270}(f''_r)=N_{90}(f_r)-N_{270}(f_r)$. Hence, Property~\textsf{(H2)} holds for $f''_l$ and $f''_r$.     
	
	\smallskip\noindent{\bf Case 3: $\indeg_\nu(u)>1$ and $\indeg_\nu(v)=1$.} This case is symmetric to Case~2.
	
	\smallskip\noindent{\bf Case 4: $\indeg_\nu(u)>1$ and $\indeg_\nu(v)>1$.} In this case, there are three non-symmetric subcases to analyze, depending on the outdegree of $u$ and of $v$, i.e., $\outdeg_\nu(u)=\outdeg_\nu(v)=1$, or $\outdeg_\nu(u)=1$ and $\outdeg_\nu(v)=2$ (symmetrically $\outdeg_\nu(u)=2$ and $\outdeg_\nu(v)=1$), or $\outdeg_\nu(u)=\outdeg_\nu(v)=2$.
	In all these cases, for a pole $w \in \{u,v\}$, the angles defined by $\subs(H_\nu, H'_\nu)$ around $w$ in $H''$ are the same as in $H'$ (note that if $\outdeg_\nu(w)=2$, the angles at $w$ are all right angles, and they coincide both in $H$ and $H'$). Hence, Property~\textsf{(H1)} holds for $u$ and $v$ in $H''$. About Property~\textsf{(H2)}, we analyze the different subcases separately: 
	\begin{itemize}
		\item $\outdeg_\nu(u)=\outdeg_\nu(v)=1$ (see \cref{fi:substitution-41}). Each of the two poles $u$ and $v$ has a single alias vertex, denoted as $u'$ and $v'$, respectively. Let $P_l$ (resp. $P_r$) be the path of $H$ obtained by concatenating $p_l$ (resp. $p_r$) with the segments $\overline{u'u}$ and $\overline{vv'}$. Analogously, let $P'_l$ (resp. $P'_r$) be the path of $H'$ obtained by concatenating $p'_l$ (resp. $p'_r$) with the segments $\overline{u'u}$ and $\overline{vv'}$. Using the same notation as in the previous cases, we have $\sigma(H_\nu)=n(P_l)=n(P_r)$ and $\sigma(H'_\nu)=n(P'_l)=n(P'_r)$. Hence, since $\sigma(H_\nu)=\sigma(H'_\nu)$, we have $n(P_l)=n(P_r)=n(P'_l)=n(P'_r)$, which implies that $N_{90}(f''_l)-N_{270}(f''_l)=N_{90}(f_l)-N_{270}(f_l)$ and $N_{90}(f''_r)-N_{270}(f''_r)=N_{90}(f_r)-N_{270}(f_r)$. Hence, Property~\textsf{(H2)} holds for $f''_l$ and $f''_r$. 
		
		\item $\outdeg_\nu(u)=1$ and $\outdeg_\nu(v)=2$ (see \cref{fi:substitution-42}). The pole $u$ has a single alias vertex $u'$, while $v$ has two alias vertices $v'$ and $v''$. Let $P_l$ (resp. $P_r$) be the path of $H$ obtained by concatenating $p_l$ (resp. $p_r$) with the segments $\overline{u'u}$ and $\overline{vv'}$ (resp. $\overline{vv''}$). Analogously, let $P'_l$ (resp. $P'_r$) be the path of $H'$ obtained by concatenating $p'_l$ (resp. $p'_r$) with the segments $\overline{u'u}$ and $\overline{vv'}$ (resp. $\overline{vv''}$). Since $\sigma(H_\nu)=\sigma(H'_\nu)$, we have $\frac{n(P_l)+n(P_r)}{2}=\frac{n(P'_l)+n(P'_r)}{2}$. Also, since all the angles at $v$ are right angles, we have $n(P_r) = n(P_l) + 1$ and $n(P'_r) = n(P'_l) + 1$, which implies $n(P_l)=n(P'_l)$ and $n(P_r)=n(P'_r)$. Hence, $N_{90}(f''_l)-N_{270}(f''_l)=N_{90}(f_l)-N_{270}(f_l)$ and $N_{90}(f''_r)-N_{270}(f''_r)=N_{90}(f_r)-N_{270}(f_r)$, i.e.,  Property~\textsf{(H2)} holds for $f''_l$ and $f''_r$.       
		
		\begin{figure}[tb]
			\centering
			\begin{subfigure}{1\columnwidth}
				\centering
				\includegraphics[width=0.9\columnwidth,page=4]{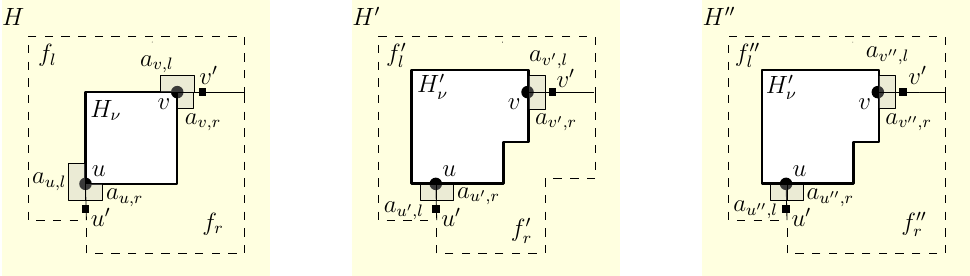}
				\subcaption{\centering}
				\label{fi:substitution-41}
			\end{subfigure}
			\hfil
			\begin{subfigure}{1\columnwidth}
				\centering
				\includegraphics[width=0.9\columnwidth,page=5]{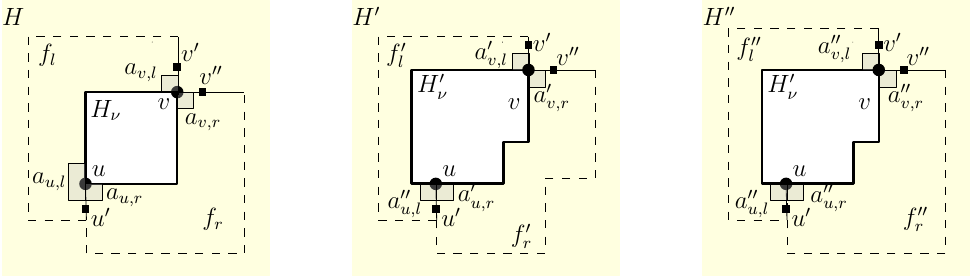}
				\subcaption{\centering}
				\label{fi:substitution-42}
			\end{subfigure}
			\hfill
			\begin{subfigure}{1\columnwidth}
				\centering
				\includegraphics[width=0.9\columnwidth,page=6]{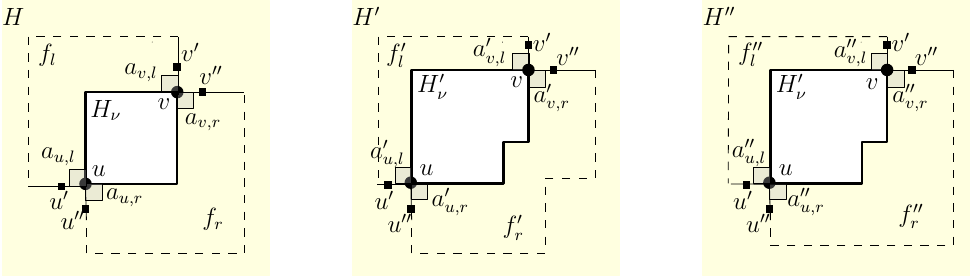}
				\subcaption{\centering}
				\label{fi:substitution-43}
			\end{subfigure}
			\caption{Case~3 of Theorem~\ref{th:substitution}: Schematic illustration of the graph $H''$ defined by $\subs(H_\nu,H'_\nu)$ when (a) $\outdeg(u)=\outdeg(v)=1$, (b) $\outdeg(u)=1$ and $\outdeg(v)=2$, and (c) $\outdeg(u)=\outdeg(v)=2$.}\label{fi:substitution-4}
		\end{figure}
		
		\item $\outdeg_\nu(u)=\outdeg_\nu(v)=2$ (see \cref{fi:substitution-43}). Each of the two poles $u$ and $v$ has two alias vertices, denoted as $\{u',u''\}$ and $\{v',v''\}$, respectively. Let $P_l$ (resp. $P_r$) be the path of $H$ obtained by concatenating $p_l$ (resp. $p_r$) with the segments $\overline{u'u}$ and $\overline{vv'}$ (resp. $\overline{u''u}$ and $\overline{vv''}$). Analogously, let $P'_l$ (resp. $P'_r$) be the path of $H'$ resulting from the concatenation of~$p'_l$ (resp. $p'_r$) with the segments $\overline{u'u}$ and $\overline{vv'}$ (resp. $\overline{u''u}$ and $\overline{vv''}$). Since $\sigma(H_\nu)=\sigma(H'_\nu)$, we have $\frac{n(P_l)+n(P_r)}{2} = \frac{n(P'_l)+n(P'_r)}{2}$. Also, since all the angles at $u$ and $v$ are right angles, we have $n(P_r) = n(P_l) + 2$ and $n(P'_r) = n(P'_l) + 2$. This implies that $n(P_l)=n(P'_l)$ and $n(P_r)=n(P'_r)$, which in turns implies that $N_{90}(f''_l)-N_{270}(f''_l)=N_{90}(f_l)-N_{270}(f_l)$ and $N_{90}(f''_r)-N_{270}(f''_r)=N_{90}(f_r)-N_{270}(f_r)$. Hence, Property~\textsf{(H2)} holds for $f''_l$ and $f''_r$.     
	\end{itemize}
\end{proof}

Based on Theorem~\ref{th:substitution}, in the following we can assume that two orthogonal components with the same spirality are equivalent, and we can describe the set of possible orthogonal representations for a component in terms of their values of spiralities.


\section{Rectilinear Plane Series-Parallel Graphs}\label{se:rect}
This section characterizes rectilinear plane series-parallel graphs.
Let $G$ be a plane series-parallel 4-graph. If $G$ is biconnected let $e$ be any edge on the external face of $G$; otherwise, we add a dummy edge $e$ that makes it biconnected (recall that, if $G$ is not biconnected we are assuming that the dummy edge $e$ can always be added in the external face).
Let $T$ be the SPQ$^*$-tree of $G$ with respect to~$e$ and let $\nu$ be a node of $T$.
%
We say that a component $G_{\nu}$ \emph{admits spirality $\sigma_\nu$} or, equivalently, that $\nu$ admits spirality $\sigma_\nu$, if there exists a rectilinear planar representation $H_\nu$ of $G_\nu$ with spirality $\sigma_\nu$ in some rectilinear planar representation~$H$ of~$G$.
The following lemmas immediately derive from the results by Di Battista et al~\cite{DBLP:journals/siamcomp/BattistaLV98}, which for any S-node or P-node $\nu$, relate the values of spirality for an orthogonal representation of $G_\nu$ to the values of spirality of the orthogonal representations of the child components of $G_\nu$ (i.e., the components corresponding to the children of~$\nu$). 
%
%
Namely, \cref{le:spirality-S-node} concentrates on an S-nodes, \cref{le:spirality-P-node-3-children} on P-nodes with three children, and \cref{le:spirality-P-node-2-children} on P-nodes with two children.  
See also \cref{fi:spirality-relationships} for an illustration.

\begin{lemma}[\cite{DBLP:journals/siamcomp/BattistaLV98}]\label{le:spirality-S-node}
	Let $\nu$ be an S-node of $T$ with children $\mu_1, \dots, \mu_h$ $(h \geq 2)$. The component $G_\nu$ admits spirality $\sigma_\nu$ if and only if $\sigma_\nu = \sum_{i=1}^{h}\sigma_{\mu_i}$, where $\sigma_{\mu_i}$ is a spirality value admitted by $G_{\mu_i}$ $(1 \leq i \leq h)$.
\end{lemma}	

\begin{lemma}[\cite{DBLP:journals/siamcomp/BattistaLV98}]\label{le:spirality-P-node-3-children}
	Let $\nu$ be a P-node of $T$ with three children $\mu_l$, $\mu_c$, and $\mu_r$. $G_{\nu}$ admits spirality $\sigma_\nu$ with $G_{\mu_l}$, $G_{\mu_c}$, $G_{\mu_r}$ in this left-to-right order if and only if there exist three values $\sigma_{\mu_l}$, $\sigma_{\mu_c}$, and $\sigma_{\mu_r}$ such that: (i) $G_{\mu_l}$, $G_{\mu_c}$, $G_{\mu_r}$ admit spirality $\sigma_{\mu_l}$, $\sigma_{\mu_c}$, $\sigma_{\mu_r}$, respectively; and (ii) $\sigma_\nu = \sigma_{\mu_l} - 2 = \sigma_{\mu_c} = \sigma_{\mu_r} + 2$.	
\end{lemma}

	\begin{figure}[tb]
	\centering
	\begin{subfigure}{0.24\columnwidth}
		\centering
		\includegraphics[width=1\columnwidth,page=1]{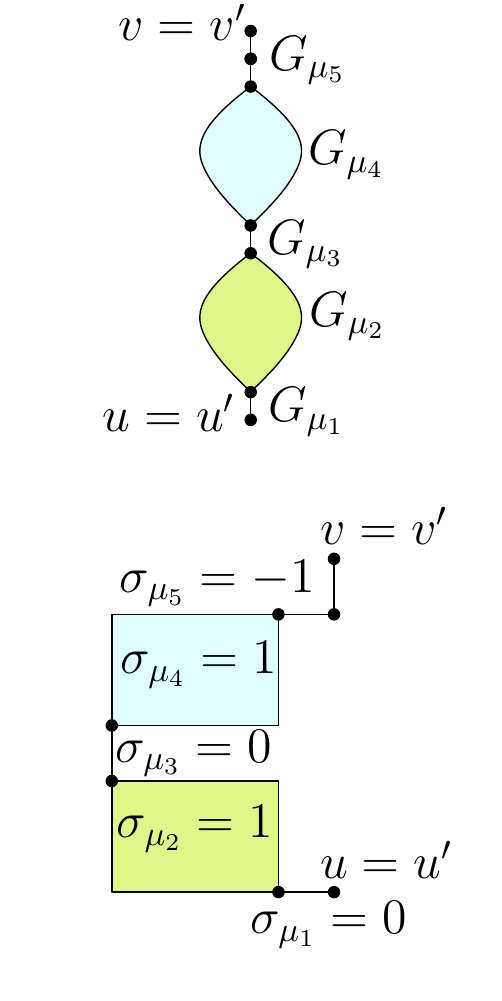}
		\subcaption{\centering S-node}
		\label{fi:spirality-relationships-S}
	\end{subfigure}
	\hfill
	\begin{subfigure}{0.26\columnwidth}
		\centering
		\includegraphics[width=1\columnwidth,page=2]{spirality-relationships.pdf}
		\subcaption{\centering P-node 3 children}
		\label{fi:spirality-relationships-P3}
	\end{subfigure}
	\hfill
	\begin{subfigure}{0.342\columnwidth}
		\centering
		\includegraphics[width=1\columnwidth,page=3]{spirality-relationships.pdf}
		\subcaption{\centering P-node 2 children}
		\label{fi:spirality-relationships-P2}
	\end{subfigure}
	\caption{Illustration of the relationships in: (a) Lemma~\ref{le:spirality-S-node} for S-nodes, (b) Lemma~\ref{le:spirality-P-node-3-children} for P-nodes with three children, and (c) Lemma~\ref{le:spirality-P-node-2-children} for P-nodes with two children.}\label{fi:spirality-relationships}
\end{figure}

If $\nu$ is a P-node with two children, denote by $\mu_l$ and $\mu_r$ its left and right child in $T$, respectively. If $\nu$ is a P-node with three children, denote by $\mu_l$, $\mu_c$, and $\mu_r$, the three children of $\nu$ from left to right. Also, for each pole $w \in \{u,v\}$ of $\nu$, the \emph{leftmost angle} at $w$ in $H$ is the angle formed by the leftmost external edge and the leftmost internal edge of $H_\nu$ incident to~$w$. The \emph{rightmost angle} at $w$ in $H$ is defined symmetrically.
We define two binary variables $\alpha_w^l$ and $\alpha_w^r$ as follows: $\alpha_w^l = 0$ ($\alpha_w^r = 0$) if the leftmost (rightmost) angle at $w$ in $H$ is of $180^\circ$, while $\alpha_w^l = 1$ ($\alpha_w^r = 1$) if this angle is of $90^\circ$.
Observe that if $\deg(w)=4$ or if $\nu$ has three children, $\alpha_w^l = \alpha_w^r = 1$. Also, if $\nu$ has two children, define two additional variables $k_w^l$ and $k_w^r$ as follows: $k_w^d = 1$ if $\indeg_{\mu_d}(w)=\outdeg_{\nu}(w)=1$,
while $k_{w}^d=1/2$ otherwise, for $d \in \{l,r\}$. For example, in 
\cref{fi:spiralities-minbend} the component of $\nu_3$ is such that  $k_u^l=k_u^r=1$, $k_v^l =k_v^r=\frac{1}{2}$, $\alpha_u^l=0$, and $\alpha_u^r=\alpha_v^l=\alpha_v^r=1$; the component of $\nu_5$ is such that $k_u^l=k_u^r=1$, $k_v^l=\frac{1}{2}$, $k_v^r=1$, $\alpha_u^l=0$, and $\alpha_u^r=\alpha_v^l=\alpha_v^r=1$..

\begin{lemma}[\cite{DBLP:journals/siamcomp/BattistaLV98}]\label{le:spirality-P-node-2-children}
	Let $\nu$ be a P-node of $T$ with two children $\mu_l$ and $\mu_r$, and with poles $u$ and $v$. $G_{\nu}$ admits spirality $\sigma_\nu$ with $G_{\mu_l}$ and $G_{\mu_r}$ in this left-to-right order if and only if there exist six values $\sigma_{\mu_l}$, $\sigma_{\mu_r}$, $\alpha_u^l$, $\alpha_u^r$, $\alpha_v^l$, and $\alpha_v^r$ such that: (i) $G_{\mu_l}$ and $G_{\mu_r}$ admit spirality $\sigma_{\mu_l}$ and $\sigma_{\mu_r}$, respectively; (ii) $\alpha_w^l \in \{0,1\}$, $\alpha_w^r \in \{0,1\}$, and $1 \leq \alpha_w^l+\alpha_w^r \leq 2$ for any $w \in \{u,v\}$; and (iii) $\sigma_\nu = \sigma_{\mu_l} - k_{u}^l \alpha_{u}^l - k_{v}^l \alpha_{v}^l = \sigma_{\mu_r} + k_{u}^r \alpha_{u}^r + k_{v}^r\alpha_{v}^r$.
\end{lemma}

In the following we prove a condition under which the plane graph $G_\nu$ is rectilinear planar, assuming that its child components (if $\nu$ is not a leaf of $T$) are rectilinear planar. This condition depends on the type of node~$\nu$ and is referred to as \emph{representability condition} of $\nu$ (or, equivalently, of $G_\nu$). Also, if the representability condition holds for $\nu$, we denote by $I_\nu$ the set of values of spirality for which $G_\nu$ is rectilinear planar, i.e., $G_\nu$ admits spirality $\sigma_\nu$ if and only if $\sigma_\nu \in I_\nu$. We prove that $I_\nu$ is always an interval (of all integer or all semi-integer values) and call it the \emph{representability interval} of $\nu$ (or,~equivalently,~of~$G_\nu$).

\subsection{Representability condition for Q$^*$-nodes and S-nodes}\label{se:qs-rect-charact}

For a Q$^*$-node $\nu$ representing a chain of $\ell$ edges, we say that $\ell$ is the \emph{length} of $\nu$. As the next lemmas prove, the components of Q$^*$- and S-nodes are always rectilinear planar, i.e., the representability condition is always~true.  

\begin{lemma}\label{le:Q*-representability}
	Let $\nu$ be a Q$^*$-node of length $\ell$. Graph $G_\nu$ is always rectilinear planar (i.e., its representability condition is always true) and its representability interval is $I_\nu = [-\ell+1, \ell-1]$.
\end{lemma}
\begin{proof}
	$G_\nu$ is a path with $\ell-1$ degree-2 vertices. For any integer $k \in [-\ell+1,0]$, a rectilinear planar representation $H_\nu$ of $G_\nu$ with spirality $k$ is obtained by making a left turn at $k$ degree-2 vertices of $G_\nu$ (going from the source to the sink pole), and no turn at any remaining vertex of $G_\nu$. Symmetrically, for any $k \in (0,\ell-1]$, we realize $H_\nu$ with spirality $k$ by making a right turn at exactly $k$ degree-2 vertices of $G_\nu$. It is clear that no values of spirality out of $I_\nu$ can be achieved.
\end{proof}	

\cref{fi:allspiralities-q} illustrates \cref{le:Q*-representability} for a Q$^*$-node $\nu$ of length $4$, for which $I_\nu=[-3,3]$. The figure depicts a rectilinear planar representation of $G_\nu$ with spirality $\sigma_\nu$ for every $\sigma_\nu \in I_\nu$. 

\begin{figure}[h]
	\captionsetup[subfigure]{labelformat=empty}
	\centering
	\begin{subfigure}{.09\columnwidth}
		\centering
		\includegraphics[width=\columnwidth,page=1]{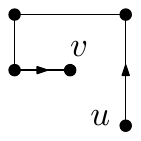}
		\subcaption{\centering $\sigma_\nu$=$-3$}
		\label{fi:allspiralities-q-1}
	\end{subfigure}
	\hfil
	\begin{subfigure}{.09\columnwidth}
		\centering
		\includegraphics[width=\columnwidth,page=2]{allspiralities-q.pdf}
		\subcaption{\centering $\sigma_\nu$=$-2$}
		\label{fi:allspiralities-q-2}
	\end{subfigure}
	\hfil
	\begin{subfigure}{.09\columnwidth}
		\centering
		\includegraphics[width=\columnwidth,page=3]{allspiralities-q.pdf}
		\subcaption{\centering $\sigma_\nu$=$-1$}
		\label{fi:allspiralities-q-3}
	\end{subfigure}
	\hfil
	\begin{subfigure}{.09\columnwidth}
		\centering
		\includegraphics[width=\columnwidth,page=4]{allspiralities-q.pdf}
		\subcaption{\centering $\sigma_\nu$=$0$}
		\label{fi:allspiralities-q-4}
	\end{subfigure}
	\hfil
	\begin{subfigure}{.09\columnwidth}
		\centering
		\includegraphics[width=\columnwidth,page=5]{allspiralities-q.pdf}
		\subcaption{\centering $\sigma_\nu$=$1$}
		\label{fi:allspiralities-q-5}
	\end{subfigure}
	\hfil
	\begin{subfigure}{.09\columnwidth}
		\centering
		\includegraphics[width=\columnwidth,page=6]{allspiralities-q.pdf}
		\subcaption{\centering $\sigma_\nu$=$2$}
		\label{fi:allspiralities-q-6}
	\end{subfigure}
	\hfil
	\begin{subfigure}{.09\columnwidth}
		\centering
		\includegraphics[width=\columnwidth,page=7]{allspiralities-q.pdf}
		\subcaption{\centering $\sigma_\nu$=$3$}
		\label{fi:allspiralities-q-7}
	\end{subfigure}
	\caption{Illustration of \cref{le:Q*-representability}. For a $Q^*$ node $\nu$ of length $4$ we have $I_\nu=[-3,3]$.  }\label{fi:allspiralities-q}
\end{figure}

\begin{lemma}\label{le:S-representability}
	Let $\nu$ be an S-node with $h \geq 2$ children $\mu_1, \dots, \mu_h$. Suppose that, for every $i \in [1,h]$, the representability interval of $G_{\mu_i}$ is $I_{\mu_i} = [m_i,M_i]$. Graph $G_\nu$ is always rectilinear planar (i.e., its representability condition is always true) and its representability interval is $I_\nu = [\sum_{i=1}^hm_i,\sum_{i=1}^hM_i]$.
\end{lemma}	
\begin{proof}
	We use induction on the number of children of $\nu$. In the base case $h=2$. By hypothesis $I_{\mu_1}=[m_1, M_1]$ and $I_{\mu_2}=[m_2,M_2]$. By Lemma~\ref{le:spirality-S-node}, a series composition of a rectilinear representation of $G_{\mu_1}$ with spirality $\sigma_{\mu_1}$ and of a rectilinear representation of $G_{\mu_2}$ with spirality $\sigma_{\mu_2}$ results in a rectilinear representation of $G_\nu$ with spirality $\sigma_\nu = \sigma_{\mu_1} + \sigma_{\mu_2}$. Hence, if $M_1 = m_1+r_1$ and $M_2 = m_2+r_2$, for two non-negative integers $r_1$ and $r_2$, then the possible values for $\sigma_\nu$ are exactly $m_1+m_2, m_1+1+m_2, \dots, m_1+r_1+m_2, \dots, m_1+r_1+m_2+1, \dots, m_1+r_1+m_2+r_2$, i.e., all values in the interval $[m_1+m_2, M_1+M_2]$. In the inductive case $h \geq 3$; consider the series composition $G'_1$ of $G_{\mu_1}, \dots, G_{\mu_{h-1}}$. Graph $G_\nu$ is the series composition of $G'_1$ and $G_{\mu_h}$. By inductive hypothesis the representability interval of $G'_1$ is $[\sum_{i=1}^{h-1}m_i,\sum_{i=1}^{h-1}M_i]$ and by Lemma~\ref{le:spirality-S-node} applied to $G'_1$ and $G_{\mu_h}$ we have $I_\nu = [\sum_{i=1}^hm_i,\sum_{i=1}^hM_i]$, using the same reasoning as for the base case.
\end{proof}

\cref{fi:allspiralities-s} illustrates \cref{le:S-representability}. \cref{fi:allspiralities-s-1} shows an S-node $\nu$ and its three children $\mu_1$, $\mu_2$, and $\mu_3$, where $\mu_1$ and $\mu_3$ are Q$^*$-nodes and $\mu_2$ is a P-node. \cref{fi:allspiralities-s-2} shows the components $G_\nu$, $G_{\mu_1}$, $G_{\mu_2}$, and $G_{\mu_3}$, where: $I_{\mu_1}=[0,0]$ and $I_{\mu_3}=[-2,2]$ by \cref{le:Q*-representability}; $I_{\mu_2}=[-1,-1]$, as $G_{\mu_2}$ only admits a rectilinear planar representation of spirality~$-1$. By \cref{le:S-representability}, $I_\nu = [\sum_{i=1}^hm_i,\sum_{i=1}^hM_i]=[0-2-1,0+2-1]=[-3,1]$. \cref{fi:allspiralities-s-3} depicts a rectilinear planar representation of $G_\nu$ with spirality $\sigma_\nu$ for every~$\sigma_\nu \in I_\nu$. 

\begin{figure}[h]
	\centering		
	\begin{subfigure}{.225\columnwidth}
		\centering
		\includegraphics[width=\columnwidth,page=6]{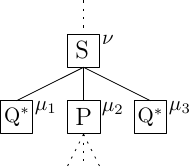}
		\subcaption{\centering}
		\label{fi:allspiralities-s-1}
	\end{subfigure}
	\hfil
	\begin{subfigure}{.35\columnwidth}
		\centering
		\includegraphics[width=\columnwidth,page=7]{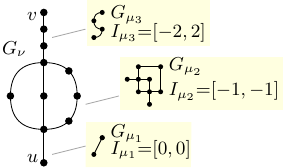}
		\subcaption{\centering}
		\label{fi:allspiralities-s-2}
	\end{subfigure}
	\hfil
	\begin{subfigure}{.8\columnwidth}
		\centering
		\includegraphics[width=\columnwidth,page=8]{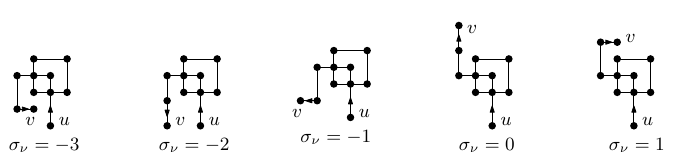}
		\subcaption{\centering}
		\label{fi:allspiralities-s-3}
	\end{subfigure}
	\caption{Illustration of \cref{le:S-representability}. (a)~An S-node $\nu$ with children $\mu_1$, $\mu_2$, and $\mu_3$. (b) The components $G_\nu$, $G_{\mu_1}$, $G_{\mu_2}$, and $G_{\mu_3}$. Since $I_{\mu_1}$=$[0,0]$, $I_{\mu_2}$=$[-1,-1]$,  and $I_{\mu_3}$=$[-3,3]$, we have $I_\nu$=$ [\sum_{i=1}^hm_i,\sum_{i=1}^hM_i]$=$[-3,1]$. (c)~A rectilinear planar representation of $G_\nu$ with spirality $\sigma_\nu$ for every $\sigma_\nu \in I_\nu$.}\label{fi:allspiralities-s}
\end{figure}


\subsection{Representability condition for P-nodes with three children}\label{se:p3-rect-charact} 

Different from S-nodes, if $\nu$ is a P-node and the pertinent graphs of the children of $\nu$ are rectilinear planar, $G_\nu$ may not be rectilinear planar. In this subsection we consider the case when $\nu$ has three children.

\begin{lemma}\label{le:P-3-children-representability}
	Let $\nu$ be a P-node with three children $\mu_l$, $\mu_c$, and $\mu_r$, ordered from left to right. Suppose that $G_{\mu_l}, G_{\mu_c}$, and $G_{\mu_r}$ are rectilinear planar and that their representability intervals are $I_{\mu_l}=[m_l, M_l]$, $I_{\mu_c}=[m_c, M_c]$, and $I_{\mu_r} = [m_r, M_r]$, respectively. Graph $G_\nu$ is rectilinear planar if and only if $[m_l-2,M_l-2] \cap [m_c,M_c] \cap [m_r+2,M_r+2] \neq \emptyset$. Also, if this representability condition holds then the representability interval of $G_\nu$ is $I_\nu = [\max\{m_l-2,m_c,m_r+2\},\min\{M_l-2,M_c,M_r+2\}]$.
\end{lemma}
\begin{proof}
	\noindent\textsf{Representability condition.} Suppose first that $G_\nu$ is rectilinear planar and let $H_\nu$ be a rectilinear planar representation of $G_\nu$ with spirality $\sigma_\nu$. By Lemma~\ref{le:spirality-P-node-3-children}, the spiralities $\sigma_{\mu_l}$, $\sigma_{\mu_c}$, and $\sigma_{\mu_r}$ for the representations of $G_{\mu_l}, G_{\mu_c}$, and $G_{\mu_r}$ in $H_\nu$ are such that $\sigma_{\mu_l}=\sigma_\nu+2$, $\sigma_{\mu_c}=\sigma_\nu$, and $\sigma_{\mu_r}=\sigma_\nu-2$. Since $\sigma_{\mu_l} \in [m_l,M_l]$, $\sigma_{\mu_c} \in [m_c,M_c]$, and $\sigma_{\mu_r} \in [m_r,M_r]$, we have $\sigma_\nu \in [m_l-2,M_l-2] \cap [m_c,M_c] \cap [m_r+2,M_r+2]$.
	Suppose vice versa that $[m_l-2,M_l-2] \cap [m_c,M_c] \cap [m_r+2,M_r+2] \neq \emptyset$, and let $k$ be any value in such intersection. Setting $\sigma_{\mu_l} = k + 2$, $\sigma_{\mu_c} = k$, and $\sigma_{\mu_r} = k - 2$ we have $\sigma_{\mu_l} \in [m_l,M_l]$, $\sigma_{\mu_c} \in [m_c,M_c]$, and $\sigma_{\mu_r} \in [m_r,M_r]$. By Lemma~\ref{le:spirality-P-node-3-children}, $G_\nu$ is rectilinear planar for a value of spirality~$\sigma_\nu = k$.
	
	\smallskip\noindent\textsf{Representability interval.} Assume that $G_\nu$ is rectilinear planar. Clearly $[\max\{m_l-2,m_c,m_r+2\},\min\{M_l-2,M_c,M_r+2\}] = [m_l-2,M_l-2] \cap [m_c,M_c] \cap [m_r+2,M_r+2]$, and by the truth of the feasiblity condition we have $[\max\{m_l-2,m_c,m_r+2\},\min\{M_l-2,M_c,M_r+2\}] \neq \emptyset$.
	Similarly to the first part of the proof of the representability condition, any rectilinear planar representation of $G_\nu$ has a value of spirality in the interaval $[\max\{m_l-2,m_c,m_r+2\},\min\{M_l-2,M_c,M_r+2\}]$.
	On the other hand, let $k \in [\max\{m_l-2,m_c,m_r+2\},\min\{M_l-2,M_c,M_r+2\}]$. Analogously to the second part of the proof of the representability condition, we can construct a rectilinear planar representation of $G_\nu$ with spirality $\sigma_\nu=k$, by combining in parallel rectilinear planar representations of $G_{\mu_l}$, $G_{\mu_c}$, and $G_{\mu_r}$ with spiralities $\sigma_{\mu_l} = \sigma_\nu + 2$, $\sigma_{\mu_c} = \sigma_\nu$, and $\sigma_{\mu_r} = \sigma_\nu - 2$.
\end{proof}

\cref{fi:three-children-spiralities} illustrates Lemma~\ref{le:P-3-children-representability}. In \cref{fi:three-children-spiralities-a}, $\nu$ has three children that are Q$^*$-nodes. By Lemma~\ref{le:Q*-representability}, $I_{\mu_l}=[-5,5]$,  $I_{\mu_c}=[-3,3]$, and $I_{\mu_c}=[-1,1]$. We have $[m_l-2,M_l-2] \cap [m_c,M_c] \cap [m_r+2,M_r+2] = [-7,3] \cap [-3,3] \cap [1,3] = \{1,2,3\}\neq \emptyset$ and, consequently, $G_\nu$ is rectilinear planar. Also, $I_\nu = [\max\{m_l-2,m_c,m_r+2\},\min\{M_l-2,M_c,M_r+2\}]=[1,3]$. In \cref{fi:three-children-spiralities-b}, the left child of $\nu$ is an S-node such that  $I_{\mu_l}=[-2,2]$. We have $[m_l-2,M_l-2] \cap [m_c,M_c] \cap [m_r+2,M_r+2] = [-4,-4] \cap [-3,3] \cap [1,3] = \emptyset$ and, consequently, $G_\nu$ is not rectilinear planar.

\begin{figure}[h]
	\captionsetup[subfigure]{labelformat=empty}
	\centering
	\begin{subfigure}{.3\columnwidth}
		\centering
		\includegraphics[width=\columnwidth,page=1]{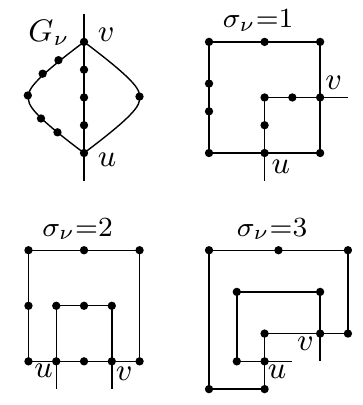}
		\subcaption{\centering~(a)}
		\label{fi:three-children-spiralities-a}
	\end{subfigure}
	\hfil
	\begin{subfigure}{.3\columnwidth}
		\centering
		\includegraphics[width=\columnwidth,page=2]{three-children-spiralities.pdf}
		\subcaption{\centering~(b)}
		\label{fi:three-children-spiralities-b}
	\end{subfigure}
	\caption{Illustration of Lemma~\ref{le:P-3-children-representability}. (a)~$G_\nu$ is rectilinear planar and $I_\nu=[1,3]$, (b)~$G_\nu$ is not rectilinear planar. }\label{fi:three-children-spiralities}
\end{figure}


\subsection{Representability condition for P-nodes with two children}\label{se:p2-rect-charact}

For a P-node $\nu$ with two children $\mu_l$ and $\mu_r$, the representability condition and interval depend on the indegree and outdegree of the poles of $\nu$ in $G_\nu$, $G_{\mu_l}$, and $G_{\mu_r}$. We define the \emph{type} of $\nu$ and of $G_\nu$ as follows (refer to~\cref{fi:P-node-types}):

\begin{itemize}
	\item \Pio{2}{\lambda\beta}: Both poles of $\nu$ have indegree two in $G_\nu$; also one pole has outdegree $\lambda$ in $G_\nu$ and the other pole has outdegree $\beta$ in $G_\nu$, for $1 \leq \lambda \leq \beta \leq 2$. This gives rise to the specific types \Pio{2}{11}, \Pio{2}{12}, and \Pio{2}{22}.
	
	\item \Pio{3d}{\lambda\beta}: One pole of $\nu$ has indegree two in $G_\nu$, while the other pole has indegree three in $G_\nu$ and indegree two in $G_{\mu_d}$ for $d \in \{l,r\}$; also one pole has outdegree $\lambda$ in $G_\nu$ and the other has outdegree $\beta$ in $G_\nu$, for $1 \leq \lambda \leq \beta \leq 2$, where $\lambda=\beta=2$ is not possible. This gives rise to the specific types \Pio{3l}{11}, \Pio{3r}{11}, \Pio{3l}{12}, \Pio{3r}{12}.
	
	\item \Pin{3dd'}: Both poles of $\nu$ have indegree three in $G_\nu$; one of the two poles has indegree two in $G_{\mu_d}$ and the other has indegree two in $G_{\mu_{d'}}$, for  $dd' \in \{ll,lr,rr\}$ (both poles have outdegree one in $G_\nu$). Hence, the specific types are~\Pin{3ll},~\Pin{3lr},~\Pin{3rr}.
\end{itemize}

\begin{figure}[!h]
	\centering
	\includegraphics[width=.85\columnwidth]{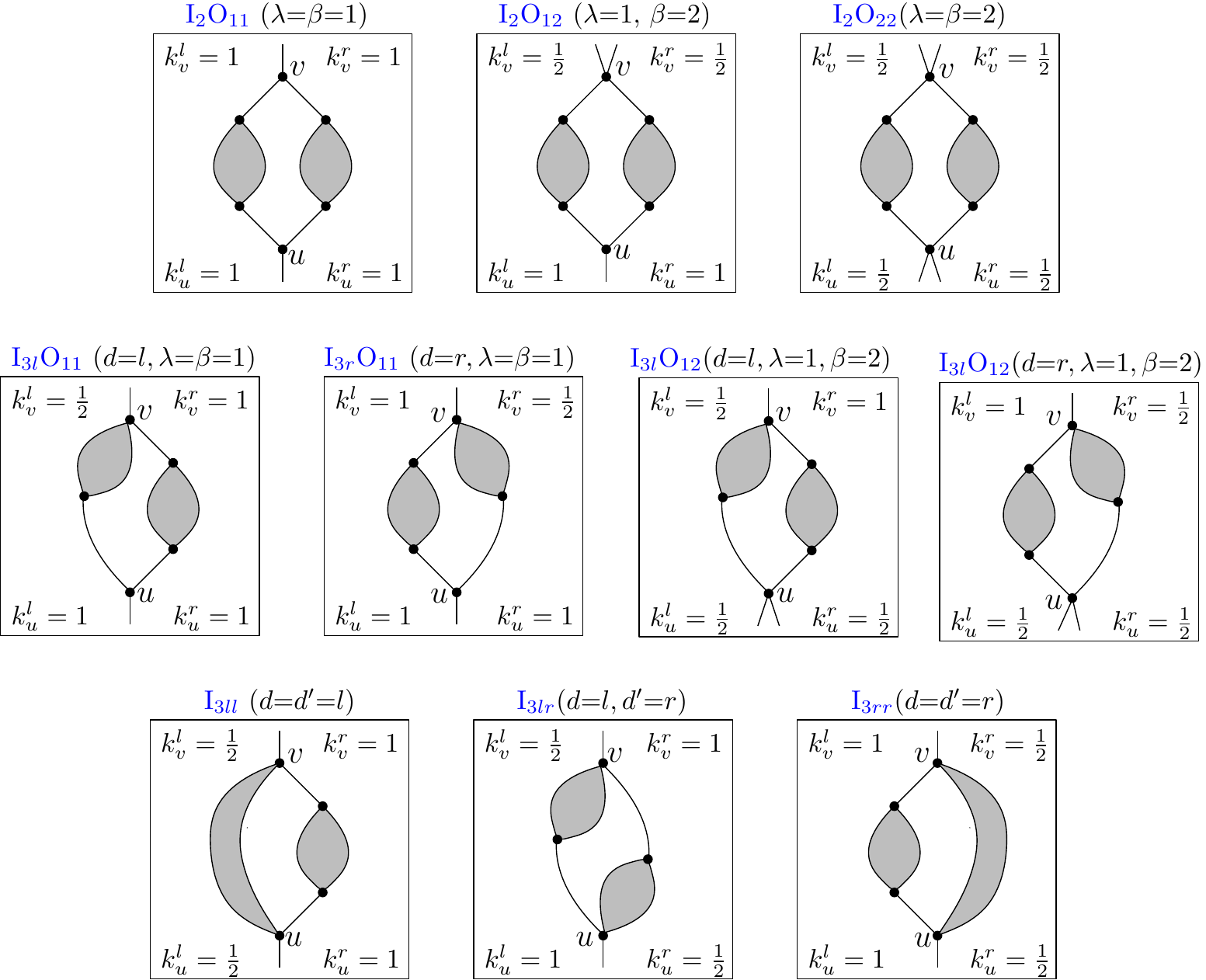}
	\caption{Schematic illustration of the different types of P-nodes with two children.}\label{fi:P-node-types}
\end{figure}

\begin{lemma}\label{le:P-2-children-support-type1}
	Let $G_\nu$ be a P-node of type \Pio{2}{\lambda\beta} with children $\mu_l$ and $\mu_r$. $G_\nu$ is rectilinear planar if and only if $G_{\mu_l}$ and $G_{\mu_r}$ are rectilinear planar for values of spiralities $\sigma_{\mu_l}$ and $\sigma_{\mu_r}$ such that $\sigma_{\mu_l}-\sigma_{\mu_r}\in[2,4-\gamma]$, where $\gamma = \lambda + \beta - 2$.
\end{lemma}
\begin{proof}
	We distinguish three cases, based on the values of $\alpha$ and $\beta$.
	
	\noindent{\bf Case 1: $\alpha=\beta=1$.}  In this case $G_\nu$ is of type \Pio{2}{11} and we prove that $G_\nu$ is rectilinear planar if and only if $G_{\mu_l}$ and $G_{\mu_r}$ are rectilinear planar for values of spiralities $\sigma_{\mu_l}$ and $\sigma_{\mu_r}$ such that $\sigma_{\mu_l}-\sigma_{\mu_r}\in[2,4]$. For a \Pio{2}{11} component we have $k_u^l=k_v^l=k_u^r=k_v^r=1$. 
	
	If $G_\nu$ is rectilinear planar then $1 \leq \alpha_{u}^l+\alpha_{u}^r \leq 2$ and $1 \leq \alpha_{v}^l+\alpha_{v}^r \leq 2$ in any rectilinear planar representation of $G_\nu$. Hence, by Lemma~\ref{le:spirality-P-node-2-children}, for any value of spirality $\sigma_\nu$  we have $\sigma_{\mu_l}-\sigma_{\mu_r}=\alpha_{u}^l+\alpha_{v}^l+\alpha_{u}^r+\alpha_{v}^r \in [2,4]$. 
	Suppose vice versa that $G_{\mu_l}$ and $G_{\mu_r}$ are rectilinear planar for values of spirality $\sigma_{\mu_l}$ and $\sigma_{\mu_r}$ such that $\sigma_{\mu_l}-\sigma_{\mu_r}\in[2,4]$. We show that $G_\nu$ admits a rectilinear planar representation $H_\nu$. To define $H_\nu$, we combine in parallel the two rectilinear planar representations of $G_{\mu_l}$ and $G_{\mu_r}$ and suitably assign the values of $\alpha_u^d$ and $\alpha_v^d$ ($d \in \{l,r\}$), depending on the value of $\sigma_{\mu_l}-\sigma_{\mu_r}$. This assignment is such that for any cycle $C$ of $G_\nu$ through $u$ and $v$, the number of $90^\circ$ angles minus the number of $270^\circ$ angles in the interior of $C$ is equal to four.
	Poles $u$ and $v$ split $C$ into two paths $\pi_l$ and $\pi_r$. The spirality $\sigma_{\mu_l}$ equals the number of right turns minus the number of left turns along $\pi_l$ while going from $u$ to $v$, which in turns corresponds to the number of $90^\circ$ angles minus the number of $270^\circ$ angles in the interior of $C$ at the vertices of $\pi_l$. Similarly, $-\sigma_{\mu_r}$ equals the number of right turns minus the number of left turns along $\pi_r$ while going from $v$ to $u$, which in turns corresponds to the number of $90^\circ$ angles minus the number of $270^\circ$ angles in the interior of $C$ at the vertices of $\pi_r$. By also taking into account the angles at $u$ and $v$ inside $C$, the number of $90^\circ$ angles minus the number of $270^\circ$ angles in the interior of $C$ can be expressed as $a_c = \sigma_{\mu_l}-\sigma_{\mu_r} + 4 - \alpha_u^l - \alpha_u^r - \alpha_v^l - \alpha_v^r$. 
	We distinguish the following three cases:
	%
	$(i)$ If $\sigma_{\mu_l}-\sigma_{\mu_r}=2$, then for every pole $w \in \{u,v\}$ we set  
	$\alpha_w^l$ and $\alpha_w^r$ such that $\alpha_w^l + \alpha_w^r = 1$. 
	$(ii)$ If $\sigma_{\mu_l}-\sigma_{\mu_r}=3$, then for one pole $w \in \{u,v\}$ we set $\alpha_w^l$ and $\alpha_w^r$ such that $\alpha_w^l + \alpha_w^r = 1$, and for the other pole $w' \in \{u,v\}$ we set $\alpha_{w'}^l = \alpha_{w'}^r = 1$.
	$(iii)$ If $\sigma_{\mu_l}-\sigma_{\mu_r}=4$, then for every pole $w \in \{u,v\}$ we set  
	$\alpha_w^l = \alpha_w^r = 1$.  
	In all the cases above, we have that $a_c=4$.
	Also, any other cycle not passing through $u$ and $v$ is an orthogonal polygon because it belongs to a rectilinear planar representation of either $G_{\mu_l}$ (with spirality $\sigma_{\mu_l}$) or $G_{\mu_r}$ (with spirality $\sigma_{\mu_r}$).
	
	\smallskip\noindent{\bf Case 2: $\alpha=1, \beta=2$.} In this case $G_\nu$ is of type \Pio{2}{12} and we prove that $G_\nu$ is rectilinear planar if and only if $G_{\mu_l}$ and $G_{\mu_r}$ are rectilinear planar for values of spiralities $\sigma_{\mu_l}$ and $\sigma_{\mu_r}$ such that $\sigma_{\mu_l}-\sigma_{\mu_r}\in[2,3]$. Suppose, w.l.o.g., that $\outdeg_\nu(u)=1$ and $\outdeg_\nu(v)=2$. We have $k_u^l=k_u^r=1$ and $k_v^l=k_v^r=\frac{1}{2}$.
	
	If $G_\nu$ is rectilinear planar then $\alpha_v^l+\alpha_v^r=2$ and $\alpha_u^l+\alpha_u^r\in[1,2]$. By Lemma~\ref{le:spirality-P-node-2-children}, $\sigma_{\mu_l}-\sigma_{\mu_r}=k_{u}^l \alpha_{u}^l +k_{v}^r \alpha_{v}^l +  k_{u}^l \alpha_{u}^r + k_{v}^r\alpha_{v}^r$, and hence $\sigma_{\mu_l}-\sigma_{\mu_r}= \alpha_{u}^l + \frac{1}{2}  \alpha_{v}^l + \alpha_{u}^r + \frac{1}{2} \alpha_{v}^r\in [2,3]$.
	Suppose vice versa that $G_{\mu_l}$ and $G_{\mu_r}$ are rectilinear planar for values of spirality $\sigma_{\mu_l}$ and $\sigma_{\mu_r}$ such that $\sigma_{\mu_l}-\sigma_{\mu_r}\in [2,3]$. We show that $G_\nu$ admits a rectilinear planar representation $H_\nu$. 
	To define $H_\nu$, we combine in parallel the two rectilinear planar representations of $G_{\mu_l}$ and $G_{\mu_r}$ and suitably set $\alpha_u^d$ and $\alpha_v^d$ ($d \in \{l,r\}$). Namely, we set $\alpha_v^l=\alpha_v^r=1$, and the values of $\alpha_u^l$ and $\alpha_u^r$ as follows: $(i)$ if $\sigma_{\mu_l}-\sigma_{\mu_r}=2$, we set $\alpha_u^l$ and $\alpha_u^r$ such that $\alpha_u^l + \alpha_u^r = 1$; $(ii)$ if $\sigma_{\mu_l}-\sigma_{\mu_r}=3$, we set $\alpha_u^l=\alpha_u^r=1$. With an argument similar to the previous case, for any cycle $C$ through $u$ and $v$, the number of $90^\circ$ angles minus the number of $270^\circ$ angles in $C$ can be expressed in this case by $a_c = \sigma_{\mu_l}-\sigma_{\mu_r} + 4 - \alpha_u^l - \alpha_u^r - 1$ (the angle at $v$ in $C$ is always of 90$^\circ$ degrees). In case $(i)$ we have $a_c = 2 + 4 - 1 - 1 = 4$; in case $(ii)$ we have $a_c = 3 + 4 - 2 - 1 = 4$.
	Also, any other cycle not passing through $u$ and $v$ is an orthogonal polygon because it belongs to a rectilinear planar representation of either $G_{\mu_l}$ or $G_{\mu_r}$.
	
	\smallskip\noindent{\bf Case 3: $\alpha=\beta=2$.} In this case $G_\nu$ is of type \Pio{2}{22} and we prove that $G_\nu$ is rectilinear planar if and only if $G_{\mu_l}$ and $G_{\mu_r}$ are rectilinear planar for values of spiralities $\sigma_{\mu_l}$ and $\sigma_{\mu_r}$ such that $\sigma_{\mu_l}-\sigma_{\mu_r}=2$. We have $k_u^l=k_u^r=\frac{1}{2}$.
	
	If $G_\nu$ is rectilinear planar then $\alpha_{u}^l+\alpha_{u}^r = \alpha_{v}^l+\alpha_{v}^r = 2$. By Lemma~\ref{le:spirality-P-node-2-children}, $\sigma_{\mu_l} = \sigma_{\nu} + 1$ and $\sigma_{\mu_r} = \sigma_{\nu} - 1$; hence $\sigma_{\mu_l}-\sigma_{\mu_r}=2$.  
	Suppose vice versa that $\sigma_{\mu_l}-\sigma_{\mu_r}=2$. We show that $G_\nu$ admits a rectilinear planar representation $H_\nu$. Again, we obtain $H_\nu$ by combining in parallel the two rectilinear planar representations of $G_{\mu_l}$ and $G_{\mu_r}$ and by suitably setting $\alpha_u^d$ and $\alpha_v^d$ ($d \in \{l,r\}$). In this case, for any cycle $C$ through $u$ and $v$, the number of $90^\circ$ angles minus the number of $270^\circ$ angles in $C$ can be expressed by $a_c = \sigma_{\mu_l}-\sigma_{\mu_r} + 1 + 1$ (both the angles at $u$ and $v$ inside $C$ is always of 90$^\circ$ degrees). We then set $\alpha_u^l=\alpha_v^l=\alpha_u^r=\alpha_v^r=1$, which guarantees $a_c = 4$.   
	Also, any other cycle not passing through $u$ and $v$ is an orthogonal polygon because it belongs to a rectilinear planar representation of either $G_{\mu_l}$ or $G_{\mu_r}$.
\end{proof}

\begin{lemma}\label{le:P-2-children-representability-type1}
	Let $\nu$ be a P-node of type \Pio{2}{\lambda\beta} with children $\mu_l$ and $\mu_r$. Suppose that $G_{\mu_l}$ and $G_{\mu_r}$ are rectilinear planar with representability intervals $I_{\mu_l}=[m_l, M_l]$ and $I_{\mu_r} = [m_r, M_r]$, respectively. Graph $G_\nu$ is rectilinear planar if and only if $[m_l-M_r,M_l-m_r] \cap [2,4-\gamma] \neq \emptyset$, where $\gamma = \lambda + \beta -2$. Also, if this representability condition holds then the representability interval of $G_\nu$ is $I_\nu = [\max\{m_l-2,m_r\}+\frac{\gamma}{2}, \min\{M_l, M_r+2\}-\frac{\gamma}{2}]$.
\end{lemma}
\begin{proof}
	We prove the correctness of the representability condition and
	the validity of the representability interval.
	
	\smallskip\noindent\textsf{Representability condition.}
	Suppose that $G_\nu$ is rectilinear planar. By Lemma~\ref{le:P-2-children-support-type1}, $G_{\mu_l}$ and $G_{\mu_r}$ admit spiralities $\sigma_{\mu_l}$ and $\sigma_{\mu_r}$, respectively, such that $\sigma_{\mu_l}-\sigma_{\mu_r} \in [2,4-\gamma]$. Hence, $m_l-M_r \leq \sigma_{\mu_l}-\sigma_{\mu_r} \leq 4-\gamma$ and $M_l-m_r \geq \sigma_{\mu_l}-\sigma_{\mu_r} \geq 2$, i.e.,  $[m_l-M_r,M_l-m_r] \cap [2,4-\gamma] \neq \emptyset$.
	
	Suppose, vice versa that $[m_l-M_r,M_l-m_r] \cap [2,4-\gamma] \neq \emptyset$. By hypothesis $G_{\mu_l}$ (resp. $G_{\mu_r}$) is rectilinear planar for every integer value of spirality in the interval $[m_l,M_l]$ (resp. $[m_r,M_r]$). This implies that for every integer value $k$ in the interval $[m_l-M_r, M_l-m_r]$, there exist rectilinear planar representations for $G_{\mu_l}$ and $G_{\mu_r}$ with spiralities $\sigma_{\mu_l}$ and $\sigma_{\mu_r}$ such that $\sigma_{\mu_l}-\sigma_{\mu_r} = k$. Since by hypothesis there exists a value $k \in [m_l-M_r,M_l-m_r] \cap [2,4-\gamma]$, there must be two values of spiralities $\sigma_{\mu_l}$ and $\sigma_{\mu_r}$ for the representations of $G_{\mu_l}$ and $G_{\mu_r}$ such that $\sigma_{\mu_l}-\sigma_{\mu_r} = k \in [2,4-\gamma]$. Hence, by Lemma~\ref{le:P-2-children-support-type1}, $G_\nu$ is rectilinear planar.
	
	\smallskip\noindent\textsf{Representability interval.}
	We analyze three cases, based on the values of $\alpha$ and $\beta$.
	
	\smallskip\noindent{\bf Case 1: $\alpha=\beta=1$.}  $G_\nu$ is of type \Pio{2}{11} and we prove that $I_\nu = [\max\{m_l-2,m_r\}, \min\{M_l, M_r+2\}]$.  
	Assume first that $\sigma_\nu$ is the spirality of a rectilinear representation of $G_\nu$. By Lemma~\ref{le:spirality-P-node-2-children}, $\sigma_\nu \in [m_l-2, M_r+2]$. Also, since for a \Pio{2}{11} component $k_u^l=k_v^l=k_u^r=k_v^r=1$, we have $\sigma_\nu = \sigma_{\mu_r} + \alpha_{u}^r + \alpha_{v}^r$, which implies $\sigma_\nu \geq m_r$. Analogously, $\sigma_\nu = \sigma_{\mu_l} - \alpha_{u}^l - \alpha_{v}^l \leq M_l$. Hence, $\sigma_\nu \in I_\nu = [\max\{m_l-2,m_r\}, \min\{M_l,M_r+2\}]$.
	
	Assume vice versa that $k$ is any integer in the interval $I_\nu = [\max\{m_l-2,m_r\}, M = \min\{M_l,M_r+2\}]$. We show that $G_\nu$ admits a rectilinear planar representation with spirality $\sigma_\nu = k$. By hypothesis $k \leq \min\{M_l, M_r + 2\} \leq M_l$; also, $k \geq \max\{m_l-2,m_r\} \geq m_l-2$, i.e., $k+2 \geq m_l$. Hence $[k,k+2] \cap [m_l, M_l] \neq \emptyset$. Analogously, $k \leq \min\{M_l, M_r + 2\} \leq M_r + 2$, i.e., $k-2 \leq M_r$; also, $k \geq \max\{m_l-2,m_r\} \geq m_r$. Hence $[k-2,k] \cap [m_r,M_r] \neq \emptyset$.
	We now distinguish the following sub-cases:
	
	\begin{itemize}	
		\item{\bf Case 1.1: $k\le M_l-2$.} Consider any two rectilinear planar representations $H_{\mu_l}$ of $G_{\mu_l}$ and $H_{\mu_r}$ of $G_{\mu_r}$ with spirality $\sigma_{\mu_l}=k+2$ and $\sigma_{\mu_r}\in [k-2,k] \cap [m_r, M_r] \neq \emptyset$, respectively. As already observed, $k+2 \geq m_l$ and by hypothesis $k+2 \le M_l$; hence $\sigma_{\mu_l} \in [m_l,M_l]$. With this choice we have $2\leq \sigma_{\mu_l}-\sigma_{\mu_r}\leq 4$, and we can combine $H_{\mu_l}$ and $H_{\mu_r}$ in parallel as in the proof of Lemma~\ref{le:P-2-children-support-type1} to obtain a rectilinear planar representation $H_\nu$ of $G_\nu$. By Lemma~\ref{le:spirality-P-node-2-children} the spirality of $H_\nu$ equals $\sigma_{\mu_l} - \alpha_{l}^u - \alpha_{l}^v=k+2 - \alpha_{l}^u - \alpha_{l}^v$ and it suffices to set $\alpha_{l}^u=\alpha_{l}^v=1$ (which is always possible, as these two values correspond to $90^\circ$ angles) to get $\sigma_\nu = k$.
		
		\item{\bf Case 1.2: $k= M_l-1$.} Consider any rectilinear planar representation $H_{\mu_l}$ of $G_{\mu_l}$ with spirality $\sigma_{\mu_l} = k+1 = M_l$. To suitably choose the spirality of a rectilinear planar representation $H_{\mu_r}$ of $G_{\mu_r}$, observe that by the representability condition $M_l-2 \geq m_r$ and, as already proved, $M_r \geq k-2$, i.e., $M_r \geq M_l-3$. It follows that $[M_l-3,M_l-2] \cap [m_r,M_r] \neq \emptyset$. Hence, either $M_l-3 \in [m_r,M_r]$ (possibly $m_r=M_r=M_l-3$) or $M_l-2 \in [m_r,M_r]$ (possibly $m_r=M_r=M_l-2$). In the first case, choose any representation $H_{\mu_r}$ with spirality $\sigma_r = M_l-3$, which implies $\sigma_{\mu_l}-\sigma_{\mu_r}=3 \in [2,4]$. In the second case, choose $H_{\mu_r}$ with spirality $\sigma_r = M_l-2$, which implies $\sigma_{\mu_l}-\sigma_{\mu_r}=2 \in [2,4]$. The representations $H_{\mu_l}$ and $H_{\mu_r}$ can be combined in parallel to get a representation of $G_\nu$ with spirality $\sigma_{\nu}=k$. Namely, by Lemma~\ref{le:spirality-P-node-2-children} we can set $\alpha_u^l=0$ and $\alpha_v^l=1$ (or vice versa); also, if $\sigma_{\mu_r}=M_l-2$ we set $\alpha_u^r=0$ and $\alpha_v^l=1$ (or vice versa), while if $\sigma_{\mu_r}=M_l-3$ we set $\alpha_u^r=\alpha_v^l=1$.
		
		\item{\bf Case 1.3: $k= M_l$.} In this case, we can combine in parallel a representation $H_{\mu_l}$ of $G_{\mu_l}$ with spirality $\sigma_{\mu_l}=k=M_l$ and a representation  
		$H_{\mu_r}$ of $G_{\mu_r}$ with spirality $\sigma_{\mu_r}=k-2=M_l-2$, which implies that $ \sigma_{\mu_l}-\sigma_{\mu_r}=2$. By the representability condition we have $M_l-2 \geq m_r$, i.e., $\sigma_{\mu_r} \geq m_r$; also, $k \leq \min\{M_l,M_r+2\}\leq M_r+2$, i.e., $\sigma_{\mu_r} \leq M_r$. Hence, $\sigma_{\mu_r} \in [m_r,M_r]$. By Lemma~\ref{le:spirality-P-node-2-children} we can set $\alpha_u^l=\alpha_v^l=0$ and $\alpha_u^r=\alpha_v^l=1$ to get a representation of $G_\nu$ with spirality $\sigma_\nu = k$. 
	\end{itemize}
	
	\smallskip\noindent{\bf Case 2: $\alpha=1, \beta=2$.}  $G_\nu$ is of type \Pio{2}{12} and we prove that $I_\nu = [\max\{m_l-2,m_r\}+\frac{1}{2}, \min\{M_l, M_r+2\}-\frac{1}{2}]$.
	Assume first that $G_\nu$ is rectilinear planar and let $H_\nu$ be a rectilinear planar representation of $G_\nu$ with spirality~$\sigma_\nu$. Let $H_{\mu_l}$ and $H_{\mu_r}$ be the rectilinear planar representations of $G_{\mu_l}$ and $G_{\mu_r}$ contained in $H_\nu$, and let $\sigma_{\mu_l}$ and $\sigma_{\mu_r}$ be their corresponding spiralities. 
	By Lemma~\ref{le:P-2-children-support-type1}, $\sigma_{\mu_l} - \sigma_{\mu_r} \in [2,3]$, i.e., $\sigma_{\mu_l} \in [2+\sigma_{\mu_r}, 3+\sigma_{\mu_r}]$. Since $\sigma_{\mu_l} \in [m_l,M_l]$ and $\sigma_{\mu_r} \in [m_r,M_r]$, we have $\sigma_{\mu_l}\ge \max\{m_l,m_r+2\}$.
	Suppose, w.l.o.g, that $\outdeg_\nu(v)=2$ and $\outdeg_\nu(u)=1$. We have $k_{u}^r=k_{u}^l=1$,  $k_{v}^r=k_{v}^l=\frac{1}{2}$, $\alpha_{u}^l \in [0,1]$, $\alpha_{u}^r \in [0,1]$, and $\alpha_{v}^l=\alpha_{v}^r = 1$. 
	By Lemma~\ref{le:spirality-P-node-2-children}, $\sigma_\nu = \sigma_{\mu_l} - \alpha_{u}^l - \frac{1}{2} \alpha_{v}^l$. Since $-\alpha_{u}^l - \frac{1}{2} \alpha_{v}^l\ge -\frac{3}{2}$, we have $\sigma_\nu\ge \max\{m_l,m_r+2\}-\frac{3}{2}$. It follows that $\sigma_\nu\ge \max\{m_l-2,m_r\}+\frac{1}{2}$. Analogously, since $\sigma_{\mu_r} \in [\sigma_{\mu_l}-3,\sigma_{\mu_l}-2]$, we have $\sigma_{\mu_r}\le \min\{M_l-2,M_r\}$. By Lemma~\ref{le:spirality-P-node-2-children}, $\sigma_\nu = \sigma_{\mu_r} + \alpha_{u}^r + \frac{1}{2} \alpha_{v}^r$. Since $\alpha_{u}^r + \frac{1}{2} \alpha_{v}^r \le \frac{3}{2}$,  we have $\sigma_\nu\le \min\{M_l-2,M_r\}+\frac{3}{2}$. It follows that  $\sigma_\nu\le \min\{M_l,M_r+2\}-\frac{1}{2}$. Therefore, $\sigma_\nu \in I_\nu$.
	
	Assume vice versa that $k$ is a semi-integer in the interval $I_\nu = [\max\{m_l-2,m_r\}+\frac{1}{2},\min\{M_l,M_r+2\}-\frac{1}{2}]$. We show that $G_\nu$ has a rectilinear planar representation with spirality $\sigma_\nu = k$.
	Since $k \in [m_l-\frac{3}{2}, M_l-\frac{1}{2}]$ we have $k + \frac{1}{2} \leq M_l$ and $k + \frac{3}{2} \geq m_l$, i.e., $[k+\frac{1}{2},k+\frac{3}{2}]\cap [m_l,M_l] \neq \emptyset$. Also, since $m_l$ and $M_l$ are both integer numbers while $k$ is semi-integer, it is impossible to have $k + 1 = m_l = M_l$. It follows that $k+\frac{1}{2} \in [m_l,M_l]$ or $k+\frac{3}{2}\in [m_l,M_l]$. 
	With the same reasoning, we have $k \in [m_r+\frac{1}{2}, M_r+\frac{3}{2}]$ and $[k-\frac{3}{2},k-\frac{1}{2}] \cap [m_r,M_r] \neq \emptyset$. Hence, $k-\frac{3}{2}\in [m_r,M_r]$ or $k-\frac{1}{2}\in [m_r,M_r]$. We now prove that $k+\frac{3}{2} \in [m_l,M_l]$ or $k-\frac{3}{2} \in [m_r,M_r]$. Suppose for a contradiction that $k+\frac{3}{2}\not \in [m_l,M_l]$ and $k-\frac{3}{2}\not \in [m_r,M_r]$. In that case $k+\frac{1}{2}\in [m_l,M_l]$ and $k-\frac{1}{2} \in [m_r,M_r]$. Consequently, $k+\frac{1}{2}=M_l$ and $k-\frac{1}{2}=m_r$. Hence, $M_l-m_r=1$ and, by the representability condition, $G_\nu$ is not rectilinear planar, a contradiction. 
	As in the previous case, a rectilinear representation of $G_\nu$ with spirality $k$ is obtained by combining in parallel a representation $H_{\mu_l}$ of $G_{\mu_l}$ with spirality $\sigma_{\mu_l}$ and a representation $H_{\mu_r}$ of $G_{\mu_r}$ with spirality $\sigma_{\mu_r}$, for two suitable values  $\sigma_{\mu_l}$ and $\sigma_{\mu_r}$. Based on the previous considerations, we distinguish the following sub-cases:      
	
	\begin{itemize}
		\item \textbf{Case 2.1: $k+\frac{3}{2}\not \in [m_l,M_l]$.} 
		This implies that $k+\frac{1}{2}\in [m_l,M_l]$ and $k-\frac{3}{2}\in [m_r,M_r]$, and therefore we set $\sigma_{\mu_l}=k+\frac{1}{2}$ and $\sigma_{\mu_r}=k-\frac{3}{2}$.
		
		\item \textbf{Case 2.2: $k-\frac{3}{2}\not \in [m_r,M_r]$.} 
		This implies that $k+\frac{3}{2}\in [m_l,M_l]$ and $k-\frac{1}{2}\in [m_r,M_r]$, and therefore we set $\sigma_{\mu_l}=k+\frac{3}{2}$ and $\sigma_{\mu_r}=k-\frac{1}{2}$.
		
		\item \textbf{Case 2.3: $k+\frac{3}{2}\in [m_l,M_l]$ and $k-\frac{3}{2}\in [m_r,M_r]$.} We set $\sigma_{\mu_l}=k+\frac{3}{2}$ and $\sigma_{\mu_r}=k-\frac{3}{2}$.
	\end{itemize} 
	
	Note that, in all the three sub-cases we have $\sigma_{\mu_l}-\sigma_{\mu_r}\in [2,3]$, hence by Lemma~\ref{le:P-2-children-support-type1} there exists a rectilinear planar representation $H_\nu$ of $G_\nu$ that contains $H_{\mu_l}$ and $H_{\mu_r}$. It remains to prove that the spirality $\sigma_\nu$ of $H_\nu$ is equal to $k$. Suppose, w.l.o.g, that $\outdeg_\nu(u)=1$ and $\outdeg_\nu(v)=2$.  We have  $k_{u}^r=k_{u}^l=1$ and $k_{v}^r=k_{v}^l=\frac{1}{2}$. Since $G_\nu$ is rectilinear planar, $\alpha_{u}^l\in[0,1]$ and  $\alpha_{v}^l=1$. By Lemma~\ref{le:spirality-P-node-2-children}, $\sigma_\nu = \sigma_{\mu_l} - \alpha_{u}^l - \frac{1}{2} \alpha_{v}^l$. In Case~2.1 we have $\sigma_\nu = k+\frac{1}{2} - \alpha_{u}^l - \frac{1}{2} \alpha_{v}^l$; choosing $\alpha_{u}^l=0$ and $\alpha_{v}^l=1$ we have $\sigma_\nu = k$. In Cases~2.2 and~2.3  we have $\sigma_\nu = k+\frac{3}{2} - \alpha_{u}^l - \frac{1}{2} \alpha_{v}^l$; choosing $\alpha_{u}^l=1$ and $\alpha_{v}^l=1$ we have  $\sigma_\nu = k$.
	
	\smallskip\noindent{\bf Case 3: $\alpha=\beta=2$.}  $G_\nu$ is of type \Pio{2}{22} and we prove that $I_\nu = [\max\{m_l-2,m_r\}+1, \min\{M_l, M_r+2\}-1]$. 
	Assume first that $G_\nu$ is rectilinear planar and let $H_\nu$ be a rectilinear planar representation of $G_\nu$ with spirality~$\sigma_\nu$. Let $H_{\mu_l}$ and $H_{\mu_r}$ be the rectilinear planar representations of $G_{\mu_l}$ and $G_{\mu_r}$ contained in $H_\nu$, and let  $\sigma_{\mu_l}$ and $\sigma_{\mu_r}$ be their spiralities. Since both $u$ and $v$ have outdegree two in $G_\nu$ we have that $\alpha_{u}^l+\alpha_{u}^r = \alpha_{v}^l+\alpha_{v}^r = 2$. By Lemma~\ref{le:spirality-P-node-2-children},  $\sigma_{\mu_l} = \sigma_{\nu} + 1$ and $\sigma_{\mu_r} = \sigma_{\nu} - 1$. By the representability condition, $\sigma_{\mu_r}=\sigma_{\mu_l}-2$. Hence $\sigma_{\mu_r}\ge m_l-2$ and $\sigma_{\mu_r}\ge \max\{m_l-2,m_r\}$. 
	Also by $\sigma_{\nu}=\sigma_{\mu_r}+1$, $\sigma_\nu\ge \max\{m_l-2,m_r\}+1$. Similarly, by the representability condition, $\sigma_{\mu_l}=\sigma_{\mu_r}+2$. Hence $\sigma_{\mu_l}\le M_r+2$ and $\sigma_{\mu_l}\le \max\{M_l,M_r+2\}$. Since $\sigma_{\mu_l}=\sigma_\nu+1$ we have $\sigma_\nu\le \max\{M_l,M_r+2\}-1$.
	
	Assume vice versa that $k$ is an integer in the interval $I_\nu = [\max\{m_l-2,m_r\}+1,\min\{M_l,M_r+2\}-1]$. We show that there exists a rectilinear planar representation of $G_\nu$ with spirality $\sigma_\nu=k$. We have  $k +1 \in [\max\{m_l,m_r+2\},\min\{M_l,M_r+2\}]$ and therefore $k +1 \in [m_l,M_l]$. Hence there exists a rectilinear planar representation $H_{\mu_l}$ of $G_{\mu_l}$ with spirality $\sigma_{\mu_l}=k+1$. Similarly, we have  $k -1 \in [\max\{m_l-2,m_r\},\min\{M_l-2,M_r\}]$ and therefore $k-1 \in [m_r,M_r]$. Hence there exists a rectilinear planar representation $H_{\mu_r}$ of $G_{\mu_r}$ with spirality $\sigma_{\mu_r} = k-1$. By the representability condition $G_\nu$ has a rectilinear planar representation $H_\nu$; also, following the same construction as in the proof of Lemma~\ref{le:P-2-children-support-type1}, the spirality of $H_\nu$ is $\sigma_\nu=k$.
\end{proof}

\begin{lemma}\label{le:P-2-children-support-type2}
	Let $G_\nu$ be a P-node of type \Pio{3d}{\alpha\beta} and let $\mu_l$ and $\mu_r$ be its two children. $G_\nu$ is rectilinear planar if and only if $G_{\mu_l}$ and $G_{\mu_r}$ are rectilinear planar for values of spiralities $\sigma_{\mu_l}$ and $\sigma_{\mu_r}$, respectively, such that $\sigma_{\mu_l}-\sigma_{\mu_r}\in[\frac{5}{2},\frac{7}{2}-\gamma]$, where $\gamma = \alpha + \beta - 2$.
\end{lemma} 
\begin{proof}
	We distinguish four cases, based on the values of $\alpha$, $\beta$, and $d$.
	
	\noindent{\bf Case 1: $\alpha=\beta=1$, $d=l$.}  $G_\nu$ is of type \Pio{3l}{11} and we prove that $G_\nu$ is rectilinear planar if and only if $G_{\mu_l}$ and $G_{\mu_r}$ are rectilinear planar for values of spiralities $\sigma_{\mu_l}$ and $\sigma_{\mu_r}$ such that $\sigma_{\mu_l}-\sigma_{\mu_r}\in[\frac{5}{2},\frac{7}{2}]$. For an \Pio{3l}{11} component we have $k_u^l=k_u^r=k_v^r=1$ and $k_v^l=\frac{1}{2}$. 
	
	If $G_\nu$ is rectilinear planar then $1 \leq \alpha_{u}^l+\alpha_{u}^r \leq 2$ and $\alpha_{v}^l=\alpha_{v}^r=1$ in any rectilinear planar representation of $G_\nu$. Hence, by Lemma~\ref{le:spirality-P-node-2-children}, for any value of spirality $\sigma_\nu$ we have $\sigma_{\mu_l}-\sigma_{\mu_r}=\alpha_{u}^l+\frac{1}{2}\alpha_{v}^l+\alpha_{u}^r+\alpha_{v}^r \in [\frac{5}{2},\frac{7}{2}]$. 
	
	Suppose vice versa that $G_{\mu_l}$ and $G_{\mu_r}$ are rectilinear planar for values of spirality $\sigma_{\mu_l}$ and $\sigma_{\mu_r}$ such that $\sigma_{\mu_l}-\sigma_{\mu_r}\in[\frac{5}{2},\frac{7}{2}]$. We show that $G_\nu$ admits a rectilinear planar representation $H_\nu$. To define $H_\nu$, we combine in parallel the two rectilinear planar representations of $G_{\mu_l}$ and $G_{\mu_r}$ and suitably assign the values of $\alpha_u^l$ and $\alpha_u^r$, depending on the value of $\sigma_{\mu_l}-\sigma_{\mu_r}$. 
	
	Let $v'$ be the alias vertex of $G_{\mu_l}$ that is in $G_\nu$. Any cycle $C$ that goes through $u$ and $v$ also passes through $v'$. We show that the number of $90^\circ$ angles minus the number of $270^\circ$ angles in the interior of $C$ is equal to four.
	Vertices $u$ and $v'$ split $C$ into two paths $\pi_l$ and $\pi_r$. Suppose to visit $C$ clockwise. The number of right turns minus left turns along $\pi_l$ while going from $u$ to $v'$ equals $\sigma_{\mu_l}+\frac{1}{2}$.  The number of right turns minus left turns along $\pi_r$ while going from $v'$ to $u$ equals $-\sigma_{\mu_r}$. Hence, the sum $\sigma_{\mu_l}+\frac{1}{2}-\sigma_{\mu_r}+2-\alpha_u^r-\alpha_u^l$ corresponds to the number of $90^\circ$ angles minus the number of $270^\circ$ angles in the interior of $C$ at the vertices of $\pi_l$. Notice that $\alpha_u^r+\alpha_u^l\in \{1,2\}$ since $u$ is a vertex of degree 3. If $\sigma_{\mu_l}-\sigma_{\mu_r}=\frac{5}{2}$ we set $\alpha_u^r+\alpha_u^l=1$ and we have $\sigma_{\mu_l}+\frac{1}{2}-\sigma_{\mu_r}+2-\alpha_u^r-\alpha_u^l=\frac{5}{2}+\frac{1}{2}+2-1=4$. Else, if $\sigma_{\mu_l}-\sigma_{\mu_r}=\frac{7}{2}$ we set $\alpha_u^r+\alpha_u^l=2$ and we have $\sigma_{\mu_l}+\frac{1}{2}-\sigma_{\mu_r}+2-\alpha_u^r-\alpha_u^l=\frac{7}{2}+\frac{1}{2}+2-2=4$.

	Also, any other cycle not passing through $u$ and $v$ is an orthogonal polygon because it belongs to a rectilinear planar representation of either $G_{\mu_l}$ (with spirality $\sigma_{\mu_l}$) or $G_{\mu_r}$ (with spirality $\sigma_{\mu_r}$). 
	
	\noindent{\bf Case 2: $\alpha=1$, $\beta=2$, $d=l$.} $G_\nu$ is of type \Pio{3l}{12} and we prove that $G_\nu$ is rectilinear planar if and only if $G_{\mu_l}$ and $G_{\mu_r}$ are rectilinear planar for values of spiralities $\sigma_{\mu_l}$ and $\sigma_{\mu_r}$ such that $\sigma_{\mu_l}-\sigma_{\mu_r}=\frac{5}{2}$ (note that this corresponds to the interval $[\frac{5}{2},\frac{7}{2}-\gamma]$ claimed in the lemma). For an \Pio{3l}{12} component, $k_u^l=k_u^r=k_v^l=\frac{1}{2}$ and $k_v^r=1$.

	If $G_\nu$ is rectilinear planar then $ \alpha_{u}^l=\alpha_{u}^r=\alpha_{v}^l=\alpha_{v}^r=1$ in any rectilinear planar representation of $G_\nu$. Hence, by Lemma~\ref{le:spirality-P-node-2-children}, for any value of spirality $\sigma_\nu$  we have $\sigma_{\mu_l}-\sigma_{\mu_r}=\frac{1}{2}\alpha_{u}^l+\frac{1}{2}\alpha_{v}^l+\frac{1}{2}\alpha_{u}^r+\alpha_{v}^r =\frac{5}{2}$. 
	
	Suppose vice versa that $G_{\mu_l}$ and $G_{\mu_r}$ are rectilinear planar for values of spirality $\sigma_{\mu_l}$ and $\sigma_{\mu_r}$ such that $\sigma_{\mu_l}-\sigma_{\mu_r}=\frac{5}{2}$. We show that $G_\nu$ admits a rectilinear planar representation $H_\nu$. To define $H_\nu$, we combine in parallel the two rectilinear planar representations of $G_{\mu_l}$ and $G_{\mu_r}$ and assign values  $\alpha_u^l=\alpha_v^l=\alpha_u^r=\alpha_v^r=1$. 
	Let $v'$ be the alias vertex of $G_{\mu_l}$ that is in $G_\nu$. Any cycle $C$ that goes through $u$ and $v$ also passes through $v'$. We show that the number of $90^\circ$ angles minus the number of $270^\circ$ angles in the interior of $C$ is equal to four.
	Vertices $u$ and $v'$ split $C$ into two paths $\pi_l$ and $\pi_r$. Suppose to visit $C$ clockwise. The number of right turns minus left turns along $\pi_l$ while going from $u$ to $v'$ equals $\sigma_{\mu_l}+\frac{1}{2}$.  The number of right turns minus left turns along $\pi_r$ while going from $v'$ to $u$ equals $-\sigma_{\mu_r}$. Also, pole $u$ forms a $90^\circ$ angle in $C$. Hence, the sum $\sigma_{\mu_l}+\frac{1}{2}-\sigma_{\mu_r}+1$ corresponds to the number of $90^\circ$ angles minus the number of $270^\circ$ angles in $C$ at the vertices of $\pi_l$. Since $\sigma_{\mu_l}-\sigma_{\mu_r}=\frac{5}{2}$ we have $\sigma_{\mu_l}+\frac{1}{2}-\sigma_{\mu_r}+1=\frac{5}{2}+\frac{1}{2}+1=4$.
	
	Also, any other cycle not passing through $u$ and $v$ is an orthogonal polygon because it belongs to a rectilinear planar representation of either $G_{\mu_l}$ (with spirality $\sigma_{\mu_l}$) or $G_{\mu_r}$ (with spirality $\sigma_{\mu_r}$).

	\noindent{\bf Case 3: $\alpha=\beta=1$, $d=r$.} Symmetric to Case~1.
	
	\noindent{\bf Case 4: $\alpha=1$, $\beta=2$, $d=r$.} Symmetric to Case~2.
\end{proof}

\begin{lemma}\label{le:P-2-children-representability-type2}
	Let $\nu$ be a P-node of type \Pio{3d}{\alpha\beta} with children $\mu_l$ and $\mu_r$.
	Suppose that $G_{\mu_l}$ and $G_{\mu_r}$ are rectilinear planar with representability intervals $I_{\mu_l}=[m_l, M_l]$ and $I_{\mu_r} = [m_r, M_r]$, respectively. Graph $G_\nu$ is rectilinear planar if and only if $[m_l-M_r,M_l-m_r] \cap [\frac{5}{2},\frac{7}{2}-\gamma] \neq \emptyset$, where $\gamma = \alpha + \beta -2$. Also, if this representability condition holds then the representability interval of $G_\nu$ is $I_\nu = [\max\{m_l-\frac{3}{2},m_r+1\}+\frac{\gamma-\rho(d)}{2},\min\{M_l-\frac{1}{2}, M_r+2\}-\frac{\gamma+\rho(d)}{2}]$, where $\rho(\cdot)$ is a function such that $\rho(r)=1$ and $\rho(l)=0$.
\end{lemma}
\begin{proof}
	We prove the correctness of the representability condition and the validity of the representability interval.
	
	\smallskip\noindent\textsf{Representability condition.}
	Suppose that $G_\nu$ is rectilinear planar. By Lemma~\ref{le:P-2-children-support-type2}, $G_{\mu_l}$ and $G_{\mu_r}$ admit spiralities $\sigma_{\mu_l}$ and $\sigma_{\mu_r}$, respectively, such that $\sigma_{\mu_l}-\sigma_{\mu_r} \in [\frac{5}{2},\frac{7}{2}-\gamma]$, where $\gamma = \alpha + \beta -2$. Hence, $m_l-M_r \leq \sigma_{\mu_l}-\sigma_{\mu_r} \leq \frac{7}{2}-\gamma$ and $M_l-m_r \geq \sigma_{\mu_l}-\sigma_{\mu_r} \geq \frac{5}{2}$, i.e.,  $[m_l-M_r,M_l-m_r] \cap [\frac{5}{2},\frac{7}{2}-\gamma] \neq \emptyset$.
	
	Suppose, vice versa that $[m_l-M_r,M_l-m_r] \cap [\frac{5}{2},\frac{7}{2}-\gamma] \neq \emptyset$. By hypothesis $G_{\mu_l}$ (resp. $G_{\mu_r}$) is rectilinear planar for every value of spirality in the interval $[m_l,M_l]$ (resp. $[m_r,M_r]$). This implies that for every semi-integer value $k$ in the interval $[m_l-M_r, M_l-m_r]$, there exist rectilinear planar representations for $G_{\mu_l}$ and $G_{\mu_r}$ with spiralities $\sigma_{\mu_l}$ and $\sigma_{\mu_r}$ such that $\sigma_{\mu_l}-\sigma_{\mu_r} = k$. Since by hypothesis there exists a value $k \in [m_l-M_r,M_l-m_r] \cap [\frac{5}{2},\frac{7}{2}-\gamma]$, there must be two values of spiralities $\sigma_{\mu_l}$ and $\sigma_{\mu_r}$ for the representations of $G_{\mu_l}$ and $G_{\mu_r}$ such that $\sigma_{\mu_l}-\sigma_{\mu_r} = k \in [\frac{5}{2},\frac{7}{2}-\gamma]$. Hence, by Lemma~\ref{le:P-2-children-support-type2}, $G_\nu$ is rectilinear planar.
	
	\smallskip\noindent\textsf{Representability interval.} 
	We analyze four cases, based on the values of $\alpha$, $\beta$, and $d$; we assume, w.l.o.g., that $v$ is the pole of degree four.
	
	\noindent{\bf Case 1: $\alpha=\beta=1$, $d=l$.} $G_\nu$ is of type \Pio{3l}{11} and we prove that $I_\nu = [\max\{m_l-\frac{3}{2},m_r+1\},\min\{M_l-\frac{1}{2}, M_r+2\}]$.
	Assume first that $\sigma_\nu$ is the spirality of a rectilinear planar representation of $G_\nu$.  Since for an \Pio{3l}{11} component we have $k_u^l=k_u^r=k_v^r=1$ and $k_v^l=\frac{1}{2}$, by Lemma~\ref{le:spirality-P-node-2-children} we 
	have $\sigma_\nu = \sigma_{\mu_r} + \alpha_{u}^r + \alpha_{v}^r$ and $\sigma_\nu = \sigma_{\mu_l} - \alpha_{u}^l - \frac{1}{2}\alpha_{v}^l$. Since $\alpha_u^l + \alpha_u^r \in \{1,2\}$ and $\alpha_v^l=\alpha_v^r=1$, we have
	$\sigma_\nu \geq m_r + 1$, $\sigma_\nu \leq M_r+2$, $\sigma_\nu \geq m_l-\frac{3}{2}$, and $\sigma_\nu \leq M_l-\frac{1}{2}$.

	We now show that, if $\sigma_\nu \in  [\max\{m_l-\frac{3}{2},m_r+1\},\min\{M_l-\frac{1}{2},M_r+2\}]$, there exists a rectilinear planar representation of $G_\nu$ with spirality $\sigma_\nu$. We have $\sigma_\nu \in [m_l-\frac{3}{2}, M_l-\frac{1}{2}]$.
	Hence, $\sigma_\nu + \frac{1}{2} \leq M_l$ and $\sigma_\nu + \frac{3}{2} \geq m_l$, i.e., $[\sigma_\nu+\frac{1}{2},\sigma_{\nu}+\frac{3}{2}]\cap [m_l,M_l]\not = \emptyset$. Also, since $m_l$ and $M_l$ are both semi-integer numbers while $\sigma_\nu$ is integer, it is impossible to have $\sigma_\nu + 1 = m_l = M_l$. It follows that $\sigma_\nu+\frac{1}{2}\in [m_l,M_l]$ or $\sigma_\nu+\frac{3}{2}\in [m_l,M_l]$. With the same reasoning, we have $\sigma_\nu \in [m_r+1, M_r+2]$ and $[\sigma_\nu-2,\sigma_{\nu}-1]\cap [m_r,M_r]\not = \emptyset$. Hence, $\sigma_\nu-2\in [m_r,M_r]$ or $\sigma_\nu-1\in [m_r,M_r]$. 
	We now prove the following.
	
	\begin{clm}\label{cl:either3/2and2-L}
		Either $\sigma_\nu+\frac{3}{2}\in [m_l,M_l]$ or $\sigma_\nu-2 \in [m_r,M_r]$.
	\end{clm}
	\begin{claimproof}
		Suppose for a contradiction that $\sigma_\nu+\frac{3}{2}\not \in [m_l,M_l]$ and $\sigma_\nu-2\not \in [m_r,M_r]$. In that case $\sigma_\nu+\frac{1}{2}\in [m_l,M_l]$ and $\sigma_\nu-1 \in [m_r,M_r]$. Consequently, $\sigma_\nu+\frac{1}{2}=M_l$ and $\sigma_\nu-1=m_r$. Hence, $M_l-m_r=\frac{3}{2}$ and,  by the representability condition, $G_\nu$ is not rectilinear planar, a contradiction. Hence, either $\sigma_\nu+\frac{3}{2}\in [m_l,M_l]$ or $\sigma_\nu-2 \in [m_r,M_r]$. 
	\end{claimproof}
	
	We can construct a rectilinear planar representation $H_{\mu_l}$ of $G_{\mu_l}$ with spirality $\sigma_{\mu_l}$ and a rectilinear planar representation $H_{\mu_r}$ of $G_{\mu_r}$ with spirality $\sigma_{\mu_r}$, based on the following cases:
	\begin{itemize}
		\item \textbf{Case (a):} $\sigma_\nu+\frac{3}{2}\not \in [m_l,M_l]$. 
		This implies that $\sigma_\nu+\frac{1}{2}\in [m_l,M_l]$ and $\sigma_\nu-2\in [m_r,M_r]$, and therefore we set 
		$\sigma_{\mu_l}=\sigma_\nu+\frac{1}{2}$ and $\sigma_{\mu_r}=\sigma_\nu-2$.
		
		\item \textbf{Case (b):} $\sigma_\nu-2\not \in [m_r,M_r]$. 
		This implies that $\sigma_\nu+\frac{3}{2}\in [m_l,M_l]$ and $\sigma_\nu-1\in [m_r,M_r]$, and therefore we set 
		$\sigma_{\mu_l}=\sigma_\nu+\frac{3}{2}$ and $\sigma_{\mu_r}=\sigma_\nu-1$.
		
		\item \textbf{Case (c):} $\sigma_\nu+\frac{3}{2}\in [m_l,M_l]$ and $\sigma_\nu-2\in [m_r,M_r]$.
		We set $\sigma_{\mu_l}=\sigma_\nu+\frac{3}{2}$ and $\sigma_{\mu_r}=\sigma_\nu-2$.
	\end{itemize} 
	
	By the claim proved above, one of the conditions in Case~(a), Case~(b), or Case~(c) is verified. In all the three cases we have $\sigma_{\mu_l}-\sigma_{\mu_r}\in [\frac{5}{2},\frac{7}{2}]$, hence, there exists a rectilinear planar representation $H_\nu$ of $G_\nu$ given the values of $\sigma_{\mu_l}$ and $\sigma_{\mu_r}$ described in the three cases. We have to prove that in the three cases the spirality of $H_\nu$ is $\sigma_\nu$. 
	By Lemma~\ref{le:spirality-P-node-2-children} we have $\sigma_\nu' = \sigma_{\mu_l} - \alpha_{u}^l - \frac{1}{2} \alpha_{v}^l$, where $\sigma_\nu'$ is the spirality of the representation $H_\nu$ of $G_\nu$ given a choice of $\sigma_{\mu_l}$, $\alpha_{v}^l$, and $\alpha_{v}^r$. In Case~(a) we have $\sigma_\nu' = \sigma_\nu+\frac{1}{2} - \alpha_{u}^l - \frac{1}{2} \alpha_{v}^l$; choosing $\alpha_{u}^l=0$ and  $\alpha_{v}^l=1$ we have  $\sigma_\nu' = \sigma_\nu$. In Cases~(b) and~(c)  we have $\sigma_\nu' = \sigma_\nu+\frac{3}{2} - \alpha_{u}^l - \frac{1}{2} \alpha_{v}^l$; choosing $\alpha_{u}^l=1$ and  $\alpha_{v}^l=1$ we have  $\sigma_\nu' = \sigma_\nu$.
	
	\smallskip
	\noindent{\bf Case 2: $\alpha=1$, $\beta=2$, $d=l$.} $G_\nu$ is of type \Pio{3l}{12} and we prove that $I_\nu = [\max\{m_l-\frac{3}{2},m_r+1\}+\frac{1}{2},\min\{M_l-\frac{1}{2}, M_r+2\}-\frac{1}{2}] = [\max\{m_l-1,m_r+\frac{3}{2}\},\min\{M_l-1, M_r+\frac{3}{2}\}$.
	Assume first that $\sigma_\nu$ is the spirality of a rectilinear planar representation of $G_\nu$.  Since for an \Pio{3l}{12} component we have $k_u^l=k_u^r=k_v^l=\frac{1}{2}$ and $k_v^r=1$, by Lemma~\ref{le:spirality-P-node-2-children} we 
	have $\sigma_\nu = \sigma_{\mu_r} + \frac{1}{2}\alpha_{u}^r + \alpha_{v}^r$ and $\sigma_\nu = \sigma_{\mu_l} - \frac{1}{2}\alpha_{u}^l - \frac{1}{2}\alpha_{v}^l$. Since $\alpha_u^l=\alpha_v^l=\alpha_u^r=\alpha_v^r=1$, we have:
	$\sigma_\nu \geq m_r + \frac{3}{2}$, $\sigma_\nu \leq M_r+\frac{3}{2}$, $\sigma_\nu \geq m_l-1$, and $\sigma_\nu \leq M_l-1$.

	We now show that, if $\sigma_\nu \in [\max\{m_l-1,m_r+\frac{3}{2}\},\min\{M_l-1,M_r+\frac{3}{2}\}]$, there exists a rectilinear planar representation of $G_\nu$ with spirality $\sigma_\nu$. 	We have $\sigma_\nu \in [m_l-1, M_l-1]$ and $\sigma_\nu \in [m_r+\frac{3}{2}, M_r+\frac{3}{2}]$. Hence, $\sigma_\nu+1 \in [m_l, M_l]$ and $\sigma_\nu-\frac{3}{2} \in [m_r, M_r]$.  We can construct a rectilinear planar representation $H_{\mu_l}$ of $G_{\mu_l}$ with spirality $\sigma_{\mu_l}=\sigma_\nu+1$ and a rectilinear planar representation $H_{\mu_r}$ of $G_{\mu_r}$ with spirality $\sigma_{\mu_r}=\sigma_\nu-\frac{3}{2}$. Notice that, for this choice, we have $\sigma_{\mu_l}-\sigma_{\mu_r}=\frac{5}{2}$, hence, there exists a rectilinear planar representation $H_\nu$ of $G_\nu$ given the values of $\sigma_{\mu_l}$ and $\sigma_{\mu_r}$. We have to prove that the spirality of $H_\nu$ is $\sigma_\nu$.
	By Lemma~\ref{le:spirality-P-node-2-children} we have $\sigma_\nu' = \sigma_{\mu_l} - \frac{1}{2}\alpha_{u}^l - \frac{1}{2}\alpha_{v}^l$, where $\sigma_\nu'$ is the spirality of the representation $H_\nu$ of $G_\nu$ given a choice of~$\sigma_{\mu_l}$, $\alpha_{v}^l$, and~$\alpha_{v}^r$. Since $\sigma_\nu=\sigma_{\mu_l}-1$, $\alpha_{v}^l=1$, and $\alpha_{v}^r=1$, we have  $\sigma_\nu' = \sigma_\nu+1-1=\sigma_\nu$.
	
	\smallskip\noindent{\bf Case 3: $\alpha=\beta=1$, $d=r$}. Symmetric to Case~1.
	
	\smallskip\noindent{\bf Case 4: $\alpha=1$, $\beta=2$, $d=r$}. Symmetric to Case~2.  
\end{proof}

\begin{lemma}\label{le:P-2-children-support-type3}
	Let $G_\nu$ be a P-node of type \Pin{3dd'} and let $\mu_l$ and $\mu_r$ be its two children. 
	$G_\nu$ is rectilinear planar if and only if $G_{\mu_l}$ and $G_{\mu_r}$ are rectilinear planar for values of spiralities $\sigma_{\mu_l}$ and $\sigma_{\mu_r}$, respectively, such that $\sigma_{\mu_l}-\sigma_{\mu_r}=3$.
\end{lemma}
\begin{proof}
	We distinguish three cases, based on the values of $d$ and $d'$. The proof for the type \Pin{3rl} is symmetric to the proof for the type~\Pin{3lr}.
	
	\noindent{\bf Case 1: $d=d'=l$.}  $G_\nu$ is of type \Pin{3ll} and we prove that $G_\nu$ is rectilinear planar if and only if $G_{\mu_l}$ and $G_{\mu_r}$ are rectilinear planar for values of spiralities $\sigma_{\mu_l}$ and $\sigma_{\mu_r}$ such that $\sigma_{\mu_l}-\sigma_{\mu_r}=3$. For an \Pin{3ll} component we have $k_u^l=k_v^l=\frac{1}{2}$ and $k_u^r=k_v^r=1$. 
	
	If $G_\nu$ is rectilinear planar, we have $\alpha_{u}^l=\alpha_{u}^r =\alpha_{v}^l=\alpha_{v}^r=1$ in any rectilinear planar representation of $G_\nu$. Hence, by Lemma~\ref{le:spirality-P-node-2-children}, for any value of spirality $\sigma_\nu$  we have $\sigma_{\mu_l}-\sigma_{\mu_r}=\frac{1}{2}\alpha_{u}^l+\frac{1}{2}\alpha_{v}^l+\alpha_{u}^r+\alpha_{v}^r=3$. 
	
	Suppose vice versa that $G_{\mu_l}$ and $G_{\mu_r}$ are rectilinear planar for values of spirality $\sigma_{\mu_l}$ and $\sigma_{\mu_r}$ such that $\sigma_{\mu_l}-\sigma_{\mu_r}=3$. We show that $G_\nu$ admits a rectilinear planar representation $H_\nu$. To define $H_\nu$, we combine in parallel the two rectilinear planar representations of $G_{\mu_l}$ and $G_{\mu_r}$ and assign values  $\alpha_u^l=\alpha_v^l=\alpha_u^r=\alpha_v^r=1$. 
	Let $u'$ and $v'$ be the alias vertices of $G_{\mu_l}$ that are in $G_\nu$. Any cycle $C$ that goes through $u$ and $v$ also passes through $u'$ and $v'$. We show that the number of $90^\circ$ angles minus the number of $270^\circ$ angles in the interior of $C$ is equal to four.
	Vertices $u'$ and $v'$ split $C$ into two paths $\pi_l$ and $\pi_r$. Suppose to visit $C$ clockwise. The number of right turns minus left turns along $\pi_l$ while going from $u'$ to $v'$ equals $\sigma_{\mu_l}+1$.  The number of right turns minus left turns along $\pi_r$ while going from $v'$ to $u'$ equals $-\sigma_{\mu_r}$. The sum of these two values corresponds to the number of $90^\circ$ angles minus the number of $270^\circ$ angles in the interior of $C$ at the vertices of $\pi_l$. Hence, $\sigma_{\mu_l}+1-\sigma_{\mu_r}=3+1=4$.
	
	Also, any other cycle not passing through $u$ and $v$ is an orthogonal polygon because it belongs to a rectilinear planar representation of either $G_{\mu_l}$ (with spirality $\sigma_{\mu_l}$) or $G_{\mu_r}$ (with spirality $\sigma_{\mu_r}$). 
	
	\noindent{\bf Case 2: $d=d'=r$.} Symmetric to Case 1, observing that $k_u^r=k_v^r=\frac{1}{2}$ and $k_u^l=k_v^l=1$.
	
	\noindent{\bf Case 3: $d=l$, $d'=r$.} $G_\nu$ is of type \Pin{3lr} and we prove that $G_\nu$ is rectilinear planar if and only if $G_{\mu_l}$ and $G_{\mu_r}$ are rectilinear planar for values of spiralities $\sigma_{\mu_l}$ and $\sigma_{\mu_r}$ such that $\sigma_{\mu_l}-\sigma_{\mu_r}=3$. For an \Pin{3lr} component we have $k_u^r=k_v^l=\frac{1}{2}$ and $k_u^l=k_v^r=1$. 
	
	If $G_\nu$ is rectilinear planar, we have $\alpha_{u}^l=\alpha_{u}^r =\alpha_{v}^l=\alpha_{v}^r=1$ in any rectilinear planar representation of $G_\nu$. Hence, by Lemma~\ref{le:spirality-P-node-2-children}, for any value of spirality $\sigma_\nu$  we have $\sigma_{\mu_l}-\sigma_{\mu_r}=\alpha_{u}^l+\frac{1}{2}\alpha_{v}^l+\frac{1}{2}\alpha_{u}^r+\alpha_{v}^r=3$.
	
	Suppose vice versa that $G_{\mu_l}$ and $G_{\mu_r}$ are rectilinear planar for values of spirality $\sigma_{\mu_l}$ and $\sigma_{\mu_r}$ such that $\sigma_{\mu_l}-\sigma_{\mu_r}=3$. We show that $G_\nu$ admits a rectilinear planar representation $H_\nu$. To define $H_\nu$, we combine in parallel the two rectilinear planar representations of $G_{\mu_l}$ and $G_{\mu_r}$ and assign values  $\alpha_u^l=\alpha_v^l=\alpha_u^r=\alpha_v^r=1$. 
	%
	Let $v'$ be the alias vertex of the pole $v$ of $G_{\mu_l}$ such that $v'$ is along an edge of $G_\nu$. Similarly, let $u'$ be the alias vertex of the pole $u$ of $G_{\mu_r}$ such that $u'$ is along an edge of $G_\nu$. Any cycle $C$ that goes through $u$ and $v$ also passes through $u'$ and $v'$. We show that the number of $90^\circ$ angles minus the number of $270^\circ$ angles in the interior of $C$ is equal to four.
	Vertices $u'$ and $v'$ split $C$ into two paths $\pi_l$ and $\pi_r$. Suppose to visit $C$ clockwise. The number of right turns minus left turns along $\pi_l$ while going from $u'$ to $v'$ equals $\sigma_{\mu_l}+\frac{1}{2}$.  The number of right turns minus left turns along $\pi_r$ while going from $v'$ to $u'$ equals $-\sigma_{\mu_r}+\frac{1}{2}$. The sum of these two values corresponds to the number of $90^\circ$ angles minus the number of $270^\circ$ angles in the interior of $C$ at the vertices of $\pi_l$. Hence, $\sigma_{\mu_l}+\frac{1}{2}-\sigma_{\mu_r}+\frac{1}{2}=3+\frac{1}{2}+\frac{1}{2}=4$.
	
	Also, any other cycle not passing through $u$ and $v$ is an orthogonal polygon because it belongs to a rectilinear planar representation of either $G_{\mu_l}$ (with spirality $\sigma_{\mu_l}$) or $G_{\mu_r}$ (with spirality $\sigma_{\mu_r}$). 
\end{proof}

\begin{lemma}\label{le:P-2-children-representability-type3}
	Let $\nu$ be a P-node of type \Pin{3dd'} with children $\mu_l$ and $\mu_r$.
	Suppose that $G_{\mu_l}$ and $G_{\mu_r}$ are rectilinear planar with representability intervals $I_{\mu_l}=[m_l, M_l]$ and $I_{\mu_r} = [m_r, M_r]$, respectively. Graph $G_\nu$ is rectilinear planar if and only if $3 \in [m_l-M_r,M_l-m_r]$. Also, if this representability condition holds then the representability interval of $G_\nu$ is $I_\nu = [\max\{m_l-1,m_r+2\}-\frac{\phi(d)+\phi(d')}{2},\min\{M_l-1, M_r+2\}-\frac{\phi(d)+\phi(d')}{2}]$, where $\phi(\cdot)$ is a function such that $\phi(r)=1$ and $\phi(l)=0$.
\end{lemma}
\begin{proof}
	We prove the correctness of the representability condition and
	the validity of the representability interval.
	
	\smallskip\noindent\textsf{Representability condition.}
	Suppose that $G_\nu$ is rectilinear planar. By Lemma~\ref{le:P-2-children-support-type3}, $G_{\mu_l}$ and $G_{\mu_r}$ admit spiralities $\sigma_{\mu_l}$ and $\sigma_{\mu_r}$, respectively, such that $\sigma_{\mu_l}-\sigma_{\mu_r} = 3$. Hence, $m_l-M_r \leq \sigma_{\mu_l}-\sigma_{\mu_r} \leq 3$ and $M_l-m_r \geq \sigma_{\mu_l}-\sigma_{\mu_r} \geq 3$, i.e.,  $3 \in [m_l-M_r,M_l-m_r]$.
	Suppose, vice versa that $3 \in [m_l-M_r,M_l-m_r]$. By hypothesis $G_{\mu_l}$ (resp. $G_{\mu_r}$) is rectilinear planar for every value of spirality in the interval $[m_l,M_l]$ (resp. $[m_r,M_r]$). This implies that there exist rectilinear planar representations for $G_{\mu_l}$ and $G_{\mu_r}$ with spiralities $\sigma_{\mu_l} \in [m_l,M_l]$ and $\sigma_{\mu_r} \in [m_r,M_r]$ such that $\sigma_{\mu_l}-\sigma_{\mu_r} = 3$. Hence, by Lemma~\ref{le:P-2-children-support-type3} $G_\nu$ is rectilinear planar.
	
	\smallskip\noindent\textsf{Representability interval.} We distinguish three cases, based on the values of $d$ and $d'$. Note that a possible forth case for the type~\Pin{3rl} is symmetric to the case for the type~\Pin{3lr}. 
	
	\noindent{\bf Case 1: $d=d'=l$.} $G_\nu$ is of type \Pin{3ll} and we prove that $I_\nu = [\max\{m_l-1,m_r+2\},\min\{M_l-1, M_r+2\}]$.
	Assume first that $\sigma_\nu$ is the spirality of a rectilinear planar representation of $G_\nu$.  Since for an \Pin{3ll} component we have $k_u^l=k_v^l=\frac{1}{2}$ and $k_u^r=k_v^r=1$, by Lemma~\ref{le:spirality-P-node-2-children} we 
	have $\sigma_\nu = \sigma_{\mu_r} + \alpha_{u}^r + \alpha_{v}^r$ and $\sigma_\nu =\sigma_{\mu_l} - \frac{1}{2}\alpha_{u}^l - \frac{1}{2}\alpha_{v}^l$. Since $\alpha_u^l=\alpha_v^l=\alpha_u^r=\alpha_v^r=1$, we have:
	$\sigma_\nu \geq m_r + 2$, $\sigma_\nu \leq M_r+2$, $\sigma_\nu \geq m_l-1$, and $\sigma_\nu \leq M_l-1$. 
	
	We now show that, if $\sigma_\nu \in [\max\{m_l-1,m_r+2\},\min\{M_l-1,M_r+2\}]$, there exists a rectilinear planar representation of $G_\nu$ with spirality $\sigma_\nu$. We have $\sigma_\nu \in [m_l-1, M_l-1]$ and $\sigma_\nu \in [m_r+2, M_r+2]$. Hence, $\sigma_\nu+1 \in [m_l, M_l]$ and $\sigma_\nu-2 \in [m_r, M_r]$. We can construct a rectilinear planar representation $H_{\mu_l}$ of $G_{\mu_l}$ with spirality $\sigma_{\mu_l}=\sigma_\nu+1$ and a rectilinear planar representation $H_{\mu_r}$ of $G_{\mu_r}$ with spirality $\sigma_{\mu_r}=\sigma_\nu-2$. Note that, for this choice, we have $\sigma_{\mu_l}-\sigma_{\mu_r}=3$, hence, there exists a rectilinear planar representation $H_\nu$ of $G_\nu$ given the values of $\sigma_{\mu_l}$ and $\sigma_{\mu_r}$. We have to prove that the spirality of $H_\nu$ is $\sigma_\nu$.
	By Lemma~\ref{le:spirality-P-node-2-children} we have $\sigma_\nu' = \sigma_{\mu_l} - \frac{1}{2} \alpha_{u}^l - \frac{1}{2}\alpha_{v}^l$, where $\sigma_\nu'$ is the spirality of the representation $H_\nu$ of $G_\nu$ given a choice of $\sigma_{\mu_l}$, $\alpha_{v}^l$, and $\alpha_{u}^l$. Since $\sigma_\nu=\sigma_{\mu_l}-1$, $\alpha_{u}^l=1$, and $\alpha_{v}^l=1$, we have  $\sigma_\nu' = \sigma_\nu+1-\frac{1}{2}-\frac{1}{2}=\sigma_\nu$.
	
	\smallskip\noindent{\bf Case 2: $d=d'=r$.} Symmetric to Case 1.
	
	\smallskip\noindent{\bf Case 3: $d=l$, $d'=r$.}  $G_\nu$ is of type \Pin{3lr} and we prove that $I_\nu = [\max\{m_l-1,m_r+2\}-\frac{1}{2},\min\{M_l-1, M_r+2\}-\frac{1}{2}]=[\max\{m_l-\frac{3}{2},m_r+\frac{3}{2}\},\min\{M_l-\frac{3}{2}, M_r+\frac{3}{2}\}]$.
	Assume first that $\sigma_\nu$ is the spirality of a rectilinear planar representation of $G_\nu$.  Since for an \Pin{3lr} component we have $k_u^r=k_v^l=\frac{1}{2}$ and $k_u^l=k_v^r=1$, by Lemma~\ref{le:spirality-P-node-2-children} we 
	have $\sigma_\nu = \sigma_{\mu_r} + \frac{1}{2}\alpha_{u}^r + \alpha_{v}^r$ and $\sigma_\nu =\sigma_{\mu_l} - \alpha_{u}^l - \frac{1}{2}\alpha_{v}^l$. Since $\alpha_u^l=\alpha_v^l=\alpha_u^r=\alpha_v^r=1$, we have:
	$\sigma_\nu \geq m_r + \frac{3}{2}$, $\sigma_\nu \leq M_r+\frac{3}{2}$, $\sigma_\nu \geq m_l-\frac{3}{2}$ and $\sigma_\nu \leq M_l-\frac{3}{2}$.
	
	We now show that, if $\sigma_\nu \in [\max\{m_l-\frac{3}{2},m_r+\frac{3}{2}\},\min\{M_l-\frac{3}{2},M_r+\frac{3}{2}\}]$, there exists a rectilinear planar representation of $G_\nu$ with spirality $\sigma_\nu$. 	We have $\sigma_\nu \in [m_l-\frac{3}{2}, M_l-\frac{3}{2}]$ and $\sigma_\nu \in [m_r+\frac{3}{2}, M_r+\frac{3}{2}]$. Hence, $\sigma_\nu+\frac{3}{2} \in [m_l, M_l]$ and $\sigma_\nu-\frac{3}{2} \in [m_r, M_r]$. We can construct a rectilinear planar representation $H_{\mu_l}$ of $G_{\mu_l}$ with spirality $\sigma_{\mu_l}=\sigma_\nu+\frac{3}{2}$ and a rectilinear planar representation $H_{\mu_r}$ of $G_{\mu_r}$ with spirality $\sigma_{\mu_r}=\sigma_\nu-\frac{3}{2}$. Notice that, for this choice, we have $\sigma_{\mu_l}-\sigma_{\mu_r}=3$, hence, there exists a rectilinear planar representation $H_\nu$ of $G_\nu$ given the values of $\sigma_{\mu_l}$ and $\sigma_{\mu_r}$. We have to prove that the spirality of $H_\nu$ is $\sigma_\nu$.
	By Lemma~\ref{le:spirality-P-node-2-children}, $\sigma_\nu' = \sigma_{\mu_l} - \alpha_{u}^l - \frac{1}{2}\alpha_{v}^l$, where $\sigma_\nu'$ is the spirality of the representation $H_\nu$ of $G_\nu$ given a choice of $\sigma_{\mu_l}$, $\alpha_{v}^l$, and $\alpha_{u}^l$. Since $\sigma_\nu=\sigma_{\mu_l}-\frac{3}{2}$, $\alpha_{u}^l=1$, and $\alpha_{v}^l=1$, we have $\sigma_\nu' = \sigma_\nu+\frac{3}{2}-1-\frac{1}{2}=\sigma_\nu$.
\end{proof}


\subsection{Representability condition for the root}\label{se:root-rect-charact} 
To finally achieve a characterization of rectilinear series-parallel graphs we need to consider the representability condition that must be verified at the level of the root, when the reference edge is not a dummy edge. Denote by $e=(u,v)$ the reference edge of $G$ and let $\rho$ be the root of $T$ with respect to $e$. Let $\eta$ be the child of $\rho$ that does not correspond to $e$, and let $u'$ and $v'$ be the alias vertices associated with the poles $u$ and $v$ of $G_\eta$. Suppose that $G_\eta$ is rectilinear planar with representability interval $I_\eta$. 

We say that $G$ satisfies the \emph{root condition} if $I_\eta \cap \Delta_\rho \neq \emptyset$, where $\Delta_\rho$ is defined as follows: $(i)$ $\Delta_\rho=[2,6]$ if $u'$ coincides with $u$ and $v'$ coincides with $v$; $(ii)$ $\Delta_\rho=[3,5]$ if exactly one of $u'$ and $v'$ coincides with $u$ and $v$, respectively; $(iii)$ $\Delta_\rho=4$ if none of $u'$ and $v'$ coincides with $u$ and $v$.

\begin{lemma}\label{le:root-condition}
	Let $e=(u,v)$ be the reference edge of $G$ and let $\rho$ be the root of $T$ with respect to $e$.
	Let $\eta$ be the child of $\rho$ that does not correspond to $e$. Suppose that $G_\eta$ is rectilinear planar with representability interval $I_\eta$. $G$ is rectilinear planar if and only if it satisfies the root condition. Also, if $G$ satisfies the root condition, it admits a rectilinear planar representation $H$ for any value of spirality $\sigma_\eta$ of $H_\eta$ such that $\sigma_\eta \in I_\eta \cap \Delta_\rho$, where $H_\eta$ is the restriction of $H$ to $G_\eta$.     
\end{lemma}
\begin{proof}
	Let $f_{\mathrm{int}}$ be the internal face of $G$ incident to $e$. Observe that $u$ and $v$ are the poles of $G_\eta$. Let $u'$ and $v'$ be the alias vertices associated with $u$ and with $v$, respectively. $H$ is a rectilinear planar representation of $G$ if and only if the following two conditions hold: The restriction $H_\eta$ of $H$ to $G_\eta$ is a rectilinear planar representation; the number $A$ of right turns minus left turns of any simple cycle of $G$ in $H$ containing $e$ and traversed clockwise in $H$ is equal to $4$. 
	We have $A=\sigma_\eta+\alpha_{u'}+\alpha_{v'}$, where: $\sigma_\eta$ is the spirality of $H_\eta$; for $w\in \{u',v'\}$, $\alpha_w=1$, $\alpha_w=0$, and $\alpha_w=-1$ if the angle formed by $w$ in $f_{\mathrm{int}}$ is equal to $90^o$, $180^o$, or $270^o$, respectively. 
	
	According to the definition of root condition, there are three cases to consider: 
	$(i)$ $\Delta_\rho=[2,6]$, $(ii)$ $\Delta_\rho=[3,5]$, and $(iii)$ $\Delta_\rho=4$. 
	Consider Case $(i)$.  Since in this case the alias vertices coincide with the poles, we have $\alpha_{u'} \in [-1,1]$, $\alpha_{v'} \in [-1,1]$, and hence $\alpha_{u'}+\alpha_{v'} \in [-2,2]$.  
	If $G$ is rectilinear planar, we have that $A = \sigma_\eta+\alpha_{u'}+\alpha_{v'}=4$ for some $\sigma_\eta \in I_\eta$ and for $\alpha_{u'}+\alpha_{v'} \in [-2,2]$. Hence, $\sigma_\eta=4-\alpha_{u'}-\alpha_{v'} \in [2,6]$, i.e., the root condition $I_\eta \cap \Delta_\rho \neq \emptyset$ holds.    
	Suppose vice versa that the root condition $I_\eta \cap \Delta_\rho \neq \emptyset$ holds. For any value $\sigma_\eta \in I_\eta \cap \Delta_\rho$ there exists a rectilinear planar representation of $H_\eta$ of $G_\eta$ with spirality $\sigma_\eta$. Also, since $\Delta_\rho=[2,6]$, we have that  $4-\sigma_\eta\in [-2,2]$, and therefore, for any possible choice of $\sigma_\eta \in I_\eta \cap \Delta_\rho$, we can suitably choose $\alpha_{u'}$ and $\alpha_{v'}$ such that $\alpha_{u'}+\alpha_{v'}=4-\sigma_\eta$, i.e., $A = \sigma_\eta+\alpha_{u'}+\alpha_{v'}=4$. It follows that $G$ is rectilinear planar and it admits a rectilinear planar representation for any value $\sigma_\eta \in I_\eta \cap \Delta_\rho$. 
	
	Cases~$(ii)$ and~$(iii)$ are proved analogously; in Case~$(ii)$  $\alpha_{u'}+\alpha_{v'}\in [-1,1]$ and in Case~$(iii)$ $\alpha_{u'}+\alpha_{v'}=0$.  
\end{proof}

\renewcommand{\arraystretch}{1.5}
\begin{table}[tb]
	\centering
	\caption{Representability conditions and intervals for the different types of nodes. In the formulas of the table, the term $\gamma$ is such that $\gamma=\lambda +\beta-2$, and $\phi(\cdot)$ is a function such that $\phi(r)=1$ and $\phi(l)=0$.}\label{ta:representability}
	\begin{tabular}{|l |c |}
		\hline
		\rowcolor{antiquewhite}\multicolumn{2}{|c|}{\bf Q$^*$-node (the pertinent graph is a path of length $\ell$)}\\
		\hline
		{Representability Condition} & true \\
		{Representability Interval} & $[-\ell+1, \ell-1]$\\
		
		\hline
		\rowcolor{antiquewhite}\multicolumn{2}{|c|}{\bf S-node with $h$ children $\mu_i$ such that $I_{\mu_i}=[m_i,M_i]$ $(i = 1, \dots, h)$}\\
		\hline
		{Representability Condition} & true \\
		{Representability Interval} & $I_\nu=[m,M]=[\sum_{i=1}^hm_i,\sum_{i=1}^hM_i]$\\
		
		\hline
		\rowcolor{antiquewhite}\multicolumn{2}{|c|}{\bf P-node with three children $\mu_l$, $\mu_c$, $\mu_r$ where $I_{\mu_i}=[m_i,M_i]$ ($i \in \{l,c,r\}$)}\\
		\hline
		{Representability Condition} & $[m_l-2,M_l-2] \cap [m_c,M_c] \cap [m_r+2,M_r+2] \neq \emptyset$\\
		{Representability Interval} & $I_\nu=[m,M]=[\max\{m_l-2,m_c,m_r+2\},\min\{M_l-2,M_c,M_r+2\}]$\\
		
		\hline
		\rowcolor{antiquewhite}\multicolumn{2}{|c|}{\bf P-node with two children $\mu_l$ and $\mu_r$ where $I_{\mu_i}=[m_i,M_i]$ ($i \in \{l,r\}$) $-$ \Pio{2}{\lambda\beta}}\\
		\hline
		{Representability Condition} & $[m_l-M_r, M_l-m_r] \cap [2, 4-\gamma] \neq \emptyset$ \\
		{Representability Interval} & $I_\nu=[m,M]=[\max\{m_l-2,m_r\}+\frac{\gamma}{2}, \min\{M_l, M_r+2\}-\frac{\gamma}{2}]$\\
		
		\hline
		\rowcolor{antiquewhite}\multicolumn{2}{|c|}{\bf P-node with two children $\mu_l$ and $\mu_r$ where $I_{\mu_i}=[m_i,M_i]$ ($i \in \{l,r\}$) $-$ \Pio{3d}{\lambda\beta}}\\
		\hline
		{Representability Condition} & $[m_l-M_r, M_l-m_r] \cap [\frac{5}{2},\frac{7}{2}-\gamma] \neq \emptyset$ \\
		{Representability Interval} & $I_\nu=[m,M]=[\max\{m_l-\frac{3}{2},m_r+1\}+\frac{\gamma-\phi(d)}{2},\min\{M_l-\frac{1}{2}, M_r+2\}-\frac{\gamma+\phi(d)}{2}]$\\
		
		\hline
		\rowcolor{antiquewhite}\multicolumn{2}{|c|}{\bf P-node with two children $\mu_l$ and $\mu_r$ where $I_{\mu_i}=[m_i,M_i]$ ($i \in \{l,r\}$) $-$ \Pin{3dd'}}\\
		\hline
		{Representability Condition} & $3 \in [m_l-M_r,M_l-m_r]$ \\
		{Representability Interval} & $I_\nu=[m,M]=[\max\{m_l-1,m_r+2\}-\frac{\phi(d)+\phi(d')}{2},\min\{M_l-1, M_r+2\}-\frac{\phi(d)+\phi(d')}{2}]$\\
		
		\hline
		\rowcolor{antiquewhite}\multicolumn{2}{|c|}{\bf P$^r$-node (the root $\rho$)}\\
		\hline
		{Root condition} & $I_\eta \cap \Delta_\rho \neq \emptyset$\\
		\hline
	\end{tabular}
\end{table}

The next theorem summarizes the main result of this section.

\begin{theor}\label{th:characterization}
	Let $G$ be a plane series-parallel 4-graph and let $T$ be an SPQ$^*$-tree of $G$. Graph~$G$ is rectilinear planar if and only if, for every node of $T$ the corresponding representability condition of Table~\ref{ta:representability} is satisfied.
\end{theor}


\section{Bend-Minimization Algorithm: Overview}\label{se:bend-min-overv}

Let $G$ be a plane series-parallel 4-graph. If $G$ is biconnected let $e$ be any edge
on the external face of $G$; otherwise, we add a dummy edge $e$ that makes it biconnected. Let $T$ be the SPQ$^*$-tree of $G$ with respect to $e$. 
Our bend-minimization algorithm works in two phases. It first visits $T$ bottom-up (in post order) to determine the number of bends of a bend-minimum orthogonal representation of $G$. Then it visits $T$ top-down to construct such an orthogonal representation.

When a node $\nu$ is considered in the bottom-up visit, the algorithm assigns to $\nu$ a \emph{budget}~$b_\nu$ of bends. This budget corresponds to the minimum number of extra bends that must be added to the budgets of the children of $\nu$ to realize an orthogonal representation of $G_\nu$. In other words, $b_\nu$ can be regarded as the minimum number of extra subdivision vertices that must be inserted along the edges of $G_\nu$ (besides those already inserted for the children of $\nu$) to make it rectilinear planar. The budget $b_\nu$ is larger than zero if and only if the representability condition of the rectilinear planarity testing for $\nu$ is not satisfied. Hence, according to \cref{ta:representability}, $b_\nu=0$ if $\nu$ is a Q$^*$- or an S-node, while it can be positive if $\nu$ is a P-node or the root of $T$. For instance, for the graph of \cref{fi:graph-minbend} and the tree~$T$ of \cref{fi:spq-tree-minbend}, the first component that requires some bends in the bottom-up visit of $T$ is $G_{\nu_3}$, namely  $b_{\nu_3}=3$; two more bends are required at the root level, i.e., $b_{\rho}=2$. When $b_\nu > 0$, a crucial and non-trivial aspect is how to efficiently compute $b_\nu$. The other key aspect is how to succinctly describe the set $I'_\nu$ of spirality values that a rectilinear representation of $G_\nu$ can take, by considering all possible distributions of the $b_\nu$ subdivision vertices along its edges. We will show that $I'_\nu$ is still an interval, which allows us to represent it in~$O(1)$ space.

\cref{se:budgets} describes how to compute the budgets $b_\nu$ and the sets $I'_\nu$ in the bottom-up visit of $T$, and it proves the optimality of the solution. \cref{se:summary} describes the top-down visit and summarizes our main result.

\section{Budgets and Optimality}\label{se:budgets}


In the following we denote by $m$ and $M$ the minimum and maximum values of the representability interval $I_\nu$ of $\nu$ when $G_\nu$ is rectilinear planar, as defined in \cref{ta:representability}. 
Also, since when we visit $\nu$, all its children have already been visited and have received their own budget of bends (i.e., of subdivision vertices for the corresponding component), we will simply assume that each child of $\nu$ is rectilinear planar. 

As observed, if $\nu$ is either a Q$^*$-node or an S-node, $b_\nu=0$. Hence, we assume that $\nu$ is a P-node. A child $\mu$ of a (non-root) P-node $\nu$ is either a Q$^*$- or an S-node. To compute $b_\nu$ and~$I'_\nu$ we define the concept of \emph{exposed edge} of $\mu$. If $\mu$ is a Q$^*$-node, every edge of $G_\mu$ is an exposed edge of $\mu$ (and of $G_\mu$). If $\mu$ is an S-node with at least one Q$^*$-node child $\mu'$, every edge of $G_\mu$ that belongs to $G_{\mu'}$ is an exposed edge of $\mu$ (and of $G_\mu$). Else, $\mu$ is an S-node that has no exposed edge. \cref{le:P-exposed-edges-support} states a crucial property. It implies that when we must insert some subdivision vertices in an S-component $\mu$ that is a child of $\nu$, these vertices can always be added along an exposed edge of $\mu$, if one exists.

\begin{lemma}\label{le:P-exposed-edges-support}
	Let $\mu$ be an S-node such that $G_\mu$ is rectilinear planar and $\mu$ has an exposed edge $e$. Let $H_\mu$ be an orthogonal representation of $G_\mu$ having $b>0$ bends. There exists an orthogonal representation $H'_\mu$ of $G_\mu$ with $b$ bends such that: (i) all the $b$ bends of $H'_\mu$ lie on~$e$; (ii) $\sigma(H'_\mu) = \sigma(H_\mu)$.
\end{lemma}
\begin{proof}
	Let $u$ and $v$ be the poles of $G_\mu$. Any path from $u$ to $v$ inside $G_\mu$ contains the exposed edge $e$. Consider an orthogonal representation $H_\mu$ of $G_\mu$ with $b > 0$ bends. Let $e' \neq e$ be an edge with at least one bend in $H_\mu$. Let $P^{uv}$ be any simple path from $u$ to $v$ of $H_\mu$ passing through $e'$. Suppose, without loss of generality, that the bend on $e'$ corresponds to a right turn along $P^{uv}$ while going from $u$ to $v$. Since by hypothesis $G_\mu$ is rectilinear planar, we can derive from $H_\mu$ a different orthogonal representation $H''_\mu$ with $b$ bends by simply moving the right bend from $e'$ to $e$, i.e., by inserting a right bend along $e$ and by straightening the right bend of $H_\mu$ along $e'$. With this transformation, the number of right and left turns along $P^{uv}$ is the same in $H''_\mu$ and $H_\mu$, and the angles at $u$ and $v$ in the two representations are also the same. This implies that $\sigma(H''_\mu)=\sigma(H_\mu)$. By repeatedly applying this transformation on $H''_\mu$ until all the $b$ bends of $H_\mu$ are moved on $e$ we get the desired representation~$H'_\mu$.  
\end{proof}

\begin{figure}[t]
	\centering
	\begin{subfigure}{.2\columnwidth}
		\centering
		\includegraphics[page=1,width=\columnwidth]{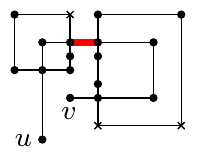}
		\subcaption{\centering $H_\mu$}
		\label{fi:exposed-edges-a}
	\end{subfigure}
	\hfil
	\begin{subfigure}{.2\columnwidth}
		\centering
		\includegraphics[page=2,width=\columnwidth]{exposed-edges.pdf}
		\subcaption{\centering $H_\mu'$}
		\label{fi:exposed-edges-b}
	\end{subfigure}
	\caption{Illustration of  \cref{le:P-exposed-edges-support}. (a)~An orthogonal representation $H_\mu$ of an S-component with spirality $3$ and having $3$ bends. An exposed edge, shown as a red thick segment. (b)~A different orthogonal representation $H'_\mu$ of the same component having the same spirality and number of bends as $H_\mu$, but such that all the bends are along the red thick exposed edge.}\label{fi:exposed-edges}
\end{figure}

An illustration of \cref{le:P-exposed-edges-support} is given in \cref{fi:exposed-edges}. In \cref{fi:exposed-edges-a} an orthogonal representation $H_\mu$ of an S-component is shown. It has spirality $3$ and $3$ bends. In \cref{fi:exposed-edges-b} a different orthogonal representation $H'_\mu$ of the same component is given, having the same spirality and number of bends as $H_\mu$, where all the bends are along an exposed edge.

\smallskip Observe that, if $\nu$ is a P-node with three children, each of them has an exposed edge (as the poles of $\nu$ have degree at most four). If $\nu$ has two children, it might have a child without exposed edges only if $\nu$ is of type \Pin{3ll} or \Pin{3rr} (see \cref{fi:P-node-types}).
\cref{sse:bottom-up-3-children} and \cref{sse:bottom-up-2-children} focus on the budget of P-nodes with three children and on the budget of P-nodes with two children, respectively. \cref{sse:bottom-up-root} concentrates on the budget of the root. 



\subsection{Budget of P-nodes with three children}\label{sse:bottom-up-3-children}
\cref{le:P-3-exposed-edges} handles the case of a P-node $\nu$ with three children $\mu_l$, $\mu_c$, and $\mu_r$ such that the corresponding components are rectilinear planar, while $G_\nu$ is not rectilinear planar. Denote by $I_{\mu_l}=[m_l,M_l]$, $I_{\mu_c}=[m_c,M_c]$, and $I_{\mu_r}=[m_r,M_r]$ the representability intervals of $\mu_l$, $\mu_c$, and $\mu_r$, respectively. Since $G_\nu$ is not rectilinear planar, the representability condition for $\nu$ is violated, i.e., $[m_l-2, M_l-2] \cap [m_c, M_c] \cap [m_r+2, M_r+2] = \emptyset$. Rename the three intervals involved in the representability condition as $[\ov{m}_x,\ov{M}_x]$, $[\ov{m}_y,\ov{M}_y]$, and $[\ov{m}_z,\ov{M}_z]$, where $x,y,z \in \{l,c,r\}$ and $x \neq y \neq z$, in such a way that $\ov{m}_z = \max\{m_l-2, m_c, m_r+2\}$ and $\ov{M}_x = \min\{M_l-2,M_c,M_r+2\}$. Namely, $z=l$ if $\ov{m}_z = m_l-2$, $z=c$ if $\ov{m}_z = m_c$, and $z=r$ if $\ov{m}_z = m_r+2$. Similarly, $x=l$ if $\ov{M}_x = M_l-2$, $x=c$ if $\ov{M}_x = M_c$, and $x=r$ if $\ov{M}_x = M_r+2$. See \cref{fi:3children-minbend-example-notation} for an example. The following holds. 

\begin{figure}[tb]
	\centering
	\begin{subfigure}{1\columnwidth}
		\centering
		\includegraphics[width=0.9\columnwidth]{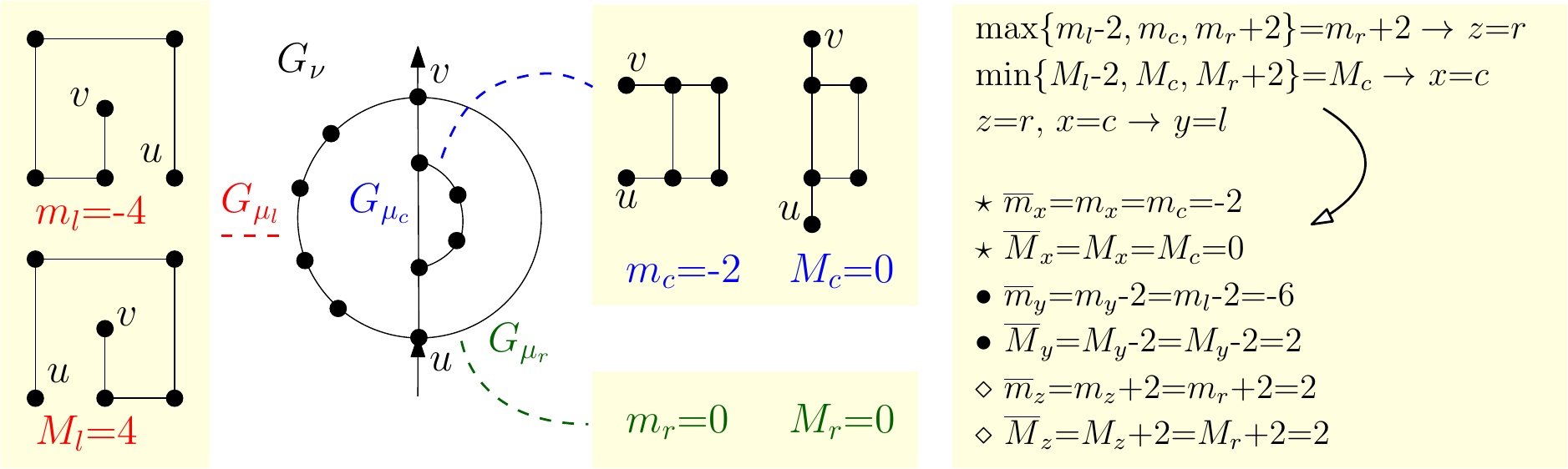}
		\subcaption{\centering}
		\label{fi:3children-minbend-example-notation}
	\end{subfigure}
	\hfil
	\begin{subfigure}{0.8\columnwidth}
		\centering
		\includegraphics[width=1\columnwidth]{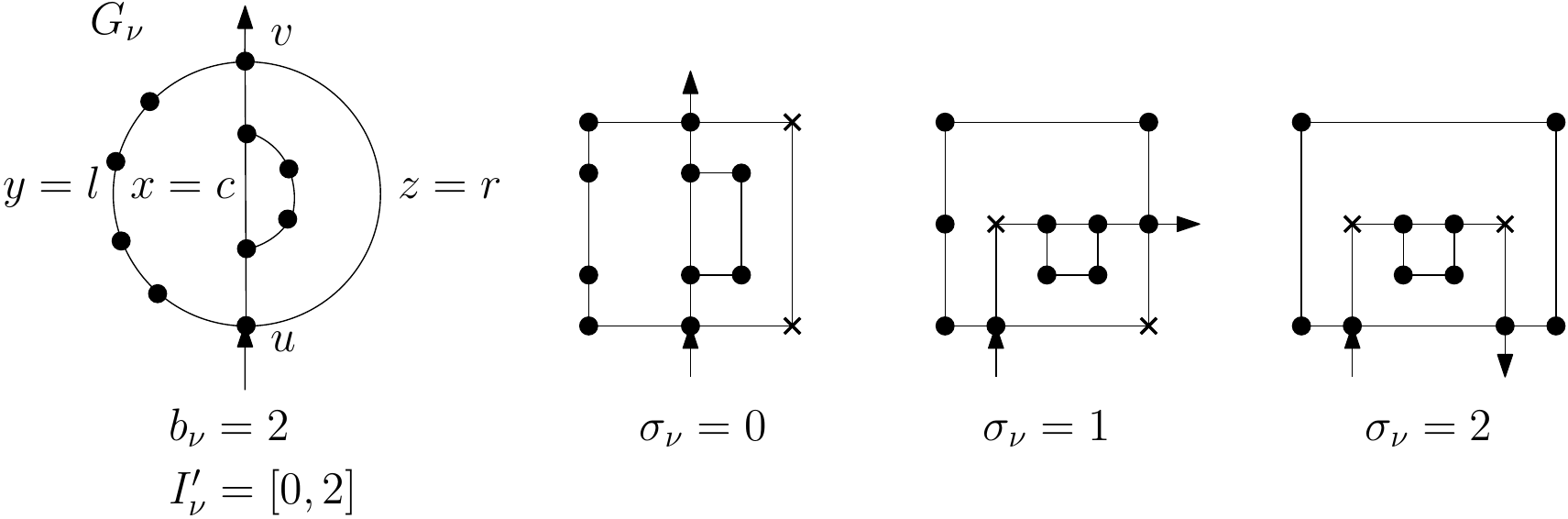}
		\subcaption{\centering}
		\label{fi:3children-minbend-example}
	\end{subfigure}
	\caption{(a) Illustration of the notation $\ov{m}_t$, $\ov{M}_t$ for $t\in \{x,y,z\}$. (b) Illustration of the statement of \cref{le:P-3-exposed-edges} for the graph $G_\nu$. By Property~(i) of the lemma $b_\nu = \ov{m}_z - \ov{M}_x =2$ and by Property~(ii) $I'_\nu = [\max\{\ov{M}_x,\ov{m}_y\}, \min\{\ov{m}_z,\ov{M}_y\}]=[0,2]$.}\label{fi:3children-minbend}
\end{figure}


\begin{proposition}\label{pr:3-intervals}
	$[\ov{m}_x,\ov{M_x}]$ and $[\ov{m}_z,\ov{M}_z]$ are disjoint, with $\ov{m}_z > \ov{M}_x$. 
\end{proposition}
\begin{proof}
	Suppose for a contradiction that $\ov{m}_z \leq \ov{M}_x$. Since $\ov{M}_x$ is the minimum of the three maxima, we have that $\ov{M}_y$ and $\ov{M}_z$ are to the right of $\ov{M}_x$. Also, since $\ov{m}_z$ is the maximum of three minima, we have that $\ov{m}_y$ and $\ov{m}_x$ are to the left of $\ov{m}_z$. Hence the three intervals $[\ov{m}_x,\ov{M}_x]$, $[\ov{m}_y,\ov{M}_y]$, and $[\ov{m}_z,\ov{M}_z]$ share the interval $[\ov{m}_z,\ov{M}_x]$, which contradicts the fact that $G_\nu$ is not rectilinear planar.
\end{proof}



\begin{lemma}\label{le:P-3-exposed-edges}
	Let $\nu$ be a P-node with three children $\mu_l$, $\mu_c$, and $\mu_r$. Let $G_{\mu_l}$, $G_{\mu_c}$, and $G_{\mu_r}$ be rectilinear planar with representability intervals $I_{\mu_l}=[m_l,M_l]$, $I_{\mu_c}=[m_c,M_c]$, and $I_{\mu_r}=[m_r,M_r]$, respectively. If $G_\nu$ is not rectilinear planar then: (i)~the budget for $\nu$ is $b_\nu = \ov{m}_z - \ov{M}_x$; and (ii)~the interval of spirality values for an orthogonal representation of $G_\nu$ with $b_\nu$ bends is $I'_\nu = [\max\{\ov{M}_x,\ov{m}_y\}, \min\{\ov{m}_z,\ov{M}_y\}]$.
\end{lemma}
\begin{proof}
	Since by hypothesis $G_\nu$ is not rectilinear planar, by the representability condition in \cref{ta:representability} we have $[\ov{m}_x,\ov{M}_x] \cap [\ov{m}_y,\ov{M}_y] \cap [\ov{m}_z,\ov{M}_z] = \emptyset$. By \cref{pr:3-intervals}, $[\ov{m}_x,\ov{M}_x] \cap [\ov{m}_z,\ov{M}_z] = \emptyset$. We prove Property~(i) and~(ii) separately. \cref{fi:3children-minbend-example} provides an illustration for the graph in \cref{fi:3children-minbend-example-notation}, where $x=c$, $y=l$, and $z=r$.
	
	\medskip
	\noindent\textsf{Proof of Property (i)}. 
	We show that $b_\nu = \ov{m}_z - \ov{M}_x$. Observe that each of the three components $G_{\mu_l}$, $G_{\mu_c}$, and $G_{\mu_r}$ is an S-component with an exposed edge. 
	
	We first prove that $b_\nu$ bends are necessary. Suppose for a contradiction that $G_\nu$ admits an orthogonal representation $H'_\nu$ with $b_\nu'< b_\nu$ bends.
	Denote by $b'_x$ and $b'_z$ the number of bends in the restriction of $H'_\nu$ to $G_{\mu_x}$ and to $G_{\mu_z}$, respectively. By \cref{le:P-exposed-edges-support}, we can assume that all the bends $b'_x$ are along an exposed edge of $G_{\mu_x}$ and all the bends $b'_z$ are along an exposed edge of $G_{\mu_z}$. Consider the underlying graph $G'_\nu$ of $H'_\nu$ obtained by replacing each bend of $H'_\nu$ with a subdivision vertex. $G'_\nu$ is rectilinear planar.
	Denote by $[m_x',M_x']$ and $[m_z',M_z']$ the spirality intervals of $G'_{\mu_x}$ and $G'_{\mu_z}$. 
	Note that each subdivision vertex along an exposed edge of $G'_{\mu_x}$ allows one more turn (either to the left or to the right) in a rectilinear planar representation of this component with respect to a rectilinear planar representation of $G_{\mu_x}$. Hence, the spirality interval of $G'_{\mu_x}$ extends the one of $G_{\mu_x}$ by $b'_x$ units, both for the minimum value and for the maximum value. The same reasoning applies to $G'_{\mu_z}$. It follows that $m_x'=m_x-b'_x$, $M_x'=M_x+b'_x$, $m_z'=m_z-b'_z$, and $M_z'=M_z+b'_z$. 
	Consider the three representability intervals $[m'_l-2,M'_l-2]$, $[m'_c,M'_c]$, and $[m'_r+2,M'_r+2]$ for $G'_{\mu_l}$, $G'_{\mu_c}$, and $G'_{\mu_r}$, respectively. Suppose that $x=l$, i.e., $\min\{M_l-2, M_c, M_r+2\}=M_l-2$; then we define $\ov{m}'_x = m'_l-2$ and $\ov{M}'_x=M'_l-2$. Similarly, if $x=c$, we define $\ov{m}'_x = m'_c$ and $\ov{M}'_x=M'_c$. Finally, if $x=r$, we define $\ov{m}'_x = m'_r+2$ and $\ov{M}'_x=M'_r+2$. Analogously, if $z=l$, i.e. $\max\{m_l-2, m_c, m_r+2\}=m_l-2$, we define $\ov{m}'_z = m_l-2$; if $z=c$, we define $\ov{m}'_z = m'_c$ and $\ov{M}'_z=M'_c$; if $z=r$, we define $\ov{m}'_z = m'_r+2$ and $\ov{M}'_z=M'_r+2$. \cref{fi:3children-minbend-example-notation-2} illustrates this notation for the graph of \cref{fi:3children-minbend-example-notation}.
	\begin{figure}[h]
		\centering
		\includegraphics[page=2,width=0.9\columnwidth]{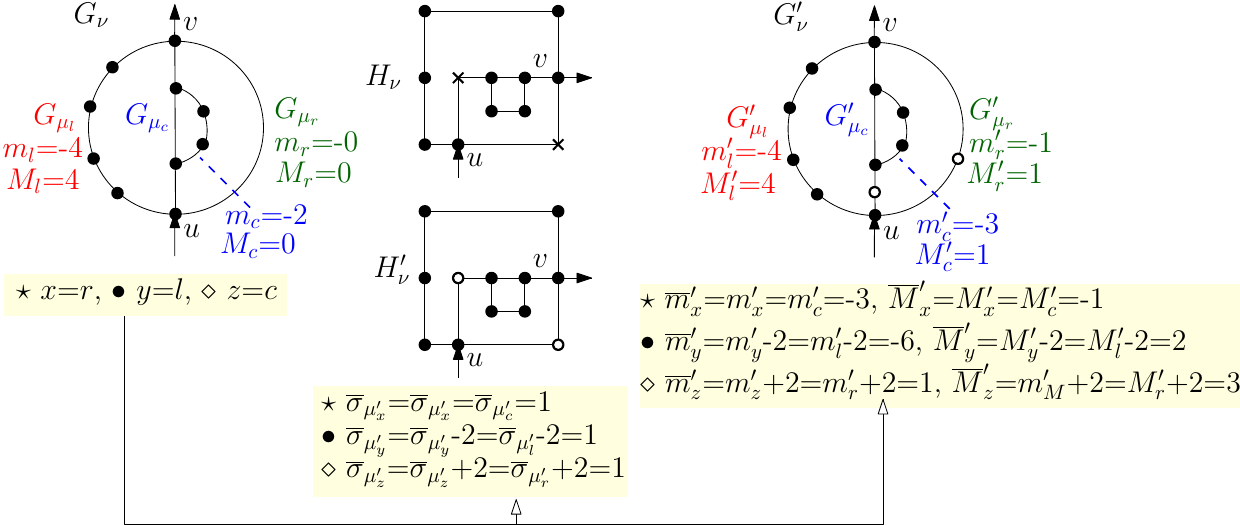}
		\caption{Illustration for the proof of Property (i) of \cref{le:P-3-exposed-edges}, on graph $G_\nu$ in \cref{fi:3children-minbend-example-notation}.}\label{fi:3children-minbend-example-notation-2}
	\end{figure}
	
	We have, $\ov{m}'_x=\ov{m}_x-b'_x$, $\ov{m}'_z=\ov{m}_z-b'_z$, $\ov{M}'_x=\ov{M}_x+b'_x$, and $\ov{M}'_z=\ov{M}_z+b'_z$. Since $G_\nu'$ is rectilinear planar,  we have $[\ov{m}'_x,\ov{M}'_x] \cap [\ov{m}'_z,\ov{M}'_z]=[\ov{m}_x-b'_x,\ov{M}_x+b'_x] \cap [\ov{m}_z-b'_z,\ov{M}_z+b'_z] \not = \emptyset$. Therefore,  $\ov{m}_z-b'_z\le \ov{M}_x+b'_x$, which implies $b_\nu=\ov{m}_z-\ov{M}_x \le b'_x+b'_z\le b'_\nu < b_\nu$, a contradiction.
	
	We now prove that $b_\nu$ bends suffice. Let $b_x$ and $b_z$ be two non-negative integers such that $b_x+b_z=b_\nu$, $b_x\ge \ov{m}_y-\ov{M}_x$, and $b_z\ge \ov{m}_z-\ov{M}_y$. Note that $b_x$ and $b_z$ always exist because $\ov{m}_y-\ov{M}_x+\ov{m}_z-\ov{M}_y=b_\nu-\ov{M}_y+\ov{m}_y\le b_\nu$. Insert $b_x$ subdivision vertices on any exposed edge of $G_{\mu_x}$ and insert $b_z$ subdivision vertices on any exposed edge of $G_{\mu_z}$. Clearly no subdivision vertex has been inserted on $G_{\mu_y}$ since $b_x+b_z=b_\nu$. Call $G'_{\mu_x}$, $G'_{\mu_y}$, and $G'_{\mu_z}$ the resulting components (note that $G'_{\mu_y}=G_{\mu_y}$). Define $\ov{m}'_x$, $\ov{m}'_z$, $\ov{M}'_x$, $\ov{M}'_z$ as in the first part of the proof.
	Suppose that $y = l$, i.e. $\min\{M_l-2, M_c, M_r+2\}\not =M_l-2$ and $\max\{m_l-2, m_c, m_r+2\}\not =m_l-2$. Then we define $\ov{m}'_y = m'_l-2$ and $\ov{M}'_y=M'_l-2$. Similarly, if $y = c$, define $\ov{m}'_y = m'_c$ and $\ov{M}'_y=M'_c$. Finally, if $y = r$, we define $\ov{m}'_y = m'_r+2$ and $\ov{M}'_y = M'_r+2$. 
	Consider the plane graph $G'_{\nu}$ obtained by the union of $G'_{\mu_x}$, $G'_{\mu_y}$, and $G'_{\mu_z}$. To prove that $G'_\nu$ is rectilinear planar, by the representability condition in \cref{ta:representability}, it suffices to show that $[\ov{m}'_x,\ov{M}'_x]\cap [\ov{m}'_y,\ov{M}'_y]\cap [\ov{m}'_z,\ov{M}'_z]=[\ov{m}_x-b_x,\ov{M}_x+b_x]\cap [\ov{m}_y,\ov{M}_y]\cap [\ov{m}_z-b_z,\ov{M}_z+b_z]\not = \emptyset$. We have:
	\begin{itemize}
		\item $b_x\ge \ov{m}_y-\ov{M}_x$, hence $\ov{m}_y\le \ov{M}_x+b_x$ and  $[\ov{m}_x-b_x,\ov{M}_x+b_x]\cap [\ov{m}_y,\ov{M}_y]\not = \emptyset$;
		\item $b_z\ge \ov{m}_z-\ov{M}_y$, hence $\ov{m}_z-b_z\le \ov{M}_y$ and $[\ov{m}_y,\ov{M}_y]\cap [\ov{m}_z-b_z,\ov{M}_z+b_z]\not = \emptyset$;
		\item $\ov{m}_z-\ov{M}_x=b_z+b_x=b_\nu$, hence $\ov{m}_z-b_z= \ov{M}_x+b_x$ and $[\ov{m}_x-b_x,\ov{M}_x+b_x]\cap [\ov{m}_z-b_z,\ov{M}_z+b_z]\not = \emptyset$.
	\end{itemize}
	
	Hence $G'_{\nu}$ has a rectilinear planar representation $H'_{\nu}$; replacing its subdivision vertices with bends, we get an orthogonal representation of $G_\nu$ with $b_\nu$ bends.
	
	\medskip\noindent\textsf{Proof of Property (ii)}.
	We show that set $I'_\nu$ is an interval of feasible spiralities for the orthogonal representations of $G_\nu$ with $b_\nu$ bends. Namely, we show that any orthogonal representation $H_\nu$ of $G_\nu$ with $b_\nu$ bends has spirality in the interval $I_\nu'$ and that for every value $\sigma_\nu\in I_\nu'$ there exists an orthogonal representation of $G_\nu$ with spirality $\sigma_\nu$ and $b_\nu$ bends. 
	
	\smallskip
	Suppose that $G_\nu$ has an orthogonal representation $H_\nu$ with $b_\nu$ bends and let $\sigma_\nu$ be the spirality of $H_\nu$. We prove that $\sigma_\nu\in I'_\nu = [\max\{\ov{M}_x,\ov{m}_y\}, \min\{\ov{m}_z,\ov{M}_y\}]$. Let $b_l$, $b_c$, and $b_r$ be the number of bends in the restriction of $H_\nu$ to $G_{\mu_l}$, $G_{\mu_C}$, and $G_{\mu_r}$, respectively, where $b_l+b_c+b_r=b_\nu$. 
	Let $H'_\nu$ be the rectilinear planar representation obtained from $H_\nu$ by replacing each bend with a subdivision vertex and let $G_\nu'$ be the underlying graph. Clearly the spirality of~$H_\nu'$ equals~$\sigma_\nu$.
	For any $t\in \{x,y,z\}$, by \cref{le:P-exposed-edges-support} we can assume that all the $b_t$ bends are along an exposed edge of $G_{\mu_t}$ and, consequently, $\sigma_{\mu_t'}\in [m_t-b_t, M_t+b_t]$. By using the same notation as in the proof of Property~(i) we define $\ov{m}'_x = m'_l-2$ and $\ov{M}'_x=M'_l-2$ if $x = l$, $\ov{m}'_x = m'_c$ and $\ov{M}'_x=M'_c$ if $x = c$, and $\ov{m}'_x = m'_r+2$ and $\ov{M}'_x = M'_r+2$ if $x = r$. The values $\ov{m}'_y$, $\ov{M}'_y$, $\ov{m}'_z$, and $\ov{M}'_z$ are defined analogously (see \cref{fi:3children-minbend-example-notation-2}). Since $G'_\nu$ is rectilinear planar, by \cref{ta:representability} we have $[\ov{m}_x-b_x,\ov{M}_x+b_x]\cap [\ov{m}_y-b_y,\ov{M}_y+b_y]\cap [\ov{m}_z-b_z,\ov{M}_z+b_z]\not = \emptyset$.
	\begin{clm}
		The following relations hold: (1)~$b_y=0$; (2)~$b_x\le \ov{M}_y-\ov{M}_x$; (3)~$b_z\le \ov{m}_z-\ov{m}_y$.
	\end{clm}
	\begin{claimproof}
		We prove the three relations by contradiction.
		\begin{itemize}
			\item[(1)] If $b_y>0$, then $b_x+b_z< b_\nu=\ov{m}_z - \ov{M}_x$. Hence, $\ov{M}_x+b_x<\ov{m}_z-b_z$ and $[\ov{m}_x-b_x,\ov{M}_x+b_x]\cap [\ov{m}_z-b_z,\ov{M}_z+b_z]=\emptyset$, a contradiction.
			\item[(2)] If $b_x> \ov{M}_y-\ov{M}_x$, then $b_y+b_z< b_\nu-(\ov{M}_y-\ov{M}_x)=\ov{m}_z - \ov{M}_x-(\ov{M}_y-\ov{M}_x)=\ov{m}_z-\ov{M}_y$. Hence, $\ov{M}_y+b_y<\ov{m}_z-b_z$ and $[\ov{m}_y-b_y,\ov{M}_y+b_y]\cap [\ov{m}_z-b_z,\ov{M}_z+b_z] = \emptyset$, a contradiction.
			\item[(3)] If $b_z> \ov{m}_z-\ov{m}_y$, then $b_x+b_y< b_\nu-(\ov{m}_z-\ov{m}_y)=\ov{m}_z - \ov{M}_x-(\ov{m}_z-\ov{m}_y)=\ov{m}_y-\ov{M}_x$. Hence, $\ov{M}_x+b_x<\ov{m}_y-b_y$ and $[\ov{m}_x-b_x,\ov{M}_x+b_x]\cap [\ov{m}_y-b_y,\ov{M}_y+b_y]= \emptyset$, a contradiction.
		\end{itemize}
	\end{claimproof}

	\begin{figure}[tb]
		\centering
		\begin{subfigure}{.4\columnwidth}
			\centering
			\includegraphics[page=1,width=\columnwidth]{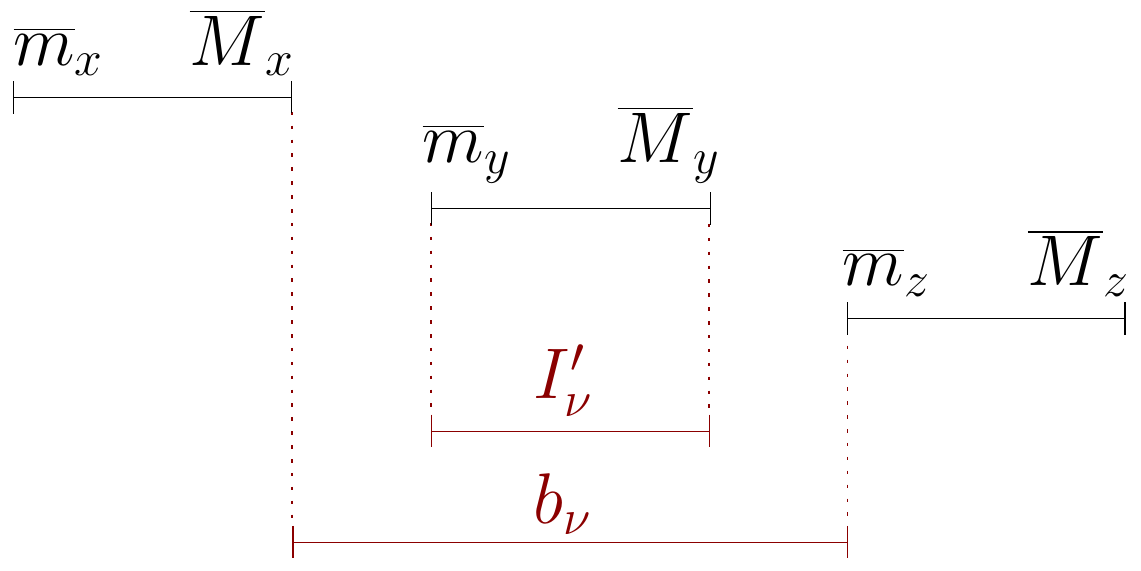}
			\subcaption{\centering}
			\label{fi:3children-nonintersecting-interval-1}
		\end{subfigure}
		\hfil
		\begin{subfigure}{.4\columnwidth}
			\centering
			\includegraphics[page=2,width=\columnwidth]{3children-nonintersecting-interval.pdf}
			\subcaption{\centering}
			\label{fi:3children-nonintersecting-interval-2}
		\end{subfigure}
		\hfil
		\begin{subfigure}{.4\columnwidth}
			\centering
			\includegraphics[page=3,width=\columnwidth]{3children-nonintersecting-interval.pdf}
			\subcaption{\centering}
			\label{fi:3children-nonintersecting-interval-3}
		\end{subfigure}
		\hfil
		\begin{subfigure}{.4\columnwidth}
			\centering
			\includegraphics[page=4,width=\columnwidth]{3children-nonintersecting-interval.pdf}
			\subcaption{\centering}
			\label{fi:3children-nonintersecting-interval-4}
		\end{subfigure}
		\caption{The four possible cases for the proof of Property~(ii) of \cref{le:P-3-exposed-edges}}\label{fi:3children-nonintersecting-interval}
	\end{figure}
	
	See for example \cref{fi:3children-minbend-example} where every bend is regarded as a subdivision vertex. We have $b_y=0$, $b_x\le \ov{M}_y-\ov{M}_x=2-0=2$, and $b_z\le  \ov{m}_z-\ov{m}_y=2+6=8$ in the three rectilinear representations of $G_\nu'$.
	We now consider spirality values of the components of $H'_\nu$ and define related values $\ov{\sigma}_{\mu_x'}=\sigma_{\mu_x'}-2$ if $x = l$, $\ov{\sigma}_{\mu_x'}=\sigma_{\mu_x'}$ if $x = c$, and $\ov{\sigma}_{\mu_x'}=\sigma_{\mu_x'}+2$ if $x = r$. The values $\ov{\sigma}_{\mu'_y}$ and $\ov{\sigma}_{\mu'_z}$ are defined analogously. See \cref{fi:3children-minbend-example-notation-2} for an illustration of the notation $\ov{\sigma}_{\mu_t'}$ for $t\in \{x,y,z\}$. We have $\ov{\sigma}_{\mu'_t}\in [\ov{m}_t-b_t,\ov{M}_t+b_t]$ for any $t\in \{x,y,z\}$. Also, by \cref{le:spirality-P-node-3-children} $\sigma_\nu=\ov{\sigma}_{\mu_x'}= \ov{\sigma}_{\mu_y'}=\ov{\sigma}_{\mu_z'}$. Hence $\sigma_\nu\in [\ov{m}_x-b_x,\ov{M}_x+b_x]\cap [\ov{m}_y-b_y,\ov{M}_y+b_y]\cap [\ov{m}_z-b_z,\ov{M}_z+b_z]$. Note that by Relation~(1) of the claim $[\ov{m}_y-b_y,\ov{M}_y+b_y]=[\ov{m}_y,\ov{M}_y]$. We show that $[\ov{m}_x-b_x,\ov{M}_x+b_x]\cap [\ov{m}_y,\ov{M}_y]\cap [\ov{m}_z-b_z,\ov{M}_z+b_z]=[\max\{\ov{M}_x,\ov{m}_y\}, \min\{\ov{m}_z,\ov{M}_y\}]$. We have four cases; see \cref{fi:3children-nonintersecting-interval}.	
	
	\smallskip\noindent\textbf{Case~(a): $\ov{M}_x<\ov{m}_y$ and $\ov{M}_y<\ov{m}_z$} (see \cref{fi:3children-nonintersecting-interval-1}). In this case $I'_\nu=[\max\{\ov{M}_x,\ov{m}_y\}, \min\{\ov{m}_z,\ov{M}_y\}] = [\ov{m}_y,\ov{M}_y]$. Refer to \cref{fi:3children-nonintersecting-interval-Case-a} for an illustration of the argument that refines \cref{fi:3children-nonintersecting-interval-1}).
	\begin{figure}[tb]
		\centering
		\includegraphics[page=6,width=0.54\columnwidth]{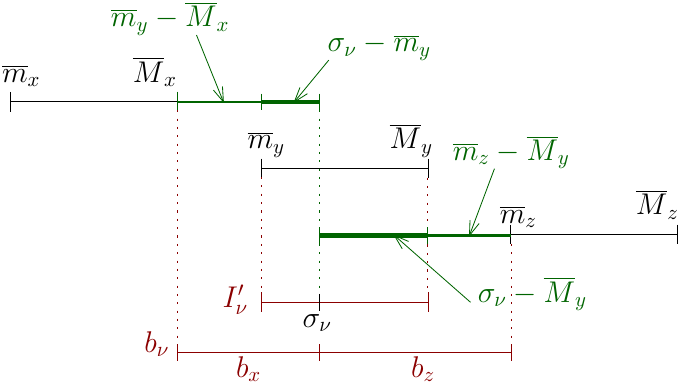}
		\caption{
			A more detailed illustration for Case~(a) in the proof of \cref{le:P-3-exposed-edges}.
			For any value $\sigma_\nu \in I'_\nu$, we need to add $b_x = \sigma_\nu - \ov{M}_x = (\sigma_\nu-\ov{m}_y)+(\ov{m}_y-\ov{M}_x)$ bends to the representation of $G_{\mu_x}$ and $b_z = \ov{m}_z-\sigma_\nu = (\ov{m}_z-\ov{M}_y) + (\ov{M}_y-\sigma_\nu)$ bends to the representation of $G_{\mu_z}$. 
		}
		\label{fi:3children-nonintersecting-interval-Case-a}
	\end{figure}
	
	First we prove that $[\ov{m}_x-b_x,\ov{M}_x+b_x]\cap [\ov{m}_y,\ov{M}_y]=[\ov{m}_y,\ov{M}_y]$.
	We have $\ov{M}_x+b_x\ge \ov{m}_y$, or $[\ov{m}_x-b_x,\ov{M}_x+b_x]\cap [\ov{m}_y,\ov{M}_y]=\emptyset$. By Relation~(2) of the claim we have $b_x\le \ov{M}_y-\ov{M}_x$. Hence, $b_x\in [\ov{m}_y-\ov{M}_x,\ov{M}_y-\ov{M}_x]$. We have $\ov{M}_x+b_x\in [\ov{M}_x+\ov{m}_y-\ov{M}_x,\ov{M}_x+\ov{M}_y-\ov{M}_x]=[\ov{m}_y,\ov{M}_y]$. 
	It follows that  $[\ov{m}_x-b_x,\ov{M}_x+b_x]\cap [\ov{m}_y,\ov{M}_y]=[\ov{m}_y,\ov{M}_y]$. 
	Second, we prove $[\ov{m}_y,\ov{M}_y]\cap [\ov{m}_z-b_z,\ov{M}_z+b_z]=[\ov{m}_y,\ov{M}_y]$.
	We have $\ov{M}_y\ge \ov{m}_z-b_z$, or $[\ov{m}_y,\ov{M}_y]\cap [\ov{m}_z-b_z,\ov{M}_z+b_z] =\emptyset$. By Relation~(3) of the claim we have $b_z\le \ov{m}_z-\ov{m}_y$. Hence, $b_z\in [\ov{m}_z-\ov{M}_y,\ov{m}_z-\ov{m}_y]$. We have $\ov{m}_z-b_z\in [\ov{m}_z-\ov{m}_z+\ov{m}_y,\ov{m}_z-\ov{m}_z-\ov{M}_y]=[\ov{m}_y,\ov{M}_y]$.
	It follows that  $[\ov{m}_y,\ov{M}_y]\cap [\ov{m}_z-b_z,\ov{M}_z+b_z]=[\ov{m}_y,\ov{M}_y]$. 
	Hence, $[\ov{m}_x-b_x,\ov{M}_x+b_x]\cap [\ov{m}_y,\ov{M}_y]\cap [\ov{m}_z-b_z,\ov{M}_z+b_z]=[\ov{m}_y,\ov{M}_y]\cap [\ov{m}_z-b_z,\ov{M}_z+b_z]=[\ov{m}_y,\ov{M}_y]=I'_\nu$.
	
	\smallskip\noindent\textbf{Case~(b): $\ov{M}_x\ge\ov{m}_y$ and $\ov{M}_y<\ov{m}_z$} (see \cref{fi:3children-nonintersecting-interval-2}). In this case $I'_\nu =[\ov{M}_x,\ov{M}_y]$. By the same reasoning as in Case~(a) we have  $[\ov{m}_y,\ov{M}_y]\cap [\ov{m}_z-b_z,\ov{M}_z+b_z]=[\ov{m}_y,\ov{M}_y]$.
	We prove that $[\ov{m}_x-b_x,\ov{M}_x+b_x]\cap [\ov{m}_y,\ov{M}_y]=[\ov{M}_x,\ov{M}_y]$.
	By Relation~(2) of the claim $b_x\in [0,\ov{M}_y-\ov{M}_x]$. We have $\ov{M}_x+b_x\in [\ov{M}_x,\ov{M}_x+\ov{M}_y-\ov{M}_x]=[\ov{M}_x,\ov{M}_y]$. 
	It follows that  $[\ov{m}_x-b_x,\ov{M}_x+b_x]\cap [\ov{m}_y,\ov{M}_y]=[\ov{M}_x,\ov{M}_y]$. 
	Hence, since $\ov{M}_x\ge\ov{m}_y$, we have $[\ov{m}_x-b_x,\ov{M}_x+b_x]\cap [\ov{m}_y,\ov{M}_y]\cap [\ov{m}_z-b_z,\ov{M}_z+b_z]=[\ov{m}_x-b_x,\ov{M}_x+b_x]\cap [\ov{m}_y,\ov{M}_y]\cap [\ov{m}_y,\ov{M}_y]\cap [\ov{m}_z-b_z,\ov{M}_z+b_z]=[\ov{m}_y,\ov{M}_y]\cap [\ov{M}_x,\ov{M}_y]=[\ov{M}_x,\ov{M}_y]=I'_\nu$.
	
	\smallskip\noindent\textbf{Case~(c): $\ov{M}_x<\ov{m}_y$ and $\ov{M}_y\ge \ov{m}_z$} (see \cref{fi:3children-nonintersecting-interval-3}). In this case $I'_\nu =[\ov{m}_y,\ov{m}_z]$.  It is possible to prove that  $[\ov{m}_x-b_x,\ov{M}_x+b_x]\cap [\ov{m}_y,\ov{M}_y]=[\ov{m}_y,\ov{M}_y]$ as we did in Case~(a).
	We prove $[\ov{m}_y,\ov{M}_y]\cap [\ov{m}_z-b_z,\ov{M}_z+b_z]=[\ov{m}_y,\ov{m}_z]$.
	By Relation~(3) of the claim $b_z\in [0,\ov{m}_z-\ov{m}_y]$. We have $\ov{m}_z-b_z\in [\ov{m}_z-\ov{m}_z+\ov{m}_y,\ov{m}_z]=[\ov{m}_y,\ov{m}_z]$.
	It follows that  $[\ov{m}_y,\ov{M}_y]\cap [\ov{m}_z-b_z,\ov{M}_z+b_z]=[\ov{m}_y,\ov{m}_z]$. 
	Hence, since $\ov{M}_y\ge \ov{m}_z$, we have $[\ov{m}_x-b_x,\ov{M}_x+b_x]\cap [\ov{m}_y,\ov{M}_y]\cap [\ov{m}_z-b_z,\ov{M}_z+b_z]=[\ov{m}_x-b_x,\ov{M}_x+b_x]\cap [\ov{m}_y,\ov{M}_y]\cap [\ov{m}_y,\ov{M}_y]\cap [\ov{m}_z-b_z,\ov{M}_z+b_z]=[\ov{m}_y,\ov{M}_y]\cap [\ov{m}_y,\ov{m}_z]= [\ov{m}_y,\ov{m}_z]=I'_\nu$.
	
	\smallskip\noindent\textbf{Case~(d): $\ov{M}_x\ge \ov{m}_y$ and $\ov{M}_y\ge \ov{m}_z$} (see \cref{fi:3children-nonintersecting-interval-4}). In this case $I'_\nu =[\ov{M}_x,\ov{m}_z]$.  It is possible to prove that  $[\ov{m}_x-b_x,\ov{M}_x+b_x]\cap [\ov{m}_y,\ov{M}_y]=[\ov{M}_x,\ov{M}_y]$ as we did in Case~(b) and that $[\ov{m}_y,\ov{M}_y]\cap [\ov{m}_z-b_z,\ov{M}_z+b_z]=[\ov{m}_y,\ov{m}_z]$ as we did for Case~(c).
	We have $[\ov{m}_x-b_x,\ov{M}_x+b_x]\cap [\ov{m}_y,\ov{M}_y]\cap [\ov{m}_z-b_z,\ov{M}_z+b_z]=[\ov{m}_x-b_x,\ov{M}_x+b_x]\cap[\ov{m}_y,\ov{M}_y]\cap [\ov{m}_y,\ov{M}_y]\cap [\ov{m}_z-b_z,\ov{M}_z+b_z]=[\ov{M}_x,\ov{M}_y]\cap [\ov{m}_y,\ov{m}_z]=[\ov{M}_x,\ov{m}_z]$ since $\ov{M}_x\ge \ov{m}_y$ and $\ov{M}_y\ge \ov{m}_z$. Hence, 
	$[\ov{m}_x-b_x,\ov{M}_x+b_x]\cap [\ov{m}_y,\ov{M}_y]\cap [\ov{m}_z-b_z,\ov{M}_z+b_z]=[\ov{M}_x,\ov{m}_z]=I'_\nu$. 
	
	\medskip
	Suppose now that we are given $\sigma_\nu\in I'_\nu =[\max\{\ov{M}_x,\ov{m}_y\}, \min\{\ov{m}_z,\ov{M}_y\}]$.
	We show that
	there exists an orthogonal representation $H_\nu$ of $G_\nu$ with $b_\nu$ bends and with spirality~$\sigma_\nu$. This is equivalent to showing that there exists a plane graph $G'_\nu$ obtained by adding $b_\nu$ subdivision vertices along some edges of $G_\nu$ such that  $G'_\nu$ has a rectilinear orthogonal representation with spirality $\sigma_\nu$. To construct $G'_\nu$ we insert a suitable number $b_z\in[0,b_\nu]$ of subdivision vertices on an exposed edge of 
	$G_{\mu_z}$ and $b_x=b_\nu-b_z$ subdivision vertices on an exposed edge of $G_{\mu_x}$ (as a consequence, we do not insert any subdivision vertex in $G_{\mu_y}$). Let $G'_{\mu_x}$, $G'_{\mu_y}$, and $G'_{\mu_z}$ be the resulting graphs. Since by hypothesis $G_{\mu_x}$, $G_{\mu_y}$, and $G_{\mu_z}$ are rectilinear planar, we have that also $G'_{\mu_x}$, $G'_{\mu_y}$, and $G'_{\mu_z}$ are rectilinear planar. Also, with the same reasoning as in the proof of Property~(i), the representability intervals of $G'_{\mu_x}$, $G'_{\mu_y}$, and $G'_{\mu_z}$ are $[m'_x,M'_x]=[m_x-b_x,M_x+b_x]$, $[m_y,M_y]$, and $[m_z-b_z,M_z+b_z]$, respectively. 
	For any $t\in \{x,y,z\}$ we define $\ov{m}'_t$, $\ov{M}'_t$, and  $\ov{\sigma}_{\mu_t'}$ as in the first part of the proof of Property~(ii). 
	We now describe how to compute $b_x$ and $b_z$, and how to set $\ov{\sigma}_{\mu_x}$, $\ov{\sigma}_{\mu_y}$, and $\ov{\sigma}_{\mu_z}$. Let  $c_1=\max\{\ov{M}_x,\ov{m}_y\}$ and $c_2= \min\{\ov{m}_z,\ov{M}_y\}$.  
	We have: $c_1=\ov{m}_y$ and $c_2=\ov{M}_y$ in the case of \cref{fi:3children-nonintersecting-interval-1}; $c_1=\ov{M}_x$ and $c_2=\ov{M}_y$ in the case of \cref{fi:3children-nonintersecting-interval-2}; $c_1=\ov{m}_y$ and $c_2=\ov{m}_z$ in the case of \cref{fi:3children-nonintersecting-interval-3}; $c_1=\ov{M}_x$ and $c_2=\ov{m}_z$ in the case of \cref{fi:3children-nonintersecting-interval-4}.   
	We have $\sigma_\nu\in I'_\nu =[c_1,c_2]$.  We set $b_z=\ov{m}_z-\sigma_\nu$ and, consequently, $b_x=b_\nu-b_z=b_\nu-\ov{m}_z+\sigma_\nu=\ov{M}_x+\sigma_\nu$. We prove that $b_z$ (and consequently $b_x$) is in the interval $[0,b_\nu]$.  We have $b_z=\ov{m}_z-\sigma_\nu\in [\ov{m}_z-c_2,\ov{m}_z-c_1] \subseteq [\ov{m}_z-\ov{m}_z,\ov{m}_z-\ov{M}_x]$, since $c_2\le \ov{m}_z$ and $c_1\ge \ov{M}_x$. Hence $b_z\in  [\ov{m}_z-\ov{m}_z,\ov{m}_z-\ov{M}_x]=[0,b_\nu]$. We now set  $\ov{\sigma}_{\mu_x}=\ov{M}_x+b_x=\sigma_\nu$ and  $\ov{\sigma}_{\mu_z}=\ov{m}_z-b_z=\sigma_\nu$. Notice that $\sigma_\nu\in [c_1,c_2]\in [\ov{m}_y,\ov{M}_y]$. Hence, it is possible to set  $\ov{\sigma}_{\mu_y}=\sigma_\nu$.
	By \cref{le:spirality-P-node-3-children} we can get a rectilinear planar representation $H'_\nu$ of $G'_\nu$ by a parallel composition of rectilinear planar representations of  $G'_{\mu_x}$, $G'_{\mu_y}$, and $G'_{\mu_z}$ with spiralities $\ov{\sigma}_{\mu_x}$ $\ov{\sigma}_{\mu_y}$, and $\ov{\sigma}_{\mu_z}$, respectively. By the same lemma, the spirality of $H'_\nu$ is $\sigma_\nu'=\ov{\sigma}_{\mu_x}=\ov{\sigma}_{\mu_y}=\ov{\sigma}_{\mu_z}=\sigma_\nu$. By replacing the subdivision vertices of $H'_\nu$ with bends we get an orthogonal representation of $G_\nu$ with $b_\nu$ bends and spirality $\sigma_\nu$.
\end{proof}

\subsection{Budget of P-nodes with two children}\label{sse:bottom-up-2-children}

Let $\mu_l$ and $\mu_r$ be the two children of $\nu$, and suppose that $G_{\mu_l}$ and $G_{\mu_r}$ are rectilinear planar with representability intervals $I_{\mu_l}=[m_l,M_l]$ and $I_{\mu_r}=[m_r,M_r]$, respectively. The representability condition for $\nu$ given in \cref{ta:representability} is expressed in terms of intersection between the interval $[m_l-M_r, M_l-m_r]$ and another interval $\Delta_\nu$ that depends on the type of P-node. Specifically: $\Delta_\nu = [2,4-\gamma]$ if $\nu$ is of type \Pio{2}{\lambda\beta}; $\Delta_\nu = [\frac{5}{2}, \frac{7}{2} - \gamma]$ if $\nu$ is of type \Pio{3d}{\lambda\beta}; and $\Delta_\nu = [3,3]$ if $\nu$ is of type \Pin{3dd'}.
In the following, given two non-intersecting intervals of real numbers $A_1 = [m_1,M_1]$ and $A_2 = [m_2,M_2]$, the \emph{distance} between $A_1$ and $A_2$ is defined as $\delta(A_1, A_2)=\min\{\lvert M_1-m_2 \rvert, \lvert M_2-m_1 \rvert\}$.

\cref{ssse:with-exposed} handles the case of a P-node with two children both having an exposed egde. 
\cref{ssse:no-exposed} handles the more involved cases in which either the left child or the right child of the P-node has no exposed edge. Note that, since the vertex-degree is at most four, at least one of the two children of the P-node must have an exposed edge.
We start giving two simple combinatorial results (\cref{le:P-2-bend-support-pio2alphabeta,le:P-2-bend-support-pio3dalphabeta}).

\begin{lemma}\label{le:P-2-bend-support-pio2alphabeta}
	Let $\nu$ be a P-node of type \Pio{2}{\lambda\beta} with children $\mu_l$ and $\mu_r$. Let $G_{\mu_l}$ and $G_{\mu_r}$ be rectilinear planar. For any rectilinear planar representation of $G_{\mu_l}$ and $G_{\mu_r}$, and for any $d\in \{l,r\}$, the following relation holds:
	
	$$k_{u}^d \alpha_{u}^d + k_{v}^d\alpha_v^d=h \Leftrightarrow h\in [\frac{\gamma}{2},2-\frac{\gamma}{2}],$$
	
	\noindent where both $h$ and  $\frac{\gamma}{2}$ are either integer or semi-integer numbers.
\end{lemma}
\begin{proof}
	For the reader's convenience, we summarize in \cref{ta:p-coefficients} the parameters associated with each type of P-node with two children. Based on these parameters, we analyze three~cases:
	
	\begin{itemize}
		\item $\lambda =\beta=1$. In this case, $\gamma=\lambda+\beta-2=0$, $h \in \{0,1,2\}$, and $k_u^l=k_v^l=k_u^r=k_v^r=1$ (see \cref{ta:p-coefficients}). Hence, $k_{u}^d \alpha_{u}^d + k_{v}^d\alpha_v^d= \alpha_{u}^d + \alpha_v^d$. Since both poles $u$ and $v$ have outdegree one, both $\alpha_{u}^d$ and $\alpha_v^d$ can take either value $0$ or $1$, hence $\alpha_{u}^d + \alpha_v^d$ can take all and only the values in the set $\{0,1,2\}$.   
		
		\item $\lambda=1$ and $\beta=2$. In this case $\gamma=1$,  $h \in \{\frac{1}{2},\frac{3}{2}\}$, $k_u^l=k_u^r=1$, and $k_v^l=k_v^r=\frac{1}{2}$ (see \cref{ta:p-coefficients}). Hence, $k_{u}^d \alpha_{u}^d + k_{v}^d\alpha_v^d= \alpha_{u}^d + \frac{1}{2}\alpha_v^d$.  Since we are assuming that $\outdeg(u)=1$ and $\outdeg(v)=2$, we have $\alpha_{u}^d\in \{0,1\}$ and $\alpha_v^d=1$, i.e., $\alpha_{u}^d + \frac{1}{2}\alpha_v^d$ equals either $\frac{1}{2}$ or $\frac{3}{2}$. 
		
		\item $\lambda=\beta=2$. In this case $\gamma=2$, $c \in \{1\}$, and $k_u^l=k_u^r=k_v^l=k_v^r=\frac{1}{2}$ (see \cref{ta:p-coefficients}). Hence, $k_{u}^d \alpha_{u}^d + k_{v}^d\alpha_v^d= \frac{1}{2}\alpha_{u}^d + \frac{1}{2}\alpha_v^d$. Since $\outdeg(u)=\outdeg(v)=2$, we have $\alpha_{u}^d=\alpha_v^d=1$.
	\end{itemize}
\end{proof}

\renewcommand{\arraystretch}{1.5}
\begin{table}[h]
	\centering
	\caption{Parameters for the P-nodes with 2 children.}\label{ta:p-coefficients}
	\begin{tabular}{|c|c|c|c|c|c|c|c|c|c|c|}
		\hline	\rowcolor{antiquewhite} \hspace{1mm} \textsc{Type} \hspace{1mm} & \hspace{1mm} \textsc{$k_u^l$} \hspace{1mm}
		& \hspace{1mm} \textsc{$k_u^r$} \hspace{1mm}& \hspace{1mm} \textsc{$k_v^l$} \hspace{1mm} & \hspace{1mm} \textsc{$k_v^r$} \hspace{1mm} & \hspace{1mm} \textsc{$\lambda$} \hspace{1mm} & \hspace{1mm} \textsc{$\beta$} \hspace{1mm} & \hspace{1mm} \textsc{$\gamma$} \hspace{1mm} & \hspace{1mm} \textsc{$d$} \hspace{1mm} & \hspace{1mm} \textsc{$d'$} \hspace{1mm} \\

		\hline
		\Pio{2}{11} & $1$ & $1$ & $1$ & $1$ & $1$ & $1$ & $0$ & - & - \\
		\hline
		\Pio{2}{12} & $1$ & $1$ & $\frac{1}{2}$ & $\frac{1}{2}$ & $1$ & $2$ & $1$ & - & - \\
		\hline
		\Pio{2}{21} & $\frac{1}{2}$ & $\frac{1}{2}$ & $1$ & $1$ & $1$ & $2$ & $1$ & - & - \\
		\hline
		\Pio{2}{22} & $\frac{1}{2}$ & $\frac{1}{2}$ & $\frac{1}{2}$ & $\frac{1}{2}$ & $2$ & $2$ & $2$ & - & - \\
		\hline
		\Pio{3l}{11} & $1$ & $1$ & $\frac{1}{2}$ & $1$ & $1$ & $1$ & $0$ & $l$ & - \\
		\hline
		\Pio{3r}{11} & $1$ & $1$ & $1$ & $\frac{1}{2}$ & $1$ & $1$ & $0$ & $r$ & - \\
		\hline
		\Pio{3l}{12} & $\frac{1}{2}$ & $\frac{1}{2}$ & $\frac{1}{2}$ & $1$ & $1$ & $2$ & $1$ & $l$ & - \\
		\hline
		\Pio{3r}{12} & $\frac{1}{2}$ & $\frac{1}{2}$ & $1$ & $\frac{1}{2}$ & $1$ & $2$ & $1$ & $r$ & - \\
		\hline
		\Pin{3ll} & $\frac{1}{2}$ & $1$ & $\frac{1}{2}$ & $1$ & $1$ & $1$ & - & $l$ & $l$ \\
		\hline
		\Pin{3lr} & $1$ & $\frac{1}{2}$ & $\frac{1}{2}$ & $1$ & $1$ & $1$ & - & $l$ & $r$ \\
		\hline
		\Pin{3rr} & $1$ & $\frac{1}{2}$ & $1$ & $\frac{1}{2}$ & $1$ & $1$ & - & $r$ & $r$ \\
		\hline
	\end{tabular}
\end{table}

\begin{lemma}\label{le:P-2-bend-support-pio3dalphabeta}
	Let $\nu$ be a P-node of type \Pio{3l}{\lambda\beta} with children $\mu_l$ and $\mu_r$. Let $G_{\mu_l}$ and $G_{\mu_r}$ be rectilinear planar. For any rectilinear planar representation of $G_{\mu_l}$ and $G_{\mu_r}$, the following relations hold:
	
	\begin{equation}\label{eq:pio3dalphabeta-l}
	k_{u}^l \alpha_{u}^l +k_{v}^l\alpha_v^l=h_l \Leftrightarrow h_l \in[\frac{\gamma}{2}+\frac{1}{2},\frac{3}{2}-\frac{\gamma}{2}], 
	\end{equation}
	\noindent where both $h_l$ and $\frac{\gamma}{2}+\frac{1}{2}$ are either integer or semi-integer numbers.
	
	\begin{equation}\label{eq:pio3dalphabeta-r}
	k_{u}^r \alpha_{u}^r +k_{v}^r\alpha_v^r=h_r \Leftrightarrow h_r \in[\frac{\gamma}{2}+1,2-\frac{\gamma}{2}], 
	\end{equation}
	\noindent where both $h_r$ and $\frac{\gamma}{2}+1$ are either integer or semi-integer numbers.
\end{lemma}
\begin{proof}
	$k_v^l=\frac{1}{2}$, $k_v^r=1$, and $\alpha_v^d=1$ for any $d \in \{l,r\}$. There are two subcases:
	\begin{itemize}
		\item $\beta=1$. In this case $\gamma=0$, $h_l \in \{\frac{1}{2},\frac{3}{2}\}$, $h_r \in \{1,2\}$, and $k_u^l=k_u^r=1$. We have $k_{u}^d \alpha_{u}^d + k_{v}^d\alpha_v^d= \alpha_{u}^d +k_v^d$.  
		Since $\alpha_u^d \in \{0,1\}$, we have $\alpha_{u}^d +k_v^d\in \{k_v^d,k_v^d +1\}$. Since $k_v^l=\frac{1}{2}$, we have $\{k_v^l,k_v^l +1\}=\{\frac{1}{2},\frac{3}{2}\}=\{\frac{\gamma}{2}+\frac{1}{2},\frac{3}{2}-\frac{\gamma}{2}\}$. Hence, Relation~\ref{eq:pio3dalphabeta-l} holds. Since $k_v^r=1$, we have $\{k_v^r,k_v^r +1\}=\{1,2\}=\{\frac{\gamma}{2}+1,2-\frac{\gamma}{2}\}$. Hence, Relation~\ref{eq:pio3dalphabeta-r} holds.
		
		\item $\beta=2$. In this case $\gamma=1$, $h_l \in \{1\}$, $h_r \in \{\frac{3}{2}\}$, and $k_u^l=k_u^r=\frac{1}{2}$. Since $\deg(u)=4$, $\alpha_u^d=1$. 
		We have $k_{u}^d \alpha_{u}^d + k_{v}^d\alpha_v^d= \frac{1}{2}+k_v^d$. Equivalently, $k_{u}^d \alpha_{u}^d + k_{v}^d\alpha_v^d \in \{\frac{1}{2}+k_v^d, \frac{1}{2}+k_v^d\}$.  
		Since $k_v^l=\frac{1}{2}$, we have
		$\{\frac{1}{2}+k_v^l,\frac{1}{2}+k_v^l\}=\{\frac{\gamma}{2}+\frac{1}{2},\frac{3}{2}-\frac{\gamma}{2}\}=\{1\}$. Hence, Relation~\ref{eq:pio3dalphabeta-l} holds. Since $k_v^r=1$, we have $\{\frac{1}{2}+k_v^r,\frac{1}{2}+k_v^r\}=\{\frac{\gamma}{2}+1,2-\frac{\gamma}{2}\}=\{\frac{3}{2}\}$. Hence, Relation~\ref{eq:pio3dalphabeta-r} holds.
	\end{itemize}
\end{proof}

\subsubsection{P-nodes with both children having an exposed edge}\label{ssse:with-exposed}

\begin{lemma}\label{le:P-2-exposed-edges}
	Let $\nu$ be a P-node with two children $\mu_l$ and $\mu_r$, each having an exposed edge. Let $G_{\mu_l}$ and $G_{\mu_r}$ be rectilinear planar with representability intervals $I_{\mu_l}=[m_l,M_l]$ and $I_{\mu_r}=[m_r,M_r]$, respectively. If $G_\nu$  is not rectilinear planar then: (i) the budget for $\nu$ is $b_\nu=\delta([m_l-M_r, M_l-m_r],\Delta_\nu)$; and (ii) the set of spirality values for an orthogonal representation of $G_\nu$ with $b_\nu$ bends is the interval $I'_\nu = [m-b_\nu, M+b_\nu]$.
\end{lemma}
\begin{proof}
	Since by hypothesis $G_\nu$ is not rectilinear planar, we have $[m_l-M_r, M_l-m_r] \cap \Delta_\nu  = \emptyset$. \cref{fi:I2O11_minbend_example-new} illustrates the statement for a $P$-node $\nu$ of type \Pio{2}{11}, by also showing how the interval $I'_\nu$ of Property (ii) is defined.
	
	\medskip\noindent\textsf{Proof of Property (i)}. We show that $b_\nu = \delta([m_l-M_r, M_l-m_r],\Delta_\nu)$. We first prove that $b_\nu$ bends are necessary. Suppose for a contradiction that $G_\nu$ admits an orthogonal representation $H'_\nu$ with $b'_\nu < b_\nu$ bends. Denote by $b'_l$ and $b'_r$ the number of bends in the restriction of $H'_\nu$ to $G_{\mu_l}$ and to $G_{\mu_r}$, respectively. By \cref{le:P-exposed-edges-support}, we can assume that all the bends $b'_l$ (resp. $b'_r$) are along an exposed edge of $G_{\mu_l}$ (resp. $G_{\mu_r}$).  Consider the underlying graph $G'_\nu$ of $H'_\nu$ obtained by replacing each bend of $H'_\nu$ with a subdivision vertex. $G'_\nu$ is rectilinear planar. Also, each subdivision vertex along an exposed edge of $G_{\mu_l}$ allows one more turn (either to the left or to the right) in a rectilinear representation of this component, i.e., it extends the spirality interval of $G_{\mu_l}$ by one unit, both for the minimum value and for the maximum value. The same considerations happen for $G_{\mu_r}$. Hence $G'_{\mu_l}$ and $G'_{\mu_r}$ are rectilinear planar with representability intervals $[m_l-b'_l,M_l+b'_l]$ and $[m_r-b'_r,M_r+b'_r]$, respectively, and the representability condition for $G'_\nu$ is $[m_l-b'_l-M_r-b'_r, M_l+b'_l-m_r+b'_r] \cap \Delta_\nu \neq \emptyset$. Since $b'_l+b'_r=b'_\nu$, we have $[m_l-M_r-b'_\nu, M_l-m_r+b'_\nu] \cap \Delta_\nu \neq \emptyset$, which implies $\delta([m_l-M_r,M_l-m_r],\Delta_\nu) \leq b'_\nu < b_\nu$, a contradiction. 
	
	\begin{figure}[tb]
		\centering
		\includegraphics[width=0.9\columnwidth]{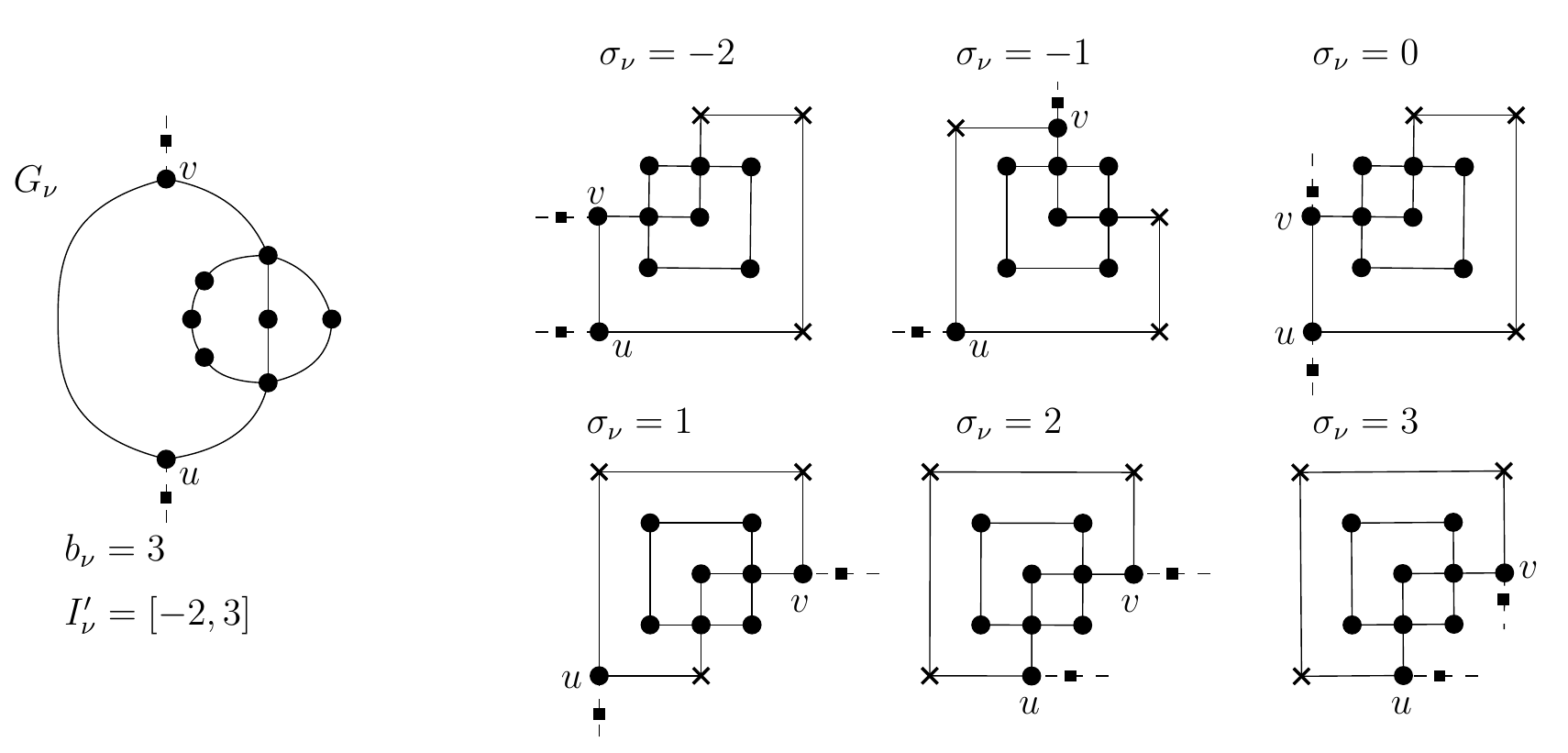}
		\caption{Illustration of \cref{le:P-2-exposed-edges} for a $P$-node $\nu$ of type \Pio{2}{11}, where $m_l=M_l=0$ and $m_r=M_r=1$. In this case $\Delta_\nu = [2,4]$, $m=\max\{m_l-2,m_r\}=-1$, and $M=\min\{M_l,M_r+2\}=0$. Since $M_l-m_r<2$, $G_\nu$ is not rectilinear planar. By Property~(i), $b_\nu=\delta([m_l-M_r, M_l-m_r],\Delta_\nu)=3$; by Property~(ii), $I'_\nu=[m-b_\nu,M+b_\nu]=[-2,3]$.
		}\label{fi:I2O11_minbend_example-new}
	\end{figure}
	
	We now prove that $b_\nu$ bends suffice. Let $b_l$ and $b_r$ be two arbitrarily chosen non-negative integers such that $b_l+b_r=b_\nu$. Insert $b_l$ (resp. $b_r$) subdivision vertices on any exposed edge of $G_{\mu_l}$ (resp. $G_{\mu_r}$). Call $G'_{\mu_l}$ and  $G'_{\mu_r}$ the resulting components. Since by hypothesis $G_{\mu_l}$ and $G_{\mu_r}$ are rectilinear planar, with the same argument as above, $G'_{\mu_l}$ and  $G'_{\mu_r}$ are both rectilinear planar with representability intervals $[m_l-b_l,M_l+b_l]$ and $[m_r-b_r,M_r+b_r]$, respectively. Consider the plane graph $G'_{\nu}$ obtained by the union of $G'_{\mu_l}$ and $G'_{\mu_r}$. Since by hypothesis $b_\nu = \delta([m_l-M_r, M_l-m_r],\Delta_\nu)$, we have that  $[m_l-M_r-b_\nu, M_l-m_r+b_\nu] \cap \Delta_\nu \neq \emptyset$. Since  $[m_l-b_l-M_r-b_r, M_l+b_l-m_r+b_r] = [m_l-M_r-b_\nu, M_l-m_r+b_\nu]$ we have $[m_l-b_l-M_r-b_r, M_l+b_l-m_r+b_r] \cap \Delta_\nu \neq \emptyset$. It follows that $G'_{\nu}$ admits a rectilinear planar representation $H'_{\nu}$. By replacing the $b_\nu$ subdivision vertices of $H'_{\nu}$ with bends, we get an orthogonal representation of $G_\nu$ with $b_\nu$ bends. 
	
	\medskip\noindent\textsf{Proof of Property (ii)}.
	Suppose without loss of generality that $\outdeg(u) \leq \outdeg(v)$. Denote by $\sigma_{\mu_l}$ and $\sigma_{\mu_r}$ the spiralities of orthogonal representations of $G_{\mu_l}$ and~$G_{\mu_r}$, respectively. 
	We distinguish three cases.

	\smallskip\noindent{\bf Case \Pio{2}{\lambda\beta}:}  In this case $\Delta_\nu = [2,4-\gamma]$, $m=\max\{m_l-2,m_r\}+\frac{\gamma}{2}$, and $M=\min\{M_l, M_r+2\}-\frac{\gamma}{2}$. 
	We show that the set $I'_\nu$ is an interval of feasible spiralities for the orthogonal representations of $G_\nu$ with $b_\nu$ bends.
	Suppose first that $G_\nu$ has an orthogonal representation $H_\nu$ with $b_\nu$ bends, and let $\sigma_\nu$ be the spirality of $H_\nu$. We prove that $\sigma_\nu\in [m-b_\nu,M+b_\nu]$.
	Let $b_l$ and $b_r$ be the number of bends in the restriction of $H_\nu$ to $G_{\mu_l}$ and to~$G_{\mu_r}$, respectively, where $b_l+b_r=b_\nu$. By \cref{le:P-exposed-edges-support}, we can assume that all the $b_l$ bends are along an exposed edge of $G_{\mu_l}$ and all the $b_r$ bends are along an exposed edge of $G_{\mu_r}$.
	Since $\sigma_{\mu_l}\in [m_l-b_l, M_l+b_l]$ and $b_l\in [0,b_\nu]$ we have $\sigma_{\mu_l}\in [m_l-b_\nu, M_l+b_\nu]$. Also, by \cref{le:spirality-P-node-2-children} we have $\sigma_\nu = \sigma_{\mu_l} -k_{u}^l \alpha_{u}^l - k_{v}^l \alpha_{v}^l$. By \cref{le:P-2-bend-support-pio2alphabeta}, we have $-k_{u}^l \alpha_{u}^l - k_{v}^l \alpha_{v}^l\in[\frac{\gamma}{2}-2,-\frac{\gamma}{2}]$. 
	Hence, $\sigma_\nu\in [m_l-b_\nu+\frac{\gamma}{2}-2, M_l+b_\nu-\frac{\gamma}{2}]$. With a symmetric argument on $\sigma_{\mu_r}$ we have that $\sigma_\nu \in [m_r-b_\nu+\frac{\gamma}{2}, M_r+b_\nu+2-\frac{\gamma}{2}]$. It follows that $\sigma_\nu \in [m_l-b_\nu+\frac{\gamma}{2}-2, M_l+b_\nu-\frac{\gamma}{2}]\cap[m_r-b_\nu+\frac{\gamma}{2}, M_r+b_\nu+2-\frac{\gamma}{2}]=[\max\{m_l-2,m_r\}+\frac{\gamma}{2}-b_\nu,\min\{M_l, M_r+2\}-\frac{\gamma}{2}+b_\nu] =[m-b_\nu, M+b_\nu]$.
	
	To complete the proof it remains to show that for every $\sigma_\nu\in [m-b_\nu,M+b_\nu]$,  
	there exists an orthogonal representation $H_\nu$ of $G_\nu$ with $b_\nu$ bends and with spirality $\sigma_\nu$. This is equivalent to showing that there exists a plane graph $G'_\nu$ obtained by adding $b_\nu$ subdivision vertices along some edges of $G_\nu$ such that  $G'_\nu$ has a rectilinear orthogonal representation with spirality $\sigma_\nu$. To construct $G'_\nu$ we insert a suitable number $b_l\in[0,b_\nu]$ of subdivision vertices on an exposed edge of 
	$G_{\mu_l}$ and $b_r=b_\nu-b_l$ subdivision vertices on an exposed edge of $G_{\mu_r}$. Let $G'_{\mu_l}$ and $G'_{\mu_r}$ be the resulting graphs. Since by hypothesis $G_{\mu_l}$ and $G_{\mu_r}$ are rectilinear planar, also $G'_{\mu_l}$ and $G'_{\mu_r}$ are rectilinear planar. With the same reasoning as in the proof of Property (i), the representability intervals of $G'_{\mu_l}$ and $G'_{\mu_r}$ are $[m_l-b_l,M_l+b_l]$ and $[m_r-b_r,M_r+b_r]$. We now describe how to compute $b_l$ and, consequently,~$b_r$.       
	
	Since by hypothesis $G_\nu$ is not rectilinear planar we have $[m_l-M_r, M_l-m_r] \cap \Delta_\nu = \emptyset$, i.e., $[m_l-M_r, M_l-m_r] \cap [2,4-\gamma] = \emptyset$. We consider two subcases:
	
	\begin{itemize}
		\item $M_l-m_r<2$. In this case, by Property~(i) we have $b_\nu = \delta([m_l-M_r, M_l-m_r],[2,4-\gamma]) = 2-M_l+m_r$. For example, in \cref{fi:I2O11_minbend_example-new} we have $\gamma=0$ and $b_\nu=3$. We set $b_l = \sigma_\nu - M_l + h$, where $h$ is a number (either integer or semi-integer) in the interval $[\frac{\gamma}{2},2-\frac{\gamma}{2}]$ such that $b_l \in [0,b_\nu]$. We first prove that such a value $h$ always exists. 
		Suppose first that $\sigma_\nu \in [m-b_\nu,m-b_\nu+1]$. In this case we choose $h = 2-\frac{\gamma}{2}$. This implies that $b_l = \sigma_\nu - M_l + 2 - \frac{\gamma}{2} \in [m-b_\nu-M_l+2-\frac{\gamma}{2}, m-b_\nu+1-M_l+2-\frac{\gamma}{2}]$.
		For example, in \cref{fi:I2O11_minbend_example-new} for every $\sigma_\nu\in[m-b_\nu,m-b_\nu+1]=[-2,-1]$ we set $b_l=\sigma_\nu - M_l + 2 - \frac{\gamma}{2}=\sigma_\nu+2$.
		Since $M_l-m_r<2$ we have $m_l-2\le M_l-2<m_r$ and hence $m_l-2<m_r$. Since $m = \max\{m_l-2,m_r\}+\frac{\gamma}{2}$, we have $m=m_r+\frac{\gamma}{2}$. Also, since we have $b_\nu = 2-M_l+m_r$, it follows $m-b_\nu-M_l+2-\frac{\gamma}{2}=m_r+\frac{\gamma}{2}-b_\nu-M_l+2-\frac{\gamma}{2}=0$. Hence, $b_l\in [0,1]$. Since by hypothesis $G_\nu$ is not rectilinear planar, $b_\nu \geq 1$, and therefore there exists a value of $h \in [\frac{\gamma}{2},2-\frac{\gamma}{2}]$ such that $b_l\in [0,b_\nu]$.

		Suppose now that $\sigma_\nu \in [m-b_\nu+2,M+b_\nu]$. In this case we choose $h = \frac{\gamma}{2}$. This implies that $b_l = \sigma_\nu - M_l + \frac{\gamma}{2}\in [m-b_\nu+2-M_l+\frac{\gamma}{2}, M+b_\nu-M_l+\frac{\gamma}{2}]$. For example, in \cref{fi:I2O11_minbend_example-new} $\sigma_\nu\in[m-b_\nu+2,M+b_\nu]=[0,3]$ and we set $b_l=\sigma_\nu - M_l - \frac{\gamma}{2}=\sigma_\nu$. We have $m-b_\nu+2-M_l+\frac{\gamma}{2}=m_r+\frac{\gamma}{2}-b_\nu+2-M_l+\frac{\gamma}{2}=\gamma$. It follows that $b_l\in [\gamma, M+b_\nu-M_l+\frac{\gamma}{2}]$.  Since $M_l<m_r+2\le M_r+2$ we have $M_l<M_r+2$. Also, since $M = \min\{M_l,M_r+2\}-\frac{\gamma}{2}$ it follows that $M=M_l-\frac{\gamma}{2}$. Hence, $M+b_\nu-M_l+\frac{\gamma}{2}=M_l-\frac{\gamma}{2}+b_\nu-M_l+\frac{\gamma}{2}=b_\nu$. Hence, $b_l\in [\gamma, b_\nu]$. Since $\gamma\ge 0$, also in this case there exists a value of $h \in [\frac{\gamma}{2},2-\frac{\gamma}{2}]$ such that $b_l\in [0,b_\nu]$.
		We represent $G_{\mu_l}'$ with spirality $\sigma_{\mu_l}'=M_l+b_l=M_l+\sigma_\nu - M_l + h=\sigma_\nu + c$ and $G_{\mu_r}'$ with spirality $\sigma_{\mu_r}'=m_r-b_r=m_r-(b_\nu-b_l)=m_r-b_\nu+b_l=m_r-b_\nu +\sigma_\nu - M_l + h=m_r-M_l-b_\nu+\sigma_\nu+h=b_\nu-2-b_\nu+\sigma_\nu+h=\sigma_\nu+h-2$. 
		We have $\sigma_{\mu_l}'-\sigma_{\mu_r}'=\sigma_\nu + h -(\sigma_\nu+h-2)=2$, and by \cref{le:P-2-children-support-type1}, $G_\nu'$ is rectilinear planar.
		It remains to show that $G_\nu'$ admits a rectilinear planar representation with spirality $\sigma_\nu'=\sigma_\nu$. Given the choice of $\sigma_{\mu_l}'$ and $\sigma_{\mu_r}'$, by \cref{le:spirality-P-node-2-children} every rectilinear planar representation of $G_\nu'$ has spirality $\sigma_\nu'=\sigma_{\mu_l}' - k_u^l \alpha_{u}^l -  k_u^r\alpha_{v}^l=\sigma_\nu + h - k_u^l \alpha_{u}^l -  k_u^r\alpha_{v}^l$. Since $h \in [\frac{\gamma}{2},2-\frac{\gamma}{2}]$, by \cref{le:P-2-bend-support-pio2alphabeta} there exists a value $k_u^l \alpha_{u}^l +  k_u^r\alpha_{v}^l$ such that $h- k_u^l \alpha_{u}^l -  k_u^r\alpha_{v}^l=0$, and thus $\sigma_\nu'=\sigma_\nu$. For example, in \cref{fi:I2O11_minbend_example-new}, for every $\sigma_\nu\in I_\nu = [-2,3]$, there is a rectilinear representation of~$G_\nu$ with $\sigma_{\mu_l}-\sigma_{\mu_r}=2$, $b_l$ and $b_l$ chosen as described above,~and~spirality~$\sigma_\nu$.

		\item $m_l-M_r>4-\gamma$. In this case, by Property~(i) we have $b_\nu = \delta([m_l-M_r, M_l-m_r],[2,4-\gamma]) = m_l-M_r-4+\gamma$.  We set $b_l = m_l-\sigma_\nu  - 2+\frac{\gamma}{2}$. We first prove that $b_l \in [0,b_\nu]$. We have $b_l = m_l-\sigma_\nu-2 +\frac{\gamma}{2} \in [m_l-M-b_\nu-2+\frac{\gamma}{2}, m_l-m+b_\nu-2+\frac{\gamma}{2}]$. Since $M_l\ge m_l>M_r+4-\gamma\ge M_r+2$, we have $M=M_r+2-\frac{\gamma}{2}$. Hence, $m_l-M-b_\nu-2+\frac{\gamma}{2}=m_l-(M_r+2-\frac{\gamma}{2})-b_\nu-2+\frac{\gamma}{2}=m_l-M_r-4+\gamma-b_\nu=0$. Also, since $m_l-2 \geq m_l-4+\gamma > M_r\ge m_r$, we have $m=m_l-2+\frac{\gamma}{2}$. Hence, $m_l-m+b_\nu-2+\frac{\gamma}{2}=m_l-(m_l-2+\frac{\gamma}{2})+b_\nu-2+\frac{\gamma}{2}=b_\nu$. It follows that $b_l\in [0,b_\nu]$.
		We represent $G_{\mu_l}'$ with spirality $\sigma_{\mu_l}'=m_l-b_l=m_l-(m_l-\sigma_\nu-2+\frac{\gamma}{2})=\sigma_\nu+2-\frac{\gamma}{2}$ and $G_{\mu_r}'$ with spirality $\sigma_{\mu_r}'=M_r+b_r=M_r+b_\nu-b_l=M_r+m_l-M_r-4+\gamma-m_l+\sigma_\nu+2-\frac{\gamma}{2}=\sigma_\nu+\frac{\gamma}{2}-2$. We have $\sigma_{\mu_l}'-\sigma_{\mu_r}'=\sigma_\nu+2-\frac{\gamma}{2}-(\sigma_\nu+\frac{\gamma}{2}-2)=4-\gamma$ and, by \cref{le:P-2-children-support-type1}, $G_\nu'$ is rectilinear planar.
		It remains to show that $G_\nu'$ admits a rectilinear representation with spirality $\sigma_\nu'=\sigma_\nu$. Given the choice of~$\sigma_{\mu_l}'$ and $\sigma_{\mu_r}'$, by \cref{le:spirality-P-node-2-children} every rectilinear representation of~$G_\nu'$ has spirality $\sigma_\nu'=\sigma_{\mu_l}' - k_u^l \alpha_{u}^l -  k_u^r\alpha_{v}^l=\sigma_\nu+2-\frac{\gamma}{2}- k_u^l \alpha_{u}^l -  k_u^r\alpha_{v}^l$. Since $2-\frac{\gamma}{2}\in [\frac{\gamma}{2},2-\frac{\gamma}{2}]$, by \cref{le:P-2-bend-support-pio2alphabeta} we can set $k_u^l \alpha_{u}^l +  k_u^r\alpha_{v}^l=2-\frac{\gamma}{2}$, and thus $\sigma_\nu'=\sigma_\nu$.
	\end{itemize}
	
	\smallskip\noindent{\bf Case \Pio{3d}{\lambda\beta}:} Recall that in this case $\lambda=1$ and $\beta\in\{1,2\}$. Assume, without loss of generality, that $\indeg(v)=3$ and $\indeg(u)=2$. Also, assume $d=l$ (the case $d=r$ is treated symmetrically). In this case $\Delta_\nu = [\frac{5}{2},\frac{7}{2}-\gamma]$, $m=\max\{m_l-\frac{3}{2},m_r+1\}+\frac{\gamma}{2}$, and $M=\min\{M_l-\frac{1}{2}, M_r+2\}-\frac{\gamma}{2}$. 
	We show that set $I'_\nu$ is an interval of feasible spiralities for the orthogonal representations of $G_\nu$ with $b_\nu$ bends. Suppose first that $G_\nu$ has an orthogonal representation $H_\nu$ with $b_\nu$ bends, and let $\sigma_\nu$ be the spirality of $H_\nu$. We prove that $\sigma_\nu\in [m-b_\nu,M+b_\nu]$. Let $b_l$ and $b_r$ be the number of bends in the restriction of $H_\nu$ to $G_{\mu_l}$ and to $G_{\mu_r}$, respectively, where $b_l+b_r=b_\nu$. 
	By \cref{le:spirality-P-node-2-children}, we have  $\sigma_\nu = \sigma_{\mu_l} - k_{u}^l \alpha_{u}^l - k_{v}^l \alpha_{v}^l$. 
	%
	By Relation~\ref{eq:pio3dalphabeta-l} of \cref{le:P-2-bend-support-pio3dalphabeta}, we have $-k_{u}^l \alpha_{u}^l - k_{v}^l \alpha_{v}^l\in[-\frac{3}{2}+\frac{\gamma}{2},-\frac{\gamma}{2}-\frac{1}{2}]$. 
	Hence, by the same reasoning as in the case \Pio{2}{\lambda\beta}, $\sigma_\nu\in [m_l-b_\nu-\frac{3}{2}+\frac{\gamma}{2}, M_l+b_\nu-\frac{\gamma}{2}-\frac{1}{2}]$. Also, by Relation~\ref{eq:pio3dalphabeta-r} of \cref{le:P-2-bend-support-pio3dalphabeta}, we have $k_u^r\alpha_u^r+k_v^r\alpha_v^r \in [\frac{\gamma}{2}+1, 2-\frac{\gamma}{2}]$. Hence, $\sigma_\nu \in [m_r-b_\nu+\frac{\gamma}{2}+1, M_r+b_\nu+2-\frac{\gamma}{2}]$. 
	It follows that $\sigma_\nu \in[m_l-b_\nu-\frac{3}{2}+\frac{\gamma}{2}, M_l+b_\nu-\frac{\gamma}{2}-\frac{1}{2}]\cap [m_r-b_\nu+1+\frac{\gamma}{2}, M_r+b_\nu+2-\frac{\gamma}{2}]
	=[\max\{m_l-\frac{3}{2},m_r+1\}+\frac{\gamma}{2}-b_\nu,\min\{M_l-\frac{1}{2}, M_r+2\}-\frac{\gamma}{2}+b_\nu] =[m-b_\nu, M+b_\nu]$.
	
	It remains to show that for every $\sigma_\nu\in [m-b_\nu,M+b_\nu]$,  
	there exists an orthogonal representation $H_\nu$ of $G_\nu$ with $b_\nu$ bends and with spirality $\sigma_\nu$. With the same notation as in the previous case, we show how to compute the values $b_l$ and $b_r$. Also, we show how to compute the spirality for the rectilinear planar representations of $G'_{\mu_l}$ and~$G'_{\mu_r}$, within their representability intervals $[m_l-b_l,M_l+b_l]$ and $[m_r-b_r,M_r+b_r]$.      
	Since by hypothesis $G_\nu$ is not rectilinear planar we have $[m_l-M_r, M_l-m_r] \cap \Delta_\nu = \emptyset$, i.e., $[m_l-M_r, M_l-m_r] \cap [\frac{5}{2},\frac{7}{2}-\gamma] = \emptyset$. We analyze the two possible subcases:

	\begin{itemize}
		\item 	$M_l-m_r<\frac{5}{2}$. By Property~(i), $b_\nu = \delta([m_l-M_r, M_l-m_r],[\frac{5}{2},\frac{7}{2}-\gamma]) = \frac{5}{2}-M_l+m_r$.  We set $b_l = \sigma_\nu - M_l + h_l$, where $h_l$ is a number (either integer or semi-integer) in the interval $[\frac{\gamma}{2}+\frac{1}{2},\frac{3}{2}-\frac{\gamma}{2}]$ such that $b_l \in [0,b_\nu]$. We first prove that such a value $h_l$ always exists for any given $\sigma_\nu \in [m-b_\nu, M+b_\nu]$. If $\sigma_\nu =m-b_\nu$, we choose $h_l = \frac{3}{2}-\frac{\gamma}{2}$. This implies that $b_l =m-b_\nu-M_l+\frac{3}{2}-\frac{\gamma}{2}$. Since $M_l-m_r<\frac{5}{2}$, we have $m_l-\frac{5}{2}\le M_l-\frac{5}{2}<m_r$, and therefore $m_l-\frac{3}{2}<m_r+1$. Since $m = \max\{m_l-\frac{3}{2},m_r+1\}+\frac{\gamma}{2}$, we have $m=m_r+1+\frac{\gamma}{2}$. Also, since $b_\nu = \frac{5}{2}-M_l+m_r$, we have $b_l = m-b_\nu-M_l+\frac{3}{2}-\frac{\gamma}{2}=m_r+1+\frac{\gamma}{2}-b_\nu-M_l+\frac{3}{2}-\frac{\gamma}{2}=0$.  If $\sigma_\nu \in [m-b_\nu+1,M+b_\nu]$, we choose $h_l = \frac{\gamma}{2}+\frac{1}{2}$. This implies that $b_l = \sigma_\nu - M_l + \frac{\gamma}{2}+\frac{1}{2}\in [m-b_\nu+1-M_l+\frac{\gamma}{2}+\frac{1}{2}, M+b_\nu-M_l+\frac{\gamma}{2}+\frac{1}{2}]$. We have $m-b_\nu+1-M_l+\frac{\gamma}{2}+\frac{1}{2}=m_r+1+\frac{\gamma}{2}-b_\nu+1-M_l+\frac{\gamma}{2}+\frac{1}{2}=(\frac{5}{2}+m_r-M_l)-b_\nu + \gamma=\gamma$. 
		Since $M_l<m_r+\frac{5}{2}\le M_r+\frac{5}{2}$ we have $M_l-\frac{1}{2}<M_r+2$. Also, since $M = \min\{M_l-\frac{1}{2},M_r+2\}-\frac{\gamma}{2}$ it follows that $M=M_l-\frac{1}{2}-\frac{\gamma}{2}$. Hence, $M+b_\nu-M_l+\frac{\gamma}{2}+\frac{1}{2}=M_l-\frac{1}{2}-\frac{\gamma}{2}+b_\nu-M_l+\frac{\gamma}{2}+\frac{1}{2}=b_\nu$. It follows that $b_l\in [\gamma, b_\nu]$. Since $\gamma\ge 0$, also in this case there exists a value of  $h_l \in [\frac{\gamma}{2}+\frac{1}{2},\frac{3}{2}-\frac{\gamma}{2}]$ such that $b_l\in [0,b_\nu]$.
		We represent $G_{\mu_l}'$ with spirality $\sigma_{\mu_l}'=M_l+b_l=M_l+\sigma_\nu - M_l + h_l=\sigma_\nu + h_l$ and $G_{\mu_r}'$ with spirality $\sigma_{\mu_r}'=m_r-b_r=m_r-(b_\nu-b_l)=m_r-b_\nu+b_l=m_r-b_\nu +\sigma_\nu - M_l + h_l=b_\nu-\frac{5}{2}-b_\nu+\sigma_\nu+h_l=\sigma_\nu+h_l-\frac{5}{2}$. We have $\sigma_{\mu_l}'-\sigma_{\mu_r}'=\sigma_\nu + h_l-(\sigma_\nu+h_l-\frac{5}{2})=\frac{5}{2}$ and, by \cref{le:P-2-children-support-type2}, $G_\nu'$ is rectilinear planar.
		It remains to show that $G_\nu'$ admits a rectilinear planar representation with spirality $\sigma_\nu'=\sigma_\nu$. Given the choice of $\sigma_{\mu_l}'$ and $\sigma_{\mu_r}'$, by \cref{le:spirality-P-node-2-children} every rectilinear planar representation of $G_\nu'$ has spirality $\sigma_\nu'=\sigma_{\mu_l}' - k_u^l \alpha_{u}^l -  k_u^l\alpha_{v}^l=\sigma_\nu + h_l - k_u^l \alpha_{u}^l -  k_u^r\alpha_{v}^l$. Since $h_l \in [\frac{\gamma}{2}+\frac{1}{2},\frac{3}{2}-\frac{\gamma}{2}]$, by Relation~\ref{eq:pio3dalphabeta-l} of \cref{le:P-2-bend-support-pio3dalphabeta} there exists a value $ k_u^l \alpha_{u}^l +  k_u^r\alpha_{v}^l$ such that $h_l- k_u^l \alpha_{u}^l -  k_u^r\alpha_{v}^l=0$, and thus $\sigma_\nu'=\sigma_\nu$. 
		
		\item $m_l-M_r>\frac{7}{2}-\gamma$. In this case, by Property~(i) we have $b_\nu = m_l-M_r-\frac{7}{2}+\gamma$.  We set $b_l = m_l-\sigma_\nu  - \frac{3}{2}+\frac{\gamma}{2}$. As before, we first prove that $b_l \in [0,b_\nu]$. We have $b_l \in [m_l-M-b_\nu-\frac{3}{2}+\frac{\gamma}{2}, m_l-m+b_\nu-\frac{3}{2}+\frac{\gamma}{2}]$. 
		Since $M_l\ge m_l>M_r+\frac{7}{2}-\gamma\ge M_r+\frac{5}{2}$, we have $M_l-\frac{1}{2} \ge M_r+2$. It follows that $M=M_r+2-\frac{\gamma}{2}$. Hence, $m_l-M-b_\nu-\frac{3}{2}+\frac{\gamma}{2}=m_l-(M_r+2-\frac{\gamma}{2})-b_\nu-\frac{3}{2}+\frac{\gamma}{2}=m_l-M_r-\frac{7}{2}+\gamma-b_\nu=0$. 
		Also, since $m_l-\frac{5}{2} \geq m_l-\frac{7}{2}+\gamma > M_r$ we have $m_l-\frac{3}{2} \ge M_r+1\ge m_r+1$. It follows that $m=m_l-\frac{3}{2}+\frac{\gamma}{2}$. Hence, $m_l-m+b_\nu-\frac{3}{2}+\frac{\gamma}{2}=m_l-(m_l-\frac{3}{2}+\frac{\gamma}{2})+b_\nu-\frac{3}{2}+\frac{\gamma}{2}=b_\nu$. It follows that $b_l\in [0,b_\nu]$. We represent $G_{\mu_l}'$ with spirality $\sigma_{\mu_l}'=m_l-b_l=m_l-(m_l-\sigma_\nu-\frac{3}{2}+\frac{\gamma}{2})=\sigma_\nu+\frac{3}{2}-\frac{\gamma}{2}$ and $G_{\mu_r}'$ with spirality $\sigma_{\mu_r}'=M_r+b_r=M_r+b_\nu-b_l=M_r+m_l-M_r-\frac{7}{2}+\gamma-m_l+\sigma_\nu+\frac{3}{2}-\frac{\gamma}{2}=\sigma_\nu+\frac{\gamma}{2}-2$. 
		We have $\sigma_{\mu_l}'-\sigma_{\mu_r}'=\sigma_\nu+\frac{3}{2}-\frac{\gamma}{2}-(\sigma_\nu+\frac{\gamma}{2}-2)=\frac{7}{2}-\gamma$ and, by \cref{le:P-2-children-support-type2}, $G_\nu'$ is rectilinear planar.
		Also, by \cref{le:spirality-P-node-2-children} every rectilinear planar representation of $G_\nu'$ has spirality $\sigma_\nu'=\sigma_{\mu_l}' - k_u^l \alpha_{u}^l -  k_u^l\alpha_{v}^l=\sigma_\nu+\frac{3}{2}-\frac{\gamma}{2}- k_u^l \alpha_{u}^l -  k_u^l\alpha_{v}^l$. By Relation~\ref{eq:pio3dalphabeta-l} of \cref{le:P-2-bend-support-pio3dalphabeta}, we can set $ k_u^l \alpha_{u}^l +  k_u^r\alpha_{v}^l=\frac{3}{2}-\frac{\gamma}{2}$, and thus $\sigma_\nu'=\sigma_\nu$.
	\end{itemize}
	
	\smallskip\noindent{\bf Case \Pin{3dd'}:} Assume that $d=l$ (the case $d=r$ is symmetric). In this case $k_v^l=\frac{1}{2}$, $k_v^r=1$, $\Delta_\nu = [3,3]$, $m=\max\{m_l-1,m_r+2\}-\frac{\phi(d')}{2}$, and $M=\min\{M_l-1, M_r+2\}-\frac{\phi(d')}{2}$. 
	Since both $u$ and $v$ have degree four, in any rectilinear representation of $G_{\mu_l}$ and $G_{\mu_l}$ we have $\alpha_{u}^l=\alpha_{u}^r=\alpha_{v}^l=\alpha_{v}^r=1$. If $d'=l$ then $k_u^l=\frac{1}{2}$ and $k_u^r=1$; hence $k_{u}^l \alpha_{u}^l + k_{v}^l \alpha_{v}^l=1$ and  $k_{u}^r \alpha_{u}^r + k_{v}^r \alpha_{v}^r=2$. 
	If $d'=r$ then $k_u^l=1$ and $k_u^r=\frac{1}{2}$; hence $k_{u}^l \alpha_{u}^l + k_{v}^l \alpha_{v}^l=k_{u}^r \alpha_{u}^r + k_{v}^r \alpha_{v}^r=\frac{3}{2}$.
	We show that set $I'_\nu$ is an interval of feasible spiralities for the orthogonal representations of~$G_\nu$ with~$b_\nu$ bends. 
	Suppose first that $G_\nu$ has an orthogonal representation $H_\nu$ with $b_\nu$ bends, and let $\sigma_\nu$ be the spirality of $H_\nu$. We prove that $\sigma_\nu\in [m-b_\nu,M+b_\nu]$. Let $b_l$ and $b_r$ be the number of bends in the restriction of $H_\nu$ to $G_{\mu_l}$ and to $G_{\mu_r}$, respectively, where $b_l+b_r=b_\nu$.
	By \cref{le:spirality-P-node-2-children}, we have $\sigma_\nu = \sigma_{\mu_l} -k_{u}^l \alpha_{u}^l - k_{v}^l \alpha_{v}^l=\sigma_{\mu_r} +k_{u}^r \alpha_{u}^r + k_{v}^r \alpha_{v}^r$. 
	Suppose first $d'=l$. In this case $\sigma_\nu = \sigma_{\mu_l} -1=\sigma_{\mu_r} +2$. Hence, by the same reasoning as in the case \Pio{2}{\lambda\beta}, $\sigma_\nu \in [m_l-b_\nu-1, M_l+b_\nu-1]$ and $\sigma_\nu\in [m_r-b_\nu+2,M_r+b_\nu+2]$. Therefore,  $\sigma_\nu \in [m_l-b_\nu-1, M_l+b_\nu-1]  \cap [m_r-b_\nu+2,M_r+b_\nu+2]=[\max\{m_l-1,m_r+2\}-b_\nu, \min\{M_l-1,M_r+2\}+b_\nu]$. Since $d'=l$, we have $\frac{\phi(d')}{2}=0$ and $\sigma_\nu\in [\max\{m_l-1,m_r+2\}-b_\nu, \min\{M_l-1,M_r+2\}+b_\nu]=[m-b_\nu, M+b_\nu]$.
	Suppose now $d'=r$. In this case $\sigma_\nu = \sigma_{\mu_l} -\frac{3}{2}=\sigma_{\mu_r} +\frac{3}{2}$. Hence, $\sigma_\nu \in [m_l-b_\nu-\frac{3}{2}, M_l+b_\nu-\frac{3}{2}] \cap [m_r-b_\nu+\frac{3}{2},M_r+b_\nu+\frac{3}{2}]=[\max\{m_l-\frac{3}{2},m_r+\frac{3}{2}\}-b_\nu, \min\{M_l-\frac{3}{2},M_r+\frac{3}{2}\}+b_\nu]=[\max\{m_l-1,m_r+2\}-\frac{1}{2}-b_\nu, \min\{M_l-1,M_r+2\}-\frac{1}{2}+b_\nu]$. Since $d'=r$, we have $\frac{\phi(d')}{2}=\frac{1}{2}$ and $\sigma_\nu\in [\max\{m_l-1,m_r+2\}-\frac{1}{2}-b_\nu, \min\{M_l-1,M_r+2\}-\frac{1}{2}+b_\nu]=[m-b_\nu, M+b_\nu]$.
	
	It remains to show that for every $\sigma_\nu\in [m-b_\nu,M+b_\nu]$,  
	there exists an orthogonal representation $H_\nu$ of $G_\nu$ with $b_\nu$ bends and with spirality $\sigma_\nu$. With the same notation as in the previous cases, we show how to compute the values $b_l$ and $b_r$. Also, we show how to compute the spirality for the rectilinear planar representations of $G'_{\mu_l}$ and~$G'_{\mu_r}$, within their representability intervals $[m_l-b_l,M_l+b_l]$ and $[m_r-b_r,M_r+b_r]$.    
	Since by hypothesis $G_\nu$ is not rectilinear planar we have $[m_l-M_r, M_l-m_r] \cap \Delta_\nu = \emptyset$, i.e., $[m_l-M_r, M_l-m_r] \cap [3,3] = \emptyset$. We analyze the two possible subcases:
	\begin{itemize}
		\item 	$M_l-m_r<3$. By Property~(i), $b_\nu = \delta([m_l-M_r, M_l-m_r],[3,3]) = 3-M_l+m_r$.
		Suppose first that $d'=l$. By setting $b_l = \sigma_\nu - M_l + 1$, we have $b_l \in [0,b_\nu]$. Namely, $b_l\in  [m-b_\nu - M_l + 1, M+b_\nu - M_l + 1]$. Since $M_l<m_r+3\le M_r+3$ we have $M_l-1<M_r+2$. Also, since $M=\min\{M_l-1, M_r+2\}$, we have $M=M_l-1$. Since $m_l-3\le M_l-3<m_r$ we have $m_l-1<m_r+2$. Also, since $m=\max\{m_l-1,m_r+2\}$, we have $m=m_r+2$. Hence, $b_l\in  [m_r+2-b_\nu - M_l + 1, M_l-1+b_\nu - M_l + 1]=[0,b_\nu]$. 
		We represent $G_{\mu_l}'$ with spirality $\sigma_{\mu_l}'=M_l+b_l=M_l+\sigma_\nu - M_l + 1=\sigma_\nu+1$ and $G_{\mu_r}'$ with spirality $\sigma_{\mu_r}'=m_r-b_r=m_r-(b_\nu-b_l)=m_r-b_\nu +\sigma_\nu - M_l + 1=(m_r-M_l+1)-b_\nu+\sigma_\nu=b_\nu-2-b_\nu+\sigma_\nu=\sigma_\nu-2$. We have $\sigma_{\mu_l}'-\sigma_{\mu_r}'=\sigma_\nu+1-(\sigma_\nu-2)=3$ and, by \cref{le:P-2-children-support-type3}, $G_\nu'$ is rectilinear planar.
		Also, given the choice of $\sigma_{\mu_l}'$ and $\sigma_{\mu_r}'$, by \cref{le:spirality-P-node-2-children} every rectilinear planar representation of $G_\nu'$ has spirality $\sigma_\nu'=\sigma_{\mu_l}' - k_u^l \alpha_{u}^l -  k_u^l\alpha_{v}^l=\sigma_\nu + 1 - 1=\sigma_\nu$. 
		%
		%
		
		Suppose now that $d'=r$. By setting $b_l = \sigma_\nu - M_l + \frac{3}{2}$ we have $b_l \in [0,b_\nu]$. Namely, $b_l\in  [m-b_\nu - M_l + \frac{3}{2}, M+b_\nu - M_l + \frac{3}{2}]$. Since $M_l < M_r+3$ we have $M_l-\frac{3}{2}<M_r+\frac{3}{2}$. Also, since $M=\min\{M_l-\frac{3}{2}, M_r+\frac{3}{2}\}$, we have $M=M_l-\frac{3}{2}$. Since $m_l-3 < m_r$ we have $m_l-\frac{3}{2}<m_r+\frac{3}{2}$. This implies that $m=\max\{m_l-\frac{3}{2},m_r+\frac{3}{2}\}=m_r+\frac{3}{2}$. Hence, $b_l\in  [m_r+\frac{3}{2}-b_\nu - M_l + \frac{3}{2}, M_l-\frac{3}{2}+b_\nu - M_l + \frac{3}{2}]=[0,b_\nu]$. 
		We represent $G_{\mu_l}'$ with spirality $\sigma_{\mu_l}'=M_l+b_l=M_l+\sigma_\nu - M_l + \frac{3}{2}=\sigma_\nu+\frac{3}{2}$ and $G_{\mu_r}'$ with spirality $\sigma_{\mu_r}'=m_r-b_r=m_r-(b_\nu-b_l)=m_r-b_\nu +\sigma_\nu - M_l + \frac{3}{2}=(m_r-M_l+\frac{3}{2})-b_\nu+\sigma_\nu=b_\nu-\frac{3}{2}-b_\nu+\sigma_\nu=\sigma_\nu-\frac{3}{2}$. We have $\sigma_{\mu_l}'-\sigma_{\mu_r}'=\sigma_\nu+\frac{3}{2}-(\sigma_\nu-\frac{3}{2})=3$ and, by \cref{le:P-2-children-support-type3}, $G_\nu'$ is rectilinear planar.
		Also, by \cref{le:spirality-P-node-2-children} every rectilinear planar representation of $G_\nu'$ has spirality $\sigma_\nu'=\sigma_{\mu_l}' - k_u^l \alpha_{u}^l -  k_u^l\alpha_{v}^l=\sigma_\nu + \frac{3}{2} - \frac{3}{2}=\sigma_\nu$. 
		
		\item  $m_l-M_r>3$. In this case, by Property~(i) we have $b_\nu = m_l-M_r-3$. 
		Suppose first that $d'=l$. By setting $b_l = m_l-\sigma_\nu  - 1$, we have $b_l \in [0,b_\nu]$. Namely, $b_l \in[m_l-b_\nu-M-1,m_l-m+b_\nu-1]$. Since $M_l\ge m_l>M_r+3$, we have $M_l-1 \ge M_r+2$. This implies that $M=\min\{M_l-1, M_r+2\}=M_r+2$. Since $m_l>M_r+3\ge m_r+2$, we have $m_l-1 \ge m_r+2$, which implies that $m=\max\{m_l-1, m_r+2\}=m_l-1$.  Hence, $b_l \in [m_l-b_\nu-M_r-2-1,m_l-m_l-1+b_\nu+1]=[0,b_\nu]$. 
		We represent $G_{\mu_l}'$ with spirality $\sigma_{\mu_l}'=m_l-b_l=m_l-m_l+\sigma_\nu  + 1=\sigma_\nu+1$ and $G_{\mu_r}'$ with spirality $\sigma_{\mu_r}'=M_r+b_r=M_r+b_\nu-b_l=M_r+b_\nu -(m_l-\sigma_\nu  - 1)=M_r+b_\nu-m_l+\sigma_\nu+1=\sigma_\nu+b_\nu-b_\nu-2=\sigma_\nu-2$. 
		We have $\sigma_{\mu_l}'-\sigma_{\mu_r}'=\sigma_\nu+1-(\sigma_\nu-2)=3$ and, by \cref{le:P-2-children-support-type3}, $G_\nu'$ is rectilinear planar.
		Also, by \cref{le:spirality-P-node-2-children} every rectilinear planar representation of $G_\nu'$ has spirality $\sigma_\nu'=\sigma_{\mu_l}' - k_u^l \alpha_{u}^l -  k_u^l\alpha_{v}^l=\sigma_\nu + 1 - 1=\sigma_\nu$. 

		Suppose now that $d'=r$. By setting $b_l = m_l-\sigma_\nu  - \frac{3}{2}$ we have $b_l \in [0,b_\nu]$. Namely, $b_l \in[m_l-b_\nu-M-\frac{3}{2},m_l-m+b_\nu-\frac{3}{2}]$. Since $M_l > M_r+3$, we have $M_l-\frac{3}{2} \ge M_r+\frac{3}{2}$. This implies that $M=\min\{M_l-\frac{3}{2}, M_r+\frac{3}{2}\}=M_r+\frac{3}{2}$. Since $m_l> m_r+3$, we have $m_l-\frac{3}{2} \ge m_r+\frac{3}{2}$. This implies that $m=\max\{m_l-\frac{3}{2}, m_r+\frac{3}{2}\}=m_l-\frac{3}{2}$.  Hence, $b_l \in[m_l-b_\nu-M-\frac{3}{2},m_l-m+b_\nu-\frac{3}{2}]=[m_l-b_\nu-M_r-\frac{3}{2}-\frac{3}{2},m_l-m_l-\frac{3}{2}+b_\nu+\frac{3}{2}]=[0,b_\nu]$. 
		We represent $G_{\mu_l}'$ with spirality $\sigma_{\mu_l}'=m_l-b_l=m_l-m_l+\sigma_\nu  + \frac{3}{2}=\sigma_\nu+\frac{3}{2}$ and $G_{\mu_r}'$ with spirality $\sigma_{\mu_r}'=M_r+b_r=M_r+b_\nu-b_l=M_r+b_\nu -(m_l-\sigma_\nu  - \frac{3}{2})=\sigma_\nu+b_\nu-b_\nu-\frac{3}{2}=\sigma_\nu-\frac{3}{2}$. 
		We have $\sigma_{\mu_l}'-\sigma_{\mu_r}'=\sigma_\nu+\frac{3}{2}-(\sigma_\nu-\frac{3}{2})=3$ and, by \cref{le:P-2-children-support-type3}, $G_\nu'$ is rectilinear planar.
		Also, by \cref{le:spirality-P-node-3-children} every rectilinear planar representation of $G_\nu'$ has spirality $\sigma_\nu'=\sigma_{\mu_l}' - k_u^l \alpha_{u}^l -  k_u^l\alpha_{v}^l=\sigma_\nu + \frac{3}{2} - \frac{3}{2}=\sigma_\nu$. 
	\end{itemize}
\end{proof}

\subsubsection{P-nodes having a child with no exposed edges}\label{ssse:no-exposed}
We now consider the case of a P-node $\nu$ having a child that does not contain an exposed edge (see \cref{le:P-2-without-exposed-edges-left,le:P-2-without-exposed-edges-right}). Denote by $\mu$ such a child node; $\mu$ is an S-node whose children $\nu_1, \nu_2, \dots, \nu_h$ ($h \geq 2$) are all P-nodes of type \Pio{2}{22}. For a given integer value $b \geq 0$, denote by $\sigma^{\max}_{\nu_i}(b)$ the \emph{maximum} value of spirality that any orthogonal representation of $G_{\nu_i}$  with at most $b$ bends can have ($i \in \{1, \dots, h\}$). When we consider the value $b+1$, the value $\sigma^{\max}_{\nu_i}(b + 1)$ may or may not increase by one unit with respect to $\sigma^{\max}_{\nu_i}(b)$, i.e., either $\sigma^{\max}_{\nu_i}(b + 1)=\sigma^{\max}_{\nu_i}(b)+1$ or $\sigma^{\max}_{\nu_i}(b + 1)=\sigma^{\max}_{\nu_i}(b)$. By plotting the function $\sigma^{\max}_{\nu_i}(b)$ as $b$ increases, we define the \emph{positive flexibility breakpoint of $G_{\nu_i}$} as the maximum number of bends $b^+_{\nu_i}$ such that for every non-negative integer $b < b^+_{\nu_i}$, we have $\sigma^{\max}_{\nu_i}(b + 1)=\sigma^{\max}_{\nu_i}(b)+1$. For example, \cref{fi:positive-flexibility} shows a component $G_{\nu_i}$ and a plot of the function $\sigma^{\max}_{\nu_i}(b)$ for $b = 0, 1,\dots, 5$. In the figure, the positive flexibility breakpoint is $b^+_{\nu_i} = 3$ because passing from $b=3$ to $b=4$ does not allow us to have an orthogonal representation of $G_{\nu_i}$ with a larger value of spirality.   
The \emph{positive flexibility breakpoint of the S-node} $\mu$ is denoted as $b^+_{\mu}$ and it is defined as the sum of the positive flexibility breakpoints of its children, i.e., $b^+_{\mu} = \sum_{i=1}^{h} b^+_{\nu_i}$.

\begin{figure}[t]
	\centering
	\includegraphics[width=0.8\columnwidth]{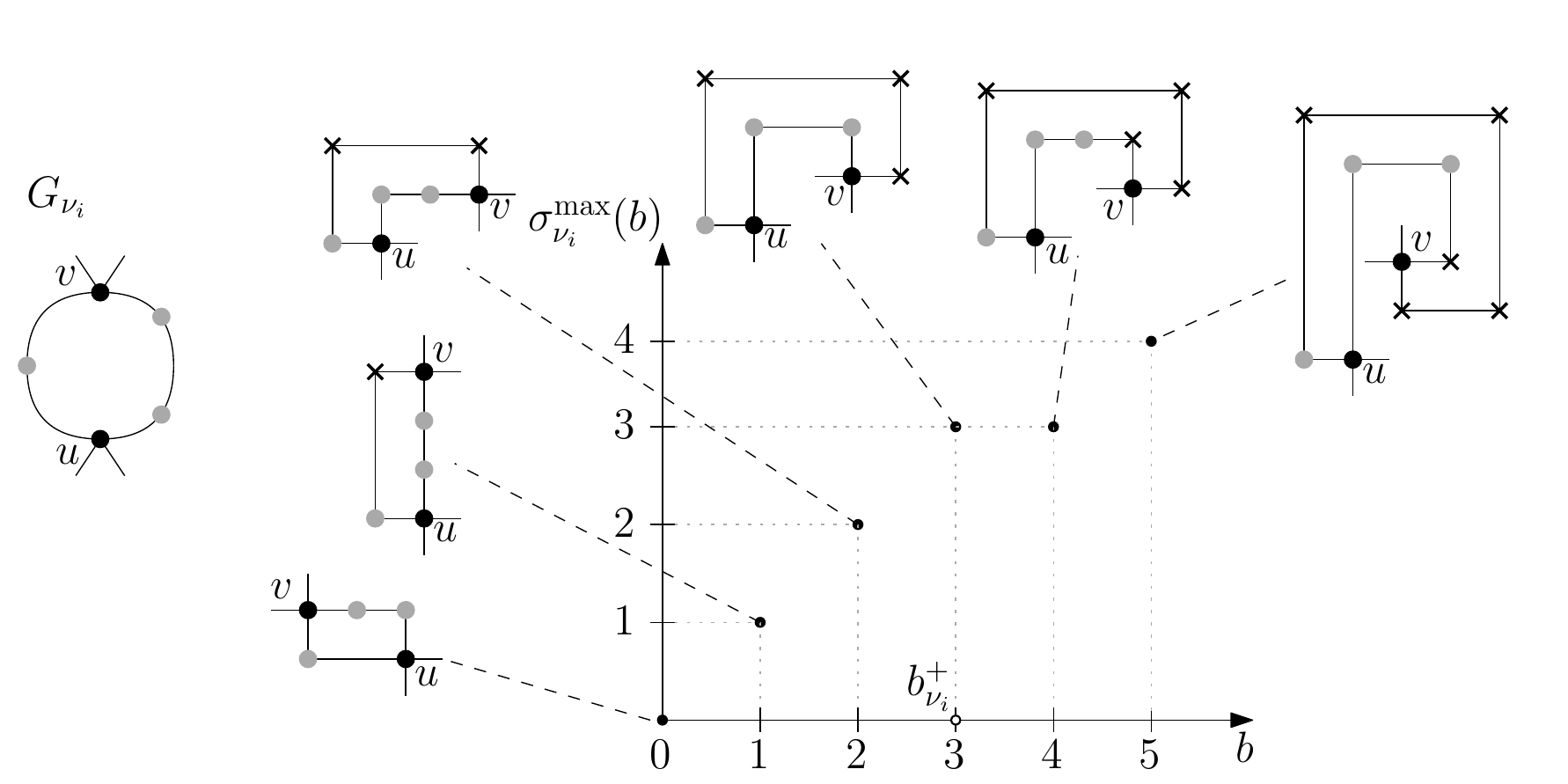}
	\caption{A P-component $G_{\nu_i}$ of type \Pio{2}{22} and a plot of function $\sigma^{\max}_{\nu_i}(b)$ for $b = 0, 1,\dots, 5$. For each value of $b$, the figure depicts an orthogonal representation of (maximum) spirality $\sigma^{\max}_{\nu_i}(b)$ among those with $b$ bends (cross vertices). The positive flexibility breakpoint is $b^+_{\nu_i} = 3$.}\label{fi:positive-flexibility}
\end{figure}

Symmetrically, for a given integer $b \geq 0$, let $\sigma^{\min}_{\nu_i}(b)$ be the \emph{minimum} value of spirality that an orthogonal representation of $G_{\nu_i}$ with at most $b$ bends can have ($i \in \{1, \dots, h\}$). The value of $\sigma^{\min}_{\nu_i}(b+1)$ may or may not decrease by one unit with respect to $\sigma^{\min}_{\nu_i}(b)$. We define the \emph{negative flexibility breakpoint of $G_{\nu_i}$} as the maximum number of bends $b^-_{\nu_i}$ such that for every non-negative integer $b < b^-_{\nu_i}$, we have $\sigma^{\min}_{\nu_i}(b + 1)=\sigma^{\min}_{\nu_i}(b)-1$. For example, \cref{fi:negative-flexibility} shows the same component $G_{\nu_i}$ as in \cref{fi:positive-flexibility} and a plot of the function $\sigma^{\min}_{\nu_i}(b)$ for $b = 0, 1,\dots, 4$. In the figure, the negative flexibility breakpoint is $b^-_{\nu_i} = 1$ because passing from $b=1$  to $b=2$ does not allow us to have an orthogonal representation of $G_{\nu_i}$ with a smaller value of spirality. The \emph{negative flexibility breakpoint of the S-node} $\mu$ is denoted as $b^-_{\mu}$ and it is defined as the sum of the negative flexibility breakpoints of its children, i.e., $b^-_{\mu} = \sum_{i=1}^{h} b^-_{\nu_i}$. 
The next lemma establishes how to compute in constant time the positive and negative flexibility breakpoints for a P-node $\nu_i$ of type \Pio{2}{22}, given the representability intervals of its children.

\begin{figure}[t]
	\centering
	\includegraphics[width=0.8\columnwidth]{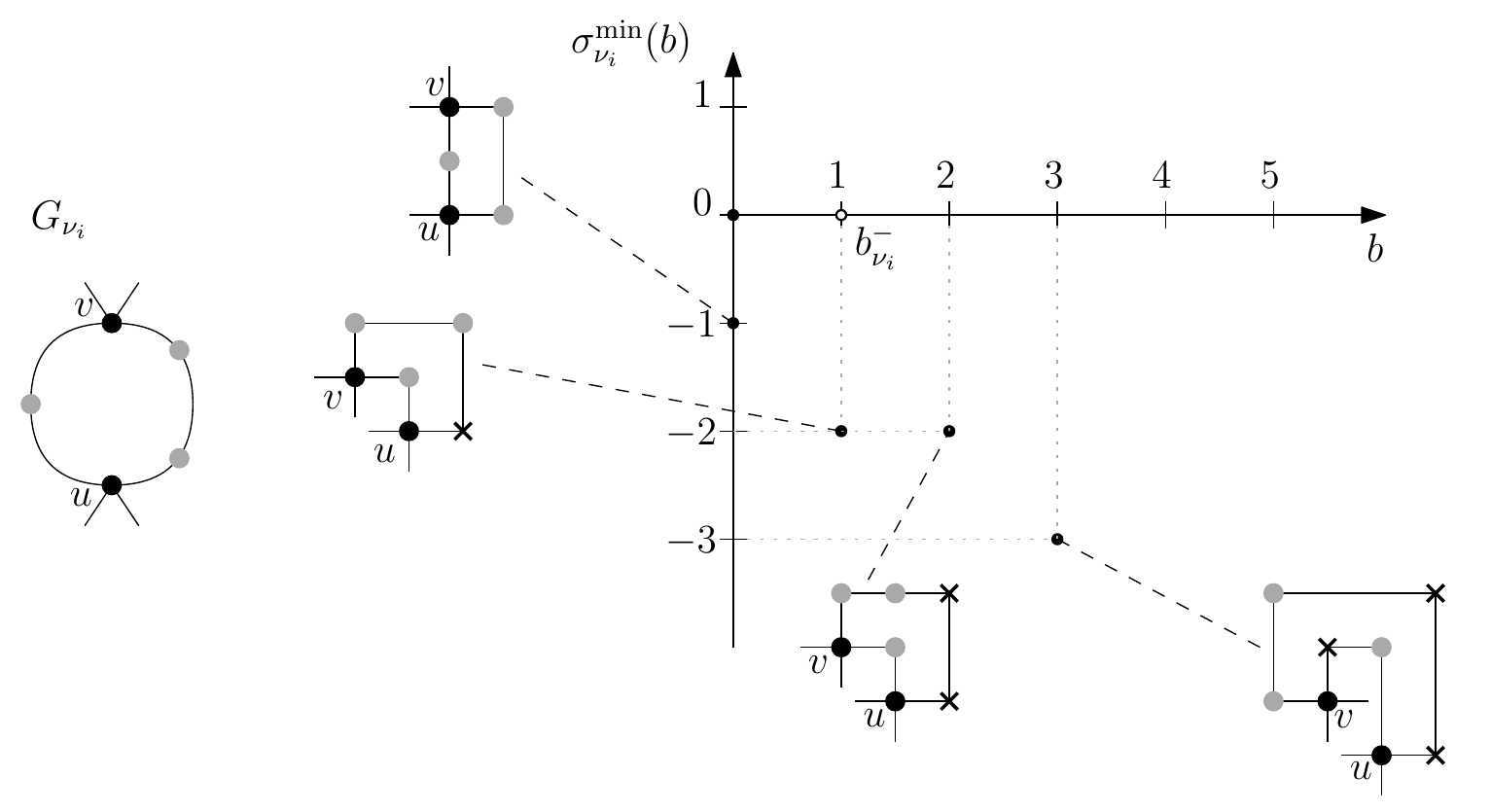}
	\caption{A P-component $G_{\nu_i}$ of type \Pio{2}{22} and a plot of function $\sigma^{\min}_{\nu_i}(b)$ for $b = 0, 1,\dots, 4$. For each value of $b$, the figure depicts an orthogonal representation of (minimum) spirality $\sigma^{\min}_{\nu_i}(b)$ among those with $b$ bends (cross vertices). The negative flexibility breakpoint is $b^-_{\nu_i} = -1$.}\label{fi:negative-flexibility}
\end{figure}

\begin{lemma}\label{le:flexibility_breakpoints}
	Let $\nu_i$ be a $P$-node  of type \Pio{2}{22} with children $\mu_l$ and $\mu_r$ such that $G_{\nu_i}$ is rectilinear planar. Let $I_{\mu_l}=[m_l,M_l]$ and $I_{\mu_r}=[m_r,M_r]$ be the representability intervals of $G_{\mu_l}$ and $G_{\mu_r}$, respectively. We have $b^+_{\nu_i} = \lvert M_r+2-M_l \rvert$ and $b^-_{\nu_i} = \lvert m_l-2-m_r \rvert$. Also, an orthogonal representation of $G_{\nu_i}$ with spirality $\sigma^{\max}_{\nu_i}(b^+_{\nu_i})$ $($resp. with $\sigma^{\min}_{\nu_i}(b^-_{\nu_i})$$)$ can be obtained by inserting all the bends on an exposed edge of either $G_{\mu_l}$ or $G_{\mu_r}$.   
\end{lemma}
\begin{proof}
	We prove that  $b^+_{\nu_i} = \lvert M_r+2-M_l \rvert$. The proof that  $b^-_{\nu_i} = \lvert m_l-m_r-2 \rvert$ is symmetric. Consider a rectilinear planar representation $H_{\nu_i}$ of $G_{\nu_i}$ with maximum value of spirality, that is spirality $\sigma^{\max}_{\nu_i}(0)$ which, 
	by \cref{ta:representability}  
	is $\sigma^{\max}_{\nu_i}(0)=\min\{M_l-1,M_r+1\}$. Let $\sigma_{\mu_l}$ and $\sigma_{\mu_r}$ be the spiralities of the left and the right orthogonal components of  $H_{\nu_i}$, respectively. By \cref{le:P-2-children-support-type1}, we know that $\sigma_{\mu_l} - \sigma_{\mu_r} = 2$. Also, by \cref{le:spirality-P-node-2-children}, $\sigma^{\max}_{\nu_i}(0)=\sigma_{\mu_l}-1=\sigma_{\mu_r}+1$. By 
	\cref{ta:representability}, either $\sigma_{\mu_l} = M_l$ or $\sigma_{\mu_r} = M_r$ (possibly both). Three cases are possible:
	
	\begin{itemize}
		\item $\sigma_{\mu_l} = M_l$ and $\sigma_{\mu_r} = M_r$. Note that, since $\sigma_{\mu_l} = \sigma_{\mu_r}+2$, in this case we have $\lvert M_r+2-M_l \rvert=0$. 
		We show that if we are allowed to subdivide an edge of $G_{\nu_i}$ with exactly one degree-2 vertex, so to obtain a graph $G'_{\nu_i}$, the maximum value of spirality that a rectilinear planar representation $H'_{\nu_i}$ of $G'_{\nu_i}$ equals $\sigma^{\max}_{\nu_i}(0)$. In other words, we show that $\sigma^{\max}_{\nu_i}(1)=\sigma^{\max}_{\nu_i}(0)$.     
		By \cref{le:P-exposed-edges-support}, we can assume that the degree-2 vertex subdivides an exposed edge, which belongs either to $G_{\mu_l}$ or to $G_{\mu_r}$. Suppose that this degree-2 vertex is inserted along an exposed edge of $G_{\mu_l}$. Denote by $G'_{\mu_l}$ and $G'_{\mu_r}$ the left component and the right component of $G'_{\nu_i}$, respectively. The maximum spirality of a rectilinear planar representation of $G'_{\mu_l}$ is $M_l+1$, while the maximum spirality of a rectilinear planar representation of $G'_{\mu_r}$ remains $M_r$, because $G'_{\mu_r}$ coincides with $G_{\mu_r}$. Assume for a contradiction that $G'_{\nu_i}$ admits a rectilinear planar representation $H'_{\nu_i}$ with spirality $\sigma^{\max}_{\nu_i}(0)+1$, and denote by $\sigma'_{\mu_l}$ and $\sigma'_{\mu_r}$ the spiralities of the left and right components of $H'_{\nu_i}$. We should have that $\sigma^{\max}_{\nu_i}(0)+1 = \sigma'_{\mu_l} - 1$ and $\sigma^{\max}_{\nu_i}(0)+1 = \sigma'_{\mu_r} + 1$. Now, $\sigma^{\max}_{\nu_i}(0)+1=M_l$ and $\sigma^{\max}_{\nu_i}(0)+1=(M_r+1)+1$, which implies that $G'_{\mu_r}$ should have a rectilinear planar representation with spirality $M_r+1$, a contradiction. A symmetric argument applies if we subdivide an exposed edge of $G_{\mu_r}$. Hence, $b^+_{\nu_i} = 0 = \lvert M_r+2-M_l \rvert$.

		\item $\sigma_{\mu_l} = M_l$ and $\sigma_{\mu_r} < M_r$. Subdivide an exposed edge of $G_{\nu_l}$ with a degree 2-vertex, and call $G'_{\nu_i}$ the graph resulting from $G_{\nu_i}$ after this subdivision. As in the previous case, denote by $G'_{\mu_l}$ and $G'_{\mu_r}$ the left and the right components of $G'_{\nu_i}$, where $G'_{\mu_r}$ coincides with $G_{\mu_r}$.  We show that $G'_{\nu_i}$ admits a rectilinear planar representation with spirality $\sigma'_\nu = \sigma^{\max}_{\nu_i}(0) + 1$. Since we have added a subdivision vertex on an exposed edge of $G_{\mu_l}$, there exists a rectilinear planar representation  $H'_{\mu_l}$ of $G'_{\mu_l}$ with spirality $\sigma'_{\mu_l}=M_l+1$. Also, since $\sigma_{\mu_r} < M_r$, there exists a rectilinear planar representation of $G'_{\mu_r}$ with spirality $\sigma'_{\mu_r}=\sigma_{\mu_r}+1$. Hence, $\sigma'_{\mu_l} - \sigma'_{\mu_r} = 2$ and therefore, by \cref{le:P-2-children-support-type1}, we can merge the representations $H'_{\mu_l}$ and $H'_{\mu_r}$ into a rectilinear planar representation $H'_{\nu_i}$ of $G'_{\nu_i}$. The spirality of $H'_{\nu_i}$ is $\sigma'_{\nu_i} = \sigma^{\max}_{\nu_i}(0) + 1$. By replacing the subdivision vertex of $H'_{\mu_l}$ with a bend, we get an orthogonal representation of $G_{\nu_i}$ with one bend and with spirality $\sigma^{\max}_{\nu_i}(0) + 1$, i.e., $\sigma^{\max}_{\nu_i}(1)=\sigma^{\max}_{\nu_i}(0) + 1$. Iterate this procedure until $\sigma'_{\mu_r}=M_r$, i.e., until $\sigma'_{\nu_i}=M_r+1$. Denote by $b$ the number of bends added in total. By \cref{le:spirality-P-node-2-children}, the spirality of the resulting orthogonal representation is $\sigma^{\max}_{\nu_i}(b)=M_l+b-1=M_r+1$, and hence $b = M_r+2-M_l$. If we consider the rectilinear representation $H'_{\nu_i}$ where the $b$ bends are replaced with degree-2 vertices, its left and right components have the maximum possible spirality in their representability intervals. Hence, the previous case applies, and inserting exactly one subdivision vertex to $G'_{\nu_i}$, does not result into a graph that admits a rectilinear planar representation with spirality greater than the one of $H'_{\nu_i}$.    
		It follows that $b^+_{\nu_i} = M_r+2-M_l$.
		
		\item $\sigma_{\mu_l} < M_l$ and $\sigma_{\mu_r} = M_r$. With a symmetric argument as in the previous case, $b^+_{\nu_i} = M_l - (M_r + 2)$. 
	\end{itemize}
\end{proof}

\begin{lemma}\label{le:P-2-without-exposed-edges-left}
	Let $\nu$ be a P-node with two children $\mu_l$ and $\mu_r$, such that $\mu_l$ has no exposed edge. Let $G_{\mu_l}$ and $G_{\mu_r}$ be rectilinear planar with representability intervals $I_{\mu_l}=[m_l,M_l]$ and $I_{\mu_r}=[m_r,M_r]$, respectively. If $G_\nu$  is not rectilinear planar then: (i) the budget for $\nu$ is $b_\nu=\delta([m_l-M_r, M_l-m_r],[3,3])$; and (ii) the interval of spirality values for an orthogonal representation of $G_\nu$ with $b_\nu$ bends is $I'_\nu = [m-b_\nu, M+\min\{b^+_{\mu_l},b_\nu\}]$ if $M_l - m_r < 3$ and  $I'_\nu = [m-\min\{b^-_{\mu_l},b_\nu\}, M+b_\nu]$ if $m_l - M_r > 3$.
\end{lemma}
\begin{proof}
	Since by hypothesis $G_\nu$ is not rectilinear planar, by \cref{ta:representability} we have $3 \notin [m_l-M_r, M_l-m_r] \cap \Delta_\nu  = \emptyset$, where $\Delta_\nu=[3,3]$. In other words, $3 \notin [m_l-M_r, M_l-m_r]$.
	
	\medskip\noindent\textsf{Proof of Property (i)}.
	We first prove that $b_\nu = \delta([m_l-M_r, M_l-m_r],[3,3])$ bends are necessary. Suppose for a contradiction that $G_\nu$ admits an orthogonal representation $H'_\nu$ with $b'_\nu < b_\nu$ bends. Let $b'_l$ and $b'_r$ be the number of bends in the restriction of $H'_\nu$ to $G_{\mu_l}$ and to $G_{\mu_r}$, respectively. By \cref{le:P-exposed-edges-support}, we can assume that all the bends $b'_r$ are along an exposed edge of $G_{\mu_r}$. 
	Consider the underlying graph $G'_\nu$ of $H'_\nu$ obtained by replacing each bend of $H'_\nu$ with a subdivision vertex. $G'_\nu$ is rectilinear planar. Let $G'_{\mu_l}$ and $G'_{\mu_r}$ be the left and right component of $G'_\nu$, respectively. 
	Hence $G'_{\mu_r}$ is rectilinear planar with representability interval $[m_r-b'_r,M_r+b'_r]$. About $G'_{\mu_l}$, its representability interval $I'_{\mu_l}=[p,q]$ is such that $[p,q] \subseteq [m_l-b_l',M_l+b_l']$. In particular $[p,q] = [m_l-b_l',M_l+b_l']$ when $b_{\mu_l}^- \geq b'_l$ and $b_{\mu_l}^+ \geq b'_l$.  
	By \cref{ta:representability}, the representability condition for $G'_\nu$ is $[p-M_r-b'_r, q-m_r+b'_r] \cap [3,3] \neq \emptyset$.
	Since $p \geq m_l-b_l'$ and $q \leq M_l+b_l'$ we have $[p-M_r-b'_r, q-m_r+b'_r] \subseteq [m_l-b_l'-M_r-b'_r,M_l+b_l'-m_r+b'_r]$, and therefore $[m_l-b_l'-M_r-b'_r,M_l+b_l'-m_r+b'_r] \cap [3,3] \neq \emptyset$.
	This equals to say that $[m_l-M_r-b'_\nu,M_l-m_r+b'_\nu] \cap [3,3] \neq \emptyset$, which implies that $\delta([m_l-M_r,M_l-m_r],[3,3]) \leq b'_\nu < b_\nu$, a contradiction.      
	
	We now prove that $b_\nu$ bends suffice. Insert $b_\nu$ subdivision vertices on any exposed edge of $G_{\mu_r}$, and let $G'_{\mu_r}$ be the resulting component. Since $G_{\mu_r}$ is rectilinear planar by hypothesis, $G'_{\mu_r}$ is also rectilinear planar and its representability interval is $[m_r-b_\nu,M_r+b_\nu]$. Consider the plane graph $G'_{\nu}$ obtained by the union of $G_{\mu_l}$ and $G'_{\mu_r}$. Since by hypothesis $b_\nu = \delta([m_l-M_r, M_l-m_r],[3,3])$, we have that  $[m_l-M_r-b_\nu, M_l- m_r+b_\nu] \cap [3,3] \neq \emptyset$. It follows that $G'_{\nu}$ is rectilinear planar. Consider any rectilinear representation $H'_{\nu}$ of $G'_{\nu}$ and replace each of its subdivision vertices with a bend. Since by construction $G'_{\nu}$ has $b_\nu$ subdivision vertices, the resulting orthogonal representation has at most $b_\nu$ bends.  
	
	\medskip\noindent\textsf{Proof of Property (ii)}. Since by hypothesis $3 \notin [m_l-M_r, M_l-m_r]$ we have two cases: either $M_l - m_r < 3$ or $M_r - m_l > 3$. For each of them, we must show that every orthogonal representation of $G_\nu$ with $b_\nu$ bends has spirality $\sigma_\nu$ in the interval $I'_\nu$ and that for every value $\sigma_\nu \in I'_\nu$ there exists an orthogonal representation of $G_\nu$ with $b_\nu$ bends and spirality~$\sigma_\nu$. We show the argument for the case~$M_l-m_r<3$ (the other case is treated analogously). 
	
	By \cref{ta:representability} we have $M=\min\{M_l-1, M_r+2\}$ and $m=\max\{m_l-1,m_r+2\}$. Since $M_l<m_r+3\le M_r+3$, we have $M_l-1<M_r+2$ and since $m_l-3\le M_l-3<m_r$, we have $m_l-1<m_r+2$.   
	Hence, $M=M_l-1$ and $m=m_r+2$.    
	Suppose first that $G_\nu$ has an orthogonal representation $H_\nu$ with $b_\nu$ bends. We prove that $\sigma_\nu\in [m-b_\nu,M+\min\{b_{\mu_l}^+,b_\nu\}]$. 
	Let $\sigma_\nu$ be the spirality of $H_\nu$, and let $\sigma_{\mu_l}$ and $\sigma_{\mu_r}$ be the spiralities of the restrictions $H_{\mu_l}$ and $H_{\mu_r}$ of $H_\nu$ to  $G_{\mu_l}$ and $G_{\mu_r}$, respectively. By \cref{le:P-2-children-support-type1} we have $\sigma_{\mu_l}-\sigma_{\mu_r}=3$. Let $b_l$ and $b_r$ be the number of bends in $H_{\mu_l}$ and $H_{\mu_r}$, respectively. Clearly, $b_l + b_r = b_\nu$. We now prove the following claim.
	
	\begin{clm} 
		$b_l\in [0,\min\{b_{\mu_l}^+,b_\nu\}]$.
	\end{clm}
	\begin{claimproof}
		Suppose for a contradiction that $b_l\not \in [0,\min\{b_{\mu_l}^+,b_\nu\}]$. Clearly, $b_l \in [0,b_\nu]$, and thus, under this hypothesis, it must be
		$\min\{b_{\mu_l}^+,b_\nu\}=b_{\mu_l}^+$ and $b_l\in [b_{\mu_l}^++1,b_\nu]$.
		We show that there exists an orthogonal representation $H_\nu'$ of $G_\nu$ with less than $b_\nu$ bends, which is impossible by definition of budget $b_\nu$.  
		We will denote by $H'_{\mu_l}$ and $H'_{\mu_r}$ the restrictions of $H_\nu'$ to $G_{\mu_l}$ and to $G_{\mu_r}$, respectively. Also,  $b_l'$ and $b_r'$ will denote the number of bends of $H'_{\mu_l}$ and $H'_{\mu_r}$, while $\sigma_{\mu_l}'$ and $\sigma_{\mu_r}'$ will denote the spiralities of $H'_{\mu_l}$ and $H'_{\mu_r}$, respectively.
		We analyze two cases: $b_l=b_{\mu_l}^++1$ and $b_l\in [b_{\mu_l}^++2, b_\nu]$.
		
		\begin{itemize}   
			\item Suppose $b_l=b_{\mu_l}^++1$. We set $b_l'=b_l-1=b_{\mu_l}^+$ and $b_r'=b_r$. By definition of positive flexibility breakpoint it is possible to set $\sigma_{\mu_l}'=\sigma_{\mu_l}$. Also, we set $\sigma_{\mu_r}'=\sigma_{\mu_r}$. Since $\sigma_{\mu_l}'-\sigma_{\mu_r}'=\sigma_{\mu_l}-\sigma_{\mu_r}=3$, by \cref{le:P-2-children-support-type1}, $H_\nu'$ exists and it has $b_l'+b_r'=b_\nu-1<b_\nu$ bends. 
			\item Suppose $b_l\in [b_{\mu_l}^++2, b_\nu]$. We set $b_l'=b_l-2$ and $b_r'=b_r+1$. By definition of positive flexibility breakpoint in this case it is possible to set $\sigma_{\mu_l}'=\sigma_{\mu_l}-1$. Also, since $G_{\mu_r}$ has an exposed edge, by \cref{le:P-exposed-edges-support} it is possible to use it to place the extra bend of the right component. It follows that we can set $\sigma_{\mu_r}'=\sigma_{\mu_r}-1$. Therefore, $\sigma_{\mu_l}'-\sigma_{\mu_r}'=\sigma_{\mu_l}-\sigma_{\mu_r}=3$, and by \cref{le:P-2-children-support-type1}, $H_\nu'$ exists and it has $b_l'+b_r'<b_l+b_r=b_\nu$ bends. 
		\end{itemize}
		
		\noindent Hence, in both cases $H_\nu'$ has less bends than $H_\nu$, a contradiction.
	\end{claimproof}
	
	Since $G_\nu$ is of type \Pin{3ll}, by \cref{le:spirality-P-node-2-children} and \cref{ta:p-coefficients} we have $\sigma_\nu = \sigma_{\mu_l} -1=\sigma_{\mu_r} +2$. By the claim we have $b_l\in [0,\min\{b_{\mu_l}^+,b_\nu\}]$, hence $m_l-\min\{b_{\mu_l}^+,b_\nu\}\le \sigma_{\mu_l}\le M_l+\min\{b_{\mu_l}^+,b_\nu\}$ and hence $\sigma_\nu \in [m_l-\min\{b_{\mu_l}^+,b_\nu\}-1, M_l+\min\{b_{\mu_l}^+,b_\nu\}-1]$. Also, $b_r\in [0,b_\nu]$ and $\sigma_\nu\in [m_r-b_\nu+2,M_r+b_\nu+2]$. 
	Hence $\sigma_\nu\ge \max\{m_l-\min\{b_{\mu_l}^+,b_\nu\}-1,m_r-b_\nu+2\}$. By Property~(i), $b_\nu = \delta([m_l-M_r, M_l-m_r],[3,3]) = 3-M_l+m_r$, hence, $m_r-b_\nu+2=m_r-3+M_l-m_r+2=M_l-1\ge m_l-\min\{b_{\mu_l}^+,b_\nu\}-1$. It follows that  $\sigma_\nu\ge m_r-b_\nu+2$. Also, $\sigma_\nu\le \min\{M_l+\min\{b_{\mu_l}^+,b_\nu\}-1, M_r+b_\nu+2\}$. Since by hypothesis $M_l-1\le m_r+2\le M_r+2$ and since $\min\{b_{\mu_l}^+,b_\nu\}\le b_\nu$, we have $M_l+\min\{b_{\mu_l}^+,b_\nu\}-1\le  M_r+b_\nu+2$ and $\sigma_\nu\le M_l+\min\{b_{\mu_l}^+,b_\nu\}-1$. Hence, $\sigma_\nu \in  [m_r-b_\nu+2, M_l+\min\{b_{\mu_l}^+,b_\nu\}-1] = [m-b_\nu,M+\min\{b_{\mu_l}^+,b_\nu\}]=I'_\nu$.
	
	
	\medskip
	It remains to show that for every $\sigma_\nu\in I'_\nu = [m-b_\nu, M+\min\{b^+_{\mu_l},b_\nu\}]$, there exists an orthogonal representation $H_\nu$ of $G_\nu$ with $b_\nu$ bends and with spirality $\sigma_\nu$. With the same notation and approach as in \cref{le:P-2-exposed-edges}, we show how to set the number of bends $b_l$ and $b_r$ that can be assigned to the left and right component to guarantee the existence of~$H_\nu$. By Property~(i), $b_\nu = \delta([m_l-M_r, M_l-m_r],[3,3]) = 3-M_l+m_r$.
	We set $b_l = \sigma_\nu - M_l + 1$ and $b_r = b_\nu - b_l$. We first show that $b_l \in [0,\min\{b_{\mu_l}^+,b_\nu\}]$. We have $b_l\in  [m-b_\nu - M_l + 1, M+\min\{b_{\mu_l}^+,b_\nu\} - M_l + 1]$. 
	Hence, $b_l\in  [m_r+2-b_\nu - M_l + 1, M_l-1+\min\{b_{\mu_l}^+,b_\nu\} - M_l + 1]=[0,\min\{b_{\mu_l}^+,b_\nu\}]$. 

	Let $\nu_1, \dots, \nu_h$ be the children of $\mu_l$. Since $G_{\mu_l}$ has no exposed edge, each $G_{\nu_i}$ $(i=1, \dots, h)$ is a parallel component of type \Pio{2}{22} and \cref{le:flexibility_breakpoints} applies. We distribute the $b_l$ bends on the components $G_{\nu_i}$ in such a way that the number of bends assigned to $G_{\nu_i}$ is at most $b_{\nu_i}^+$ $(i=1, \dots, h)$. This is always possible, because $b_l \in [0,\min\{b_{\mu_l}^+,b_\nu\}]$. The bends assigned to each $G_{\nu_i}$ are inserted on an exposed edge according to the procedure in the proof of \cref{le:flexibility_breakpoints}.    
	The graph $G_{\mu_l}'$ obtained from $G_{\mu_l}$ by regarding each bend as a subdivision vertex admits a rectilinear planar representation $H'_{\mu_l}$ with spirality $\sigma_{\mu_l}'=M_l+b_l=M_l+\sigma_\nu - M_l + 1=\sigma_\nu+1$.  
	We now place $b_r$ bends along an exposed edge of $G_{\mu_r}$. The graph $G_{\mu_r}'$ obtained from $G_{\mu_r}$ by regarding each bend as a subdivision vertex admits a rectilinear planar representation $H_{\mu_r}'$ with spirality $\sigma_{\mu_r}'=m_r-b_r=m_r-(b_\nu-b_l)=m_r-b_\nu +\sigma_\nu - M_l + 1=(m_r-M_l+1)-b_\nu+\sigma_\nu=b_\nu-2-b_\nu+\sigma_\nu=\sigma_\nu-2$. 
	It follows that $\sigma_{\mu_l}'-\sigma_{\mu_r}'=\sigma_\nu+1-(\sigma_\nu-2)=3$ and, by \cref{le:P-2-children-support-type3}, $G_\nu'$ is rectilinear planar and it admits a rectilinear planar representation $H'_\nu$ obtained by combining in parallel $H_{\mu_l}'$ and $H_{\mu_r}'$. By \cref{le:spirality-P-node-2-children}, $H'_\nu$ has spirality $\sigma_\nu'=\sigma_{\mu_l}' - {k_u^l}' {\alpha_{u}^l}' -  {k_u^l}'{\alpha_v^l}'$. Since $u$ and $v$ have degree four, we have ${\alpha^l_u}'={\alpha^l_v}'={\alpha^r_u}'={\alpha^r_v}'=1$. Also, by \cref{ta:p-coefficients}, we have ${k_u^l}'={k_u^r}'=\frac{1}{2}$. Hence,    
	$\sigma'_\nu=\sigma_\nu + 1 - 1=\sigma_\nu$. By replacing the subdivision vertices of $H'_\nu$ with bends we get an orthogonal representation of $G_\nu$ with $b_\nu$ bends and spirality~$\sigma_\nu$. 
\end{proof}

The proof of \cref{le:P-2-without-exposed-edges-right} is symmetric to that of \cref{le:P-2-without-exposed-edges-left} and is omitted.

\begin{lemma}\label{le:P-2-without-exposed-edges-right}
	Let $\nu$ be a P-node with two children $\mu_l$ and $\mu_r$, such that $\mu_r$ has no exposed edge. Let $G_{\mu_l}$ and $G_{\mu_r}$ be rectilinear planar with representability intervals $I_{\mu_l}=[m_l,M_l]$ and $I_{\mu_r}=[m_r,M_r]$, respectively. If $G_\nu$  is not rectilinear planar, then: (i) the budget for $\nu$ is $b_\nu=\delta([m_l-M_r, M_l-m_r],[3,3])$; and (ii) the interval of spirality values for an orthogonal representation of $G_\nu$ with $b_\nu$ bends is $I'_\nu = [m-\min\{b^-_{\mu_r},b_\nu\},M+b_\nu]$ if $m_l - M_r > 3$ and $I'_\nu = [m-b_\nu,M+\min\{b^+_{\mu_r},b_\nu\}]$ if $M_l - m_r < 3$.  
\end{lemma}


\subsection{Budget of the root}\label{sse:bottom-up-root}
We finally show how to compute the budget of the root $\rho$ of $T$. Recall that $\rho$ is the P$^r$-node describing the parallel composition of the reference edge $e=(u,v)$ with the rest of the graph. If $e$ is a dummy edge, the budget of $\rho$ is zero, because $e$ does not need to be drawn. Thus we assume that the graph is biconnected and $e$ is a real edge. Let $\eta$ be the child of $\rho$ that does not correspond to $e$, and let $u'$ and $v'$ be the alias vertices associated with the poles $u$ and $v$ of $G_\eta$. If $G_\eta$ is rectilinear planar with representability interval $I_\eta$, by the root condition in \cref{ta:representability} we know that $G$ is rectilinear planar if and only if $I_\eta \cap \Delta_\rho \neq \emptyset$.
Recall that $\Delta_\rho$ is defined as follows: $(i)$ $\Delta_\rho=[2,6]$ if $u'$ coincides with $u$ and $v'$ coincides with $v$; $(ii)$ $\Delta_\rho=[3,5]$ if exactly one of $u'$ and $v'$ coincides with $u$ and $v$, respectively; $(iii)$ $\Delta_\rho=4$ if none of $u'$ and $v'$ coincides with $u$ and $v$. We prove the following.

\begin{lemma}\label{le:root-budget}
	Let $e=(u,v)$ be the reference edge of $G$ and let $\rho$ be the root of $T$ with respect to $e$. Let $\eta$ be the child of $\rho$ that does not correspond to $e$. Suppose that $G_\eta$ is rectilinear planar with representability interval $I_\eta$. If $G$ is not rectilinear planar then $b_\rho=\delta(I_\eta,\Delta_\rho)$.
\end{lemma}	
\begin{proof}
	See \cref{fi:root-condition} for an illustration of the statement.
	Let $f_{\mathrm{int}}$ be the internal face of $G$ incident to $e$.
	Note that $H$ is an orthogonal representation of $G$ if and only if the following two conditions hold: The restriction $H_\eta$ of $H$ to $G_\eta$ is an orthogonal representation; the number $A$ of right turns minus left turns of any simple cycle of $G$ in $H$ containing $e$ and traversed clockwise in $H$ is equal to $4$. 
	We have $A=\sigma_\eta-\sigma_e+\alpha_{u'}+\alpha_{v'}$, where: $\sigma_\eta$ is the spirality of $H_\eta$; $\sigma_e$ is the spirality of $e$; for $w\in \{u',v'\}$, $\alpha_w=1$, $\alpha_w=0$, and $\alpha_w=-1$ if the angle formed by $w$ in $f_{\mathrm{int}}$ equals $90^o$, $180^o$, or $270^o$, respectively. 
	
	Since $G$ is not rectilinear planar, $H$ must contain some bends. Denote by $b$ the number of bends of $H$. Each of these $b$ bends placed along an edge of $G_\eta$ contributes to increase or decrease $\sigma_\eta$ by at most one unit, therefore increasing or decreasing $A$ by at most one. Also placing this bend along $e$ contributes to increase or decrease $A$ by at most one. Hence, without loss of generality, we can assume that all the $b$ bends are placed along $e$, which implies that $H_\eta$ does not contain any bend. Hence, we have $\sigma_\eta \in I_\eta = [m_\eta,M_\eta]$ and $\sigma_e\in [-b,b]$. We now show that if $b$ is the minimum value such that $A=4$ then $b=b_\rho$. Observe that when $b$ is minimum, we have that $\lvert \sigma_e \rvert=b$, and therefore $\lvert \sigma_e \rvert>0$.  
	
	Consider first Case $(i)$ in the definition of $\Delta_\rho$, i.e., $\Delta_\rho=[2,6]$. Since in this case the alias vertices coincide with the poles, we have $\alpha_{u'} \in [-1,1]$, $\alpha_{v'} \in [-1,1]$, and hence $\alpha_{u'}+\alpha_{v'} \in [-2,2]$.  
	Since $\sigma_\eta+\alpha_{u'}+\alpha_{v'}+\sigma_e=4$ there are two possible subcases for any possible choice of tha values $\sigma_\eta$, $\alpha_{u'}$, and $\alpha_{v'}$: either (a)~$\sigma_\eta+\alpha_{u'}+\alpha_{v'}<4$ or (b)~$\sigma_\eta+\alpha_{u'}+\alpha_{v'}>4$.
	
	\begin{itemize}
		\item Case~(a). The maximum value for $\sigma_\eta+\alpha_{u'}+\alpha_{v'}$ is $M_\eta+2$. Hence $M_\eta+2<4$, which implies $M_\eta<2$. It follows that $b_\rho=\delta(I_\eta,\Delta_\rho)=2-M_\eta$. Also, since $b$ is the minimum value such that $A=4$ and $\sigma_\eta+\alpha_{u'}+\alpha_{v'}<4$, we have that $b=\sigma_e$ and $M_\eta+2+b=4$, which implies $b = 2-M_\eta = b_\rho$. Therefore, an orthogonal representation of $G$ with $b$ bends is constructed by placing $b$ bends along $e$ and choosing $\sigma_\eta = M_\eta$. 
		
		\item Case~(b). The minimum value for  $\sigma_\eta+\alpha_{u'}+\alpha_{v'}$ is $m_\eta-2$. Hence $m_\eta-2>4$, which implies $m_\eta>6$. It follows that $b_\rho=\delta(I_\eta,\Delta_\eta)=m_\eta-6$. Also, since $b$ is the minimum value such that $A=4$ and $\sigma_\eta+\alpha_{u'}+\alpha_{v'}>4$, we have that $b=-\sigma_e$ and $m_\eta-2-b=4$, which implies $b = m_\eta - 6 = b_\rho$. Therefore, an orthogonal representation of $G$ with $b$ bends is constructed by placing $b$ bends along $e$ and choosing $\sigma_\eta = m_\eta$.
	\end{itemize}  
	
	The proofs for Case~$(ii)$ ($\Delta_\rho=[3,5]$) and Case~$(iii)$ ($\Delta_\rho=4$) are analogous, observing that in Case~$(ii)$ we have $\alpha_{u'}+\alpha_{v'}\in [-1,1]$ and in Case~$(iii)$ we have $\alpha_{u'}+\alpha_{v'}=0$.
\end{proof}


\begin{figure}[t]
	\captionsetup[subfigure]{labelformat=empty}
	\centering
	\begin{subfigure}{.2\columnwidth}
		\centering
		\includegraphics[width=\columnwidth,page=3]{root-condition.pdf}
		\subcaption{\centering~Case $(i)$}
		\label{fi:root-condition-a}
	\end{subfigure}
	\hfil
	\begin{subfigure}{.2\columnwidth}
		\centering
		\includegraphics[width=\columnwidth,page=2]{root-condition.pdf}
		\subcaption{\centering~Case $(ii)$}
		\label{fi:root-condition-b}
	\end{subfigure}
	\hfil
	\begin{subfigure}{.2\columnwidth}
		\centering
		\includegraphics[width=\columnwidth,page=1]{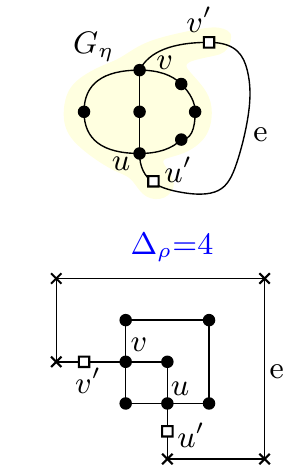}
		\subcaption{\centering~Case $(iii)$}
		\label{fi:root-condition-c}
	\end{subfigure}
	\caption{Illustration of Lemma~\ref{le:root-budget}.  We have $I_\eta=[-1,-1]$ and: In Case~(i), $\Delta_\rho=[2,6]$ and $b_\rho=\delta(I_\eta,\Delta_\rho)=3$; in Case~(ii), $\Delta_\rho=[3,5]$ and $b_\rho=\delta(I_\eta,\Delta_\rho)=4$; in Case~(iii), $\Delta_\rho=4$ and $b_\rho=\delta(I_\eta,\Delta_\rho)=5$.
		%
		%
	}\label{fi:root-condition}
\end{figure}

\cref{ta:minbend-representability} summarizes how to compute $b_\nu$ and $I'_\nu$ for the different types of nodes $\nu$. 

\renewcommand{\arraystretch}{1.5}
\begin{table}[h]
	\centering
	\caption{Summary of how to compute $b_\nu$ and $I'_\nu$ for the different types of nodes $\nu$. In the formulas, we have $\gamma=\lambda +\beta-2$ and $\phi(\cdot)$ is such that $\phi(r)=1$ and $\phi(l)=0$. Values $m$ and $M$ are computed as shown in \cref{ta:representability}.}\label{ta:minbend-representability}
	\begin{tabular}{|l |c |}
		
		\hline
		\rowcolor{antiquewhite}\multicolumn{2}{|c|}{\bf P-node with two children, each having an exposed edge --  \cref{le:P-2-exposed-edges}}\\
		\hline
		{Budget} & $\delta([m_l-M_r, M_l-m_r],\Delta_\nu)$
		\\ 
		{Representability Interval} & $[m-b_\nu, M+b_\nu]$ \\

		\hline
		\rowcolor{antiquewhite}\multicolumn{2}{|c|}{\bf P-node with three children (each child has an exposed edge) --  \cref{le:P-3-exposed-edges}}\\
		\hline
		{Budget} & $\overline{m}_z-M_x$ \\
		{Representability Interval} & $[\max\{\overline{M}_x,\overline{m}_y\},\min\{\overline{m}_z,\overline{M}_y\}]$\\


		\hline
		\rowcolor{antiquewhite}\multicolumn{2}{|c|}{\bf P-node with two children, such that $\mu_d$, with $d\in \{l,r\}$ has no exposed edge}\\
		\rowcolor{antiquewhite}\multicolumn{2}{|c|}{\bf --  \cref{le:P-2-without-exposed-edges-left} if $d=l$ and \cref{le:P-2-without-exposed-edges-right} if $d=r$}\\ 
		\hline
		{Budget} & $b_\nu=\delta([m_l-M_r, M_l-m_r],[3,3])$\\ 
		{Representability Interval} &  If $d=l$:\\& $[m-b_\nu, M+\min\{b^+_{\mu_l},b_\nu\}]$ if $M_l - m_r < 3$ \\ &  $[m-\min\{b^-_{\mu_l},b_\nu\}, M+b_\nu]$ if $m_l - M_r > 3$.
		\\ & If $d=r$:\\ & $[m-\min\{b^-_{\mu_r},b_\nu\},M+b_\nu]$ if $m_l - M_r > 3$ \\ & $[m-b_\nu,M+\min\{b^+_{\mu_r},b_\nu\}]$ if $M_l - m_r < 3$. \\
		
		\hline
		\rowcolor{antiquewhite}\multicolumn{2}{|c|}{\bf P$^r$-node (the root $\rho$)}\\
		\hline
		{Budget} & $b_\rho=\delta(I_\eta,\Delta_\rho)$\\
		\hline
	\end{tabular}
\end{table}

\bigskip

\bigskip
\subsection{Optimality of the Approach}\label{sse:optimality} 
Our bottom-up algorithm incrementally computes for each node $\nu$ of $T$ the budget of bends needed to realize an orthogonal representation of $G_\nu$. We prove that, the total budget at the level of the root of $T$ corresponds to the number of bends of a bend-minimum orthogonal representation of $G$. More formally, for a node $\nu$ of $T$, the \emph{cumulative budget} $B_\nu$ of $\nu$ is the sum of the budgets of all nodes in the subtree of~$T$ rooted at $\nu$. If $\nu$ is a leaf of~$T$, $B_\nu = b_\nu = 0$. If $\nu$ is an internal node with children $\mu_1, \mu_2, \dots, \mu_h$ then $B_\nu = b_\nu + \sum_{i=1, \dots, h}B_{\mu_i}$. Hence, the cumulative budget $B_\rho$ of the root corresponds to the total number of bends at the end of the bottom-up visit. 

\begin{theor}\label{th:bottom-up-cottectness}
	Let $G$ be a plane series-parallel 4-graph, $T$ be an SPQ$^*$-tree of $G$, and $\rho$ be the root of $T$. The cumulative budget $B_{\rho}$ computed by the bottom-up visit equals the number of bends of a bend-minimum orthogonal representation~of~$G$.
\end{theor}
\begin{proof}
	We prove that once the algorithm has processed a node $\nu$ in the bottom-up visit, the cumulative budget $B_\nu$ equals the number of bends of a bend-minimum orthogonal representation of~$G_\nu$.
	This implies that once the algorithm has visited the root $\rho$ of~$T$, the cumulative budget $B_\rho$ equals the number of bends of a bend-minimum orthogonal representation of~$G_\rho = G$.
	The proof is by induction on the depth of the subtree of~$T$~rooted~at~$\nu$.
	
	\medskip \noindent \textsf{Base Case.} For a leaf $\nu$ of $T$ ($\nu$ is a Q$^*$-node), the statement is trivial, as $B_\nu=0$.
	
	\medskip \noindent \textsf{Inductive Case.} Let $\nu$ be a node of $T$ that is not a leaf. Denote by $\mu_1, \mu_2, \dots, \mu_h$ the children of $\nu$. By inductive hypothesis, each $B_{\mu_i}$ $(i = 1, \dots, h)$ corresponds to the number of bends of a bend-minimum representation of~$G_{\mu_i}$. By definition, $B_\nu = b_\nu + \sum_{i=1, \dots, h}B_{\mu_i}$, where $b_\nu$ is the minimum number of bends that must be used in addition to $\sum_{i=1, \dots, h}B_{\mu_i}$ to realize an orthogonal representation of $G_\nu$ ($b_\nu = 0$ if $\nu$ is an S-node). Budget $B_\nu$ corresponds to the minimum number of bends of any orthogonal representation $H_\nu$ of $G_\nu$ with the property that the number of bends $b(H_{\mu_i})$ of the restriction $H_{\mu_i}$ of $H_\nu$ to $G_{\mu_i}$ is such that $b(H_{\mu_i}) \geq B_{\mu_i}$ ($i \in \{1, \dots, h\}$).  Let $H'_\nu$ be any bend-minimum orthogonal representation of $G_\nu$, and let $H'_{\mu_i}$ be the restriction of $H'_\nu$ to $G_{\mu_i}$. Since, by inductive hypothesis, there is no orthogonal representation of  $G_{\mu_i}$ with less bends than $B_{\mu_i}$, we have $b(H'_{\mu_i}) \geq B_{\mu_i}$, and hence, by the previous observation, $b(H'_\nu) \geq B_\nu$. On the other hand, since $H'_\nu$ is bend-minimum for $G_\nu$, we also have $b(H'_\nu) \leq B_\nu$, i.e., $b(H'_\nu) = B_\nu$.
\end{proof}


\section{Bend-Minimization in Linear Time}\label{se:summary}
The bottom-up visit described above, equips each node $\nu$ of $T$ with three information: the budget $b_\nu$, the cumulative budget $B_\nu$, and, if $\nu  \neq \rho$, an interval $I'_\nu=[m'_\nu,M'_\nu]$ of all possible spirality values that an orthogonal representation of $G_\nu$ with $B_\nu$ bends can have. Once the bottom-up visit of $T$ has been completed, our algorithm performs a top-down visit of $T$ to suitably add to $G$ a number $B_\rho$ of subdivision vertices, so that the resulting graph $G'$ admits a rectilinear planar representation.

At the beginning of the top-down visit, the root $\rho$ is considered. Let $\eta$ be the child of $\rho$ that does not correspond to $e$. If $I'_\eta \cap \Delta_\rho \neq \emptyset$, the algorithm selects an arbitrary value $\sigma_\eta \in I'_\eta \cap \Delta_\rho$ as target spirality value for a representation of $G'_\eta$ within a rectilinear planar representation of $G'$. If, vice versa, $I'_\eta \cap \Delta_\rho = \emptyset$, according to the proof of Lemma~\ref{le:root-budget}, the algorithm subdivides $e$ with $b_e=\delta(I'_\eta,\Delta_\rho)$ bends and sets the target spirality value $\sigma_\eta$ as either $\sigma_\eta=M'_\eta$ (if $M'_\eta$ is smaller than the infimum of $\Delta_\rho$) or $\sigma_\eta=m'_\eta$ (if $m'_\eta$ is larger than the supremum of $\Delta_\rho$).
In the next step of the top-down visit, the algorithm considers node $\eta$, for which the target spirality value $\sigma_\eta$ has been previously fixed. If $\eta$ is an S-node then $b_\eta=0$, i.e., no subdivision vertices must be added in this step. If $\eta$ is a P-node and $b_\eta > 0$, the algorithm suitably adds $b_\eta$ subdivision vertices along some edges of $G_\eta$. To do so, it applies the procedures described in the second part of the proof of Property~(ii) of Lemma~\ref{le:P-3-exposed-edges}, of Lemma~\ref{le:P-2-exposed-edges}, or of Lemmas~\ref{le:P-2-without-exposed-edges-left} and~\ref{le:P-2-without-exposed-edges-right}, depending on whether $\eta$ is a P-node with three children, a P-node with two children both having an exposed edge, or a P-node with two children one of which has no exposed edge. 
Then, the algorithm sets the spirality values for each child of $\eta$. 
Namely, we distinguish the following cases:

\smallskip\noindent{\bf Case~1:} $\eta$ is an S-node, with children $\mu_1, \dots, \mu_h$ $(i \in \{1, \dots, h\})$. Let $I'_{\mu_i}=[m'_i,M'_i]$ be the representability interval of $\mu_i$ (assuming that any orthogonal representation of $G_{\mu_i}$ will have $B_{\nu_i}$ bends). We must find a value $\sigma_{\mu_i} \in [m'_i,M'_i]$ for each $i = 1, \dots, h$ such that $\sum_{i=i}^h \sigma_{\mu_i} = \sigma_\eta$. To this aim, initially set $\sigma_{\mu_i} = M'_i$ for each $i = 1, \dots, h$ and consider $s = (\sum_{i=i}^h \sigma_{\mu_i}) - \sigma_\eta$. By Lemma~\ref{le:spirality-S-node}, $s \geq 0$. If $s = 0$ we are done. Otherwise, iterate over all $i=1, \dots, h$ and for each $i$
decrease both $\sigma_{\mu_i}$ and $s$ by the value $\min\{s, M'_i - m'_i\}$, until~$s=0$.

\smallskip\noindent{\bf Case~2:} $\eta$ is a P-node with three children, $\mu_l$, $\mu_c$, and $\mu_r$. By Lemma~\ref{le:spirality-P-node-3-children}, it suffices to set $\sigma_{\mu_l}=\sigma_{\eta}+2$, $\sigma_{\mu_c}=\sigma_{\eta}$, and $\sigma_{\mu_r}=\sigma_{\eta}-2$.

\smallskip\noindent{\bf Case~3:} $\eta$ is a P-node with two children, $\mu_l$ and $\mu_r$. Let $u$ and $v$ be the poles of $\eta$. By Lemma~\ref{le:spirality-P-node-2-children}, $\sigma_{\mu_l}$ and  $\sigma_{\mu_r}$ must be determined in such a way that $\sigma_{\mu_l} = \sigma_{\eta}+k_u^l\alpha_u^l+k_v^l\alpha_v^l$ and $\sigma_{\mu_r} = \sigma_{\eta}-k_u^r\alpha_u^r-k_v^r\alpha_v^r$. The values of $k_u^l$, $k_v^l$, $k_u^r$, $k_v^r$ are fixed by the indegree and outdegree of $u$ and $v$. Hence, it suffices to choose the values of $\alpha_u^l$, $\alpha_v^l$, $\alpha_u^r$, $\alpha_v^r$ such that they are consistent with the type of $\eta$ and yield $\sigma_{\mu_l} \in I'_{\mu_l}$ and $\sigma_{\mu_r} \in I'_{\mu_r}$. Since each $\alpha_w^d$ $(w \in \{u,v\}, d \in \{l,r\})$ is either $0$ or $1$, there are at most four combinations of values~to~consider.

\smallskip\noindent{\bf Case~4:} $\eta$ is a Q$^*$-node. In this case the algorithm does nothing, as no further subdivision vertices must be added.  

\smallskip \noindent
In the subsequent steps of the top-down visit, for every node $\nu$ the algorithm applies the same procedure as for $\eta$ to determine a target spirality value $\sigma_\nu$ and to suitably distribute the $b_\nu$ subdivision vertices along the edges of~$G_\nu$.

\begin{theor}\label{th:bend-minimum}
	Let $G$ be an $n$-vertex plane series-parallel 4-graph. There exists an $O(n)$-time algorithm that computes a bend-minimum orthogonal representation of $G$.
\end{theor}
\begin{proof}
	If $G$ is biconnected let $e$ be any edge of $G$ on the external face; otherwise, let $e$ be a dummy edge added on the external face to make $G$ biconnected. Let $T$ be an SPQ$^*$-tree of $G$ with respect to $e$. The algorithm executes the bottom-up and the top-down visits described above. 
	Once the top-down visit is completed and $B_\rho$ subdivision vertices have been suitably inserted in $G$, a rectilinear planar representation of the subdivision of $G$ is easily computed from the values of spiralities of each component and from the values of the angles at the poles of each component. From this representation we obtain a bend-minimum orthogonal representation of $G$ by replacing the subdivision vertices with bends. Since the obtained orthogonal representation has $B_\rho$ bends, by Theorem~\ref{th:bottom-up-cottectness} it has the minimum number of bends.
	
	We now analyze the time complexity of the algorithm. $T$ can be computed in $O(n)$ time and it consists of $O(n)$ nodes~\cite{DBLP:books/ph/BattistaETT99}. For a node $\nu$ of $T$ that is not a Q$^*$-node, we denote by $n_\nu$ the number of children of $\nu$.
	
	Consider first the bottom-up visit. Let $\nu$ be a visited node of $T$. If $\nu$ is a Q$^*$-node then $b_\nu=0$ and, by Table~\ref{ta:representability}, $I'_\nu = I_\nu$ is computed in $O(1)$ time (we can assume that the length $\ell$ of the chain of edges represented by $\nu$ is stored at $\nu$ during the construction of $T$). If $\nu$ is an S-node, we still have $b_\nu=0$ and by Table~\ref{ta:representability} $I'_\nu = I_\nu$ is computed in $O(n_\nu)$ time. If $\nu$ is a P-node with three children, $b_\nu$ and $I'_\nu$ are computed in $O(1)$ time by Lemma~\ref{le:P-3-exposed-edges}. If $\nu$ is a P-node with two children each having an exposed edge, $b_\nu$ and $I'_\nu$ are computed in $O(1)$ time by Lemma~\ref{le:P-2-exposed-edges}. Suppose now that $\nu$ is a P-node with an S-node child that has no exposed edge, and assume that this S-node is the left child $\mu_l$ of $\nu$. By Lemma~\ref{le:P-2-without-exposed-edges-left}, $b_\nu$ and $I'_\nu$ can be computed in $O(1)$ time if we know the flexibility breakpoints (positive and negative) of $\mu_l$. By Lemma~\ref{le:flexibility_breakpoints}, the flexibility breakpoints of $\mu_l$ can be computed in $O(n_{\mu_l})$ time. By Lemma~\ref{le:P-2-without-exposed-edges-right},
	the same reasoning applies if the right child of $\nu$ is an S-node with no exposed edge.
	Finally, if $\nu$ coincides with the root $\rho$ of $T$, $b_\nu$ is easily computed in $O(1)$ time by Lemma~\ref{le:root-budget}. In summary, for each $S$-node $\nu$ of $T$, the bottom-up visit requires $O(n_\nu)$ time to compute $b_\nu$, $I'_\nu$, and the flexibility breakpoints of $\nu$ if needed. For any other type of node, the visit takes $O(1)$ time. Thus, the bottom-up visit requires $O(n)$ overall time.   
	
	In the top-down visit, for every non-leaf node $\nu$ of $T$, the algorithm spends $O(b_\nu)$ time to add $b_\nu$ subdivision vertices. Also, we should consider the extra time $t$ required to decide what are the edges along which these bends must be added and what is the target spirality value for each child of $\nu$. Namely, if $\nu$ is the root, $t=O(1)$. If $\nu$ is 
	an S-node, $t=O(n_\nu)$ by Case~1 of the top-down visit above described. If $\nu$ is a P-node with three children, by Case~2 of the top-down visit and by Lemma~\ref{le:P-3-exposed-edges}, $t=O(1)$. If $\nu$ is a P-node with two children each having an exposed edge, by Case~3 of the top-down visit and by Lemma~\ref{le:P-2-exposed-edges}, $t=O(1)$. Finally, if $\nu$ is a P-node with two children, one of which is an S-node $\mu$ with no exposed edge, by Case~3 and by Lemmas~\ref{le:P-2-without-exposed-edges-left} and~\ref{le:P-2-without-exposed-edges-right}, $t=O(n_\mu)$. Hence, since $B_\rho=\sum_\nu b_\nu = O(n)$~\cite{DBLP:journals/siamcomp/Tamassia87}, the top-down visit takes $O(n)$ overall time, and a rectilinear planar representation of the subdivision of $G$ is easily computed in $O(n)$ time from the values of spiralities of each component and from the values of the angles at the poles of each component.
\end{proof}

\section{Conclusions and Open Problems}\label{se:open}
We proved that there exists an optimal linear-time algorithm that computes a bend-minimum orthogonal drawing of a plane series-parallel 4-graph; this result solves, for a popular and widely studied family of plane graphs, a question opened for over 30 years, thus shedding new light on the complexity of computing orthogonal drawings of plane graphs with the minimum number of bends. It is also worth remarking that, despite the sophisticated analysis and key ingredients needed to prove our main theorem, the resulting bend-minimization algorithm is relatively easy to implement, as at every node of the SPQ$^*$-tree it just requires to apply some simple formulas, as summarized in \cref{ta:representability,ta:minbend-representability}. We conclude by suggesting two open problems:
%
%
%
\medskip
\begin{itemize}
	\item[--] \textbf{Problem 1.} Our result holds for the series-parallel graphs that are also called two-terminal, which are either biconnected or that can be made biconnected with the addition of a single edge. Can we extend Theorem~\ref{th:bend-minimum} to 1-connected plane 4-graphs whose biconnected components are two-terminal series-parallel~graphs, also known as partial 2-trees? Here the main difficulty of extending our approach is to succinctly describe the possible values of spiralities for those components that contain cut-vertices. In fact, a cut-vertex requires the imposition of some constraints at its angles, in order to correctly compose the different biconnected components that contain it. These constraints forbid some vertex angles, and in turns some spirality values that an orthogonal representation can take.
	\item[--] \textbf{Problem 2.} Is it possible to find a linear-time algorithm for the bend-minimization problem of triconnected plane 4-graphs? A positive answer to this question, together with our result, can be used to solve the problem of computing a bend-minimum orthogonal drawing of a general plane 4-graph in linear time.
\end{itemize}

\medskip


\bibliography{bibliography}
\bibliographystyle{plain}

\end{document}